\documentclass[reqno,10pt]{amsart}
\usepackage{hyperref}
\usepackage{amssymb,amsmath,amsthm,}
\usepackage{a4wide}
\usepackage{amscd}
\usepackage{amsfonts}
\usepackage{amssymb}
\usepackage{latexsym}
\usepackage{esint}
\usepackage{mathrsfs}

\newtheorem{theorem}{Theorem}

\newtheorem{lemma}[theorem]{Lemma}
\newtheorem{proposition}[theorem]{Proposition}
\newtheorem{remark}[theorem]{Remark}
\let\a=\alpha
\let\e=\varepsilon
\let \z=\zeta
\let\d=\delta

\let\s=\sigma
\let\pt=\partial
\let\t=\vartheta
\let\O=\Omega
\let\G=\Gamma
\let\o=\omega
\let\g=\gamma

\let\l=\lambda

\let\f=\varphi
\let\t=\vartheta
\let\vr=\varrho
\let\b=\beta

\newcommand{\sm}{\sqrt{\mu}}
\newcommand{\R}{\mathbb{R}}
\newcommand{\cg}{\mathfrak{c}}
\newcommand{\ag}{\mathfrak{a}}
\newcommand{\bg}{\mathfrak{b}}
\renewcommand{\P}{\mathbf{P}}
\renewcommand{\S}{\mathbf{S}}
\newcommand{\ip}{(\mathbf{I}-\mathbf{P})}
\newcommand{\ipc}{(\mathbf{I}-\mathbf{P}_\c)}
\newcommand{\be}{\begin{equation}}
\newcommand{\bm}{\begin{multline}}
\newcommand{\ee}{\end{equation}}
\newcommand{\dd}{\mathrm{d}}

\renewcommand{\c}{\mathfrak{u} }

\newcommand{\1}{\mathbf{1}}
\newcommand{\lbr}{[\hskip-1pt [}
\newcommand{\rbr}{]\hskip-1pt ]}
\newcommand{\eq}{\eqref}
\numberwithin{equation}{section}
\numberwithin{theorem}{section}
\date{}
\title[Hydrodynamic Limit of a Kinetic Gas Flow Past an Obstacle\dots]{ {\large {
Hydrodynamic Limit of a Kinetic Gas Flow Past an Obstacle}}}
\author{R. Esposito}
\thanks{(R. E.) International Research Center M\&MOCS, Univ. dell'Aquila,
Cisterna di Latina, (LT) 04012 Italy}
\author{Y. Guo}
\thanks{(Y. G.) Division of Applied Mathematics, Brown University,
Providence, RI 02812, U.S.A.}
\author{R. Marra}
\thanks{(R. M.) Dipartimento di Fisica and Unit\`a INFN, Universit\`a di Roma
Tor Vergata, 00133 Roma, Italy}
\begin{document}

\begin{abstract}
Given an obstacle in $\R^3$ and a non-zero velocity with small amplitude at the infinity,   we construct the unique steady Boltzmann solution flowing around such an obstacle with the prescribed velocity as $|x|\to \infty$, which approaches the corresponding Navier-Stokes steady flow, as the mean-free path goes to zero. Furthermore, we establish the error estimate between the Boltzmann solution and its Navier-Stokes approximation. Our method consists of new $L^6$ and $L^3$ estimates in the unbounded exterior domain, as well as an iterative scheme preserving the positivity of the distribution function.
\end{abstract}

\maketitle
\tableofcontents 
\section{Introduction}\label{intro}
Let $\O$ be a smooth bounded open subset of $\R^3$ and $\overline{\O}$ its closure. A gas moves in $\O^c=\R^3\backslash\overline{\O}$ with prescribed velocity $\c$ at infinity  and vanishing velocity on $\pt \O$, evolving according to the incompressible Navier-Stokes equations. The steady boundary value problem for this system is classical in Fluid Mechanics and a huge literature has been devoted to it \cite{Bab,Fin,Lady,Lam,Ler,OS} (see also \cite{Gal} and references quoted therein). 
One of the main difficulties of this problem is related to the presence  of the ``wake" \cite{Th} and the corresponding slow decay to $\c$ of the velocity field at infinity. 

In the case of a rarefied gas, an alternative description is possible in terms of the Boltzmann equation and suitable boundary conditions. In this paper we study the link between these two descriptions in the small Knudsen numbers and low Mach numbers regime.

It is well known that in this regime the time dependent Boltzmann equation behaves as the incompressible Navier-Stokes equation,\cite{BGL89,BU,DEL,Guo06,GJ,GJJ,LM,MS,SR}. 
Much less is know  for the corresponding steady Boltzmann problem, where the natural $L^1$ and entropy estimates are not available, and only the entropy production can be exploited. 

Ukai and Asano \cite{UA1, UA2}, see also \cite{UYZ}, studied the Boltzmann equation in the exterior domain with fixed Knudsen number. They considered a rarefied gas outside a piecewise smooth convex domain of $\R^3$, with suitable boundary conditions and a prescribed Maxwellian behavior at infinity. The Maxwellian at infinity was centered at a small velocity field. For this problem Ukai and Asano were able to prove existence of the  steady solution and its dynamical stability.

Our main result is the construction of the steady solution to the Boltzmann equation in the exterior domain and  the estimate of its closeness to the steady incompressible Navier Stokes equation when Knudsen and Mach numbers are small.
Recently in \cite{EGKM2} we have constructed the solution to the Boltzmann equation for small Knudsen and Mach numbers in a smooth bounded domain, under the action of a suitably small external force and small variations of the boundary temperature. The exterior problem is even more difficult, due to the need of good decay properties for large $x$. 

\medskip
Before describing the difficulties to achieve our program, let us state more precisely the problem and the result. 

We assume that $\O\subset \R^3$ is a $C^2$ bounded domain, not necessarily convex.
Let $x \in\O^c= \R^3\backslash \overline{\O}$ and $v\in \R^3$. Let $F(x,v)\ge0$ be the (unnormalized) distribution function of a rarefied gas in $\O^c$ with position $x$ and velocity $v$, satisfying the steady Boltzmann equation 
\be v\cdot\nabla F=\frac{1}{\e}Q(F,F),\quad\text{ in }\O^c\label{1}\ee
where $\nabla\equiv\nabla_x$ and
\begin{eqnarray} Q(f,g)(v)&=& Q^+(f,g)-Q^{-}(f,g),\notag\\
Q^+(f,g)(v)&=& \int_{\R^3}  \dd v_*\int_{\{\o\in \R^3\,:\, |\o|=1\}} \dd\o B(\o, v-v_*)f(v')g(v_*'),\\
Q^-(f,g)(v)&=& f(v)\int_{\R^3}  v_*\int_{\{\o\in \R^3\,:\, |\o|=1\}} \dd\o B(\o, v-v_*)g(v_*).\end{eqnarray}
Here $v'$ and $v'_*$ are the incoming velocities in the elastic collision, defined by
\be v'= v-\o(v-v_*)\cdot \o,\quad v_*'= v_*+\o(v-v_*)\cdot \o,\ee
and $B(\o,V)$ is the cross section for hard potentials with Grad's angular cutoff, so that \newline $\int_{\{|\o|=1\}}\dd\o B(V,\o)\lesssim |V|^\theta$ for $0\le \theta\le 1$ depending on the interaction potential. In particular, $B(\o,V)=|\o\cdot V|$ for hard spheres and $\theta=1$.

We assume diffuse reflection boundary condition: Let $\gamma=\pt\O\times\R^3=\gamma_+\cup\gamma_-\cup \gamma_0$, with
\be
\gamma _{\pm} =\{(x,v)\in \partial \Omega \times \R^{3}\ :\
n(x)\cdot v\gtrless 0\}, 
\quad \gamma _{0} =\{(x,v)\in \partial \Omega \times \R^{3}\ :\
n(x)\cdot v=0\},\label{gammapm}\ee
$n(x)$ denoting the normal at $x$ to $\pt\O$ pointing inside $\O$.
Let
\begin{equation}
M_{\rho, u,T}:=\frac{\rho}{(2\pi T)^{\frac 3 2}}\exp\Big[-\frac {|v-u|^2}{2 T
}\Big],\label{maxwel}
\end{equation}
 be the {local} Maxwellian with density $\rho$, mean velocity $u$, and temperature $T$ and 
 \be \mu=M_{1,0,1}=\frac{1}{(2\pi)^{\frac 3 2}}\exp\Big[-\frac {|v|^2}{2 
}\Big].\label{mawelstan}\ee
On the boundary $F$ satisfies the \textit{diffuse reflection} condition defined as
\begin{equation}
F (x,v)=\mathcal{P}^w_\g(F)(x,v) \quad \text{ on } \g_-, \label{bc0}
\end{equation}
where
\begin{equation}
\mathcal{P}^w_{\gamma } (F)(x,v){ \ := \ }M^w(x,v)\int_{\{n(x)\cdot {w}>0\}}\dd {w}\, {F}(x,{w})\{n(x)\cdot {w}\},  \label{pgamma}
\end{equation}
with the {\it wall Maxwellian} defined as 
\begin{equation}\label{Tw}
M^w= \sqrt{2\pi}\mu=\frac{1}{2\pi}\exp\Big[-\frac {|v|^2}{2 
}\Big],\quad \int_{\{v\cdot n\gtrless 0\}} \dd v  M^w(v)|n\cdot v|= 1.\ee
We also specify the \textit{condition at infinity}. Since we  study the problem in the small Mach number regime, we assume that the velocity at infinity is of order $\e$. In other words, fixed a constant vector $\c $, denoting 
\be v_\c:=v-\e\c, \quad \mu_\c(v):=\mu(v_\c)=M_{1,\e \c,1}(v),\ee
we assume in a suitable sense
\be \lim_{|x|\to \infty} F(x,v)=\mu_\c(v).\label{Finfty}\ee
Note that we have prescribed the same uniform temperature on $\pt \O$ and at infinity for sake of simplicity, but {we} believe that a temperature difference of order $\e$ could be included. We do not discuss this. The case of sufficiently small difference of temperature for fixed $\e$ has been discussed in \cite{UYZ1}.

Let the couple velocity field and pressure, $(U, p)$,  be solution to the Stationary Incompressible Navier-Stokes equation (SINS) in $\O^c$:
\be U\cdot \nabla U + \nabla P= \mathfrak{v}\Delta U, \quad \nabla\cdot U=0, \quad U=0\text{ on }\pt \O, \quad U\to \c, \text{ as } |x|\to \infty\ee
where $\mathfrak{v}>0$ is the viscosity coefficient. 
It is convenient to represent $U=u+\c$, with $(u, P)$ solving 
\be (\c+u)\cdot \nabla u + \nabla P= \mathfrak{v}\Delta u, \quad \nabla\cdot u=0, \quad u=-\c\text{ on }\pt \O, \quad u\to 0, \text{ as } |x|\to \infty.\label{INS}\ee
Solutions to this equation do exist in $L^p$, for any $p>2$ and uniqueness is ensured for $|\c|$ small (see e.g. \cite{Gal}, Thm. X.6.4).

\medskip
  Our aim is to show that $F\approx M_{1,\e(u+\c),1}$ as $\e\to 0$. More precisely, since $M_{1,\e(u+\c),1}=\mu_\c+\e f_1\sm_\c+O(\e^2)$, where
\be f_1= \sqrt{\mu_\c} u\cdot v_\c,\label{f1}\ee
  we need to show that  
$\e^{-1}(F- \mu_\c) \approx f_1\sm_\c$ as $\e\to 0$
is in $L^p$ for any $p>2$, with the same decay of $u$.
Therefore, we set $\tilde R=\e^{-\frac 1 2}\mu_\c^{-\frac 1 2}[F-\mu_\c-\e f_1\sm_\c]$  and write the equation for $\tilde R$.
Let $L_\c$ be the usual linearized Boltzmann operator defined as
\be L_\c f=-\mu_\c^{-\frac 1 2} [Q(\mu_\c, \mu_\c^{\frac 1 2} f)+Q(\mu_\c^{\frac 1 2} f,\mu_\c)]:=\nu f -K f,\ee
where: $\nu(v)=\int_{\R^3\times \{|\o|=1\}}dv_*d\o B(v-v_*,\o)\mu(v_*)$ is such that $0\le \nu_0|v|^\theta\le \nu(v)\le \nu_1|v|^\theta$; $K$ is a compact operator on $L^2(\R^3_v)$.  $L_\c$ is an operator on $L^2(\R^3_v)$ whose null space  is 
\be \mathop{\rm Null}L_\c=\mathop{\rm span}\{1,v_\c, |v_\c|^2\}\sqrt{\mu_\c},\ee 
Let $\P_\c$ be the orthogonal projector on $\mathop{\rm Null}L_\c$. In particular, $L_\c f_1=0$.
Thus we have
\be v\cdot \nabla\tilde R+\e^{-1} L_\c \tilde R=[\G_\c(f_1,\tilde R)+ \G_\c(\tilde R,f_1)]+ \e^{\frac 1 2}\G_\c(\tilde R, \tilde R)+ \e^{-\frac 1 2} [\G_\c(f_1,f_1)-v\cdot \nabla f_1],\label{rtemp} \ee
where
\begin{eqnarray}\Gamma_\c(f,g)&=&\G_\c^+(f,g)-\G_\c^-(f,g),\notag\\ \G_\c^\pm(f,g)&=&\mu_\c^{-\frac 1 2} Q^\pm(\mu_\c^{\frac 1 2}f, \mu_\c^{\frac 1 2}g),\\
\tilde \Gamma_\c(f,g)&=&\frac 1 2\big[\Gamma_\c(f,g)+\Gamma_\c(g,f)\big].\notag\end{eqnarray}
To remove the divergent term in \eqref{rtemp}, we note that, since $\nabla\cdot u=0$, then 
\be\P_\c ({v_\c}\cdot \nabla f_1)= 0, \label{Pvnablaf1}\ee 
and 
\be f_2= L_\c^{-1}\big[-\ipc[{v_\c} \cdot\nabla f_1]+\Gamma_\c(f_1,f_1)\big]\label{f2}\ee
is well defined and is in $L^{p}$ for any $p>\frac 4 3$, because so is $\nabla  u$ (see e.g. \cite{Gal}, Thm. X.6.4).
Since $u$ solves the SINS equation, then it is easy to check that
\be \P_\c[  {v_\c}\cdot \nabla f_2+  {\c\cdot \nabla f_1}]=0. \label{Pvnablaf2}\ee
Therefore, by setting $R=\tilde R -\e^{\frac 1 2} f_2$, which means
$F=\mu_\c+\e(f_1+\e f_2+\e^{\frac 1 2}R)\sqrt{\mu_\c}$, we see that
$F$ is a stationary solution to \eq{1} if and only if $R$ solves the
{equation}:
\be v\cdot \nabla R +\e^{-1} L_\c R= L_\c^{(1)} R +\e^{\frac 1 2} \Gamma_\c(R,R) +\e^{\frac 1 2}A_\c,\label{eqR}\ee where
\be L_\c^{(1)} R = 2\tilde\Gamma_\c(f_1+\e f_2, R)
\label{L1}\ee
\be A_\c= -\ipc [v\cdot \nabla f_2]+ 2\tilde\Gamma_\c(f_1,f_2)
+\e\Gamma_\c(f_2,f_2)- {\e\c\cdot\nabla f_2}.\label{A}\ee 
Since $u\to 0$ at $\infty$, then $f_1$ and $f_2$ also go to $0$ at $\infty$. Thus we have to impose
\be\lim_{|x|\to \infty} R =0.\label{Rinfty}\ee
\medskip
For $f\in L^1(\g_\pm)$ we define  
\be P^\c_{\gamma}f= \mu_\c^{-\frac 1 2}\mathcal{P}^w_\g[\mu_\c^{\frac 1 2} f]=\sqrt{2\pi}\frac{\mu}{\sm_\c}z_{\g_+}(f),\quad z_{\g_\pm}(f)(x)=\int_{\{v\cdot n(x)\gtrless0\}}\dd v\sm_\c(v) |v\cdot n(x)| f(x,v),\label{Pgammac}\ee
$z_{\g_\pm}(f)(x)$ being the outgoing/incoming mass flux at $x\in \pt\O$.
We will omit the index $\pm$ when there is no ambiguity.

The boundary condition for $R$ is:
\be R=P^\c_{\g} R +\e^{\frac 1 2} r,\label{bcR}\ee 
where 
\be r= {P}^\c_{\g}[ f_2-\phi_\e]- (f_2-\phi_\e),\quad \text{ on } \g_-,\label{128}\ee
with $\phi_\e$ defined as 
\be \phi_\e=\e^{-2}\mu_\c^{-\frac 1 2}[ M_{1,\e(\c+ u), 1}-\mu_\c-\e \sqrt{\mu_\c} f_1],\ee 
such that
\be |\phi_\e|\le C_\beta(|u|^2+|\c|^2)\exp[-\beta |v|^2] \quad \text{ for any }\beta<\frac 1 4.\label{estzeta}\ee
Indeed, for $x\in \pt \O$, where $u(x)=-\c$,  we have $\mu=M_{1,\e(\c+u),1}$ and hence 
\be\mu= M_{1,\e(\c+ u), 1}\Big|_{\pt \O}=\mu_\c+\e \sqrt{\mu_\c} f_1\Big|_{\pt \O} + \e^2\sqrt{\mu_\c}\phi_\e\Big|_{\pt \O}.\label{mupos}\ee
and, in consequence of
$\mu=\mathcal{P}^w_\g \mu$, on $\gamma_-$  
we have
\be\mu_\c +\e f_1\mu_\c^{\frac 1 2} +\e^2\phi_\e\mu_\c^{\frac 1 2}= \mathcal{P}^w_\g[\mu_\c +\e f_1\mu_\c^{\frac 1 2} +\e^2\phi_\e\mu_\c^{\frac 1 2}].\label{mu=Pmu}\ee
On the other hand the boundary condition \eqref{bc0} for $F$ gives on $\gamma_-$  ,
\[\mu_\c +\e f_1\mu_\c^{\frac 1 2}+\e^2 f_2 \mu_\c^{\frac 1 2} +\e^{\frac 3 2} R \mu_\c^{\frac 1 2}= \mathcal{P}^w_\g[\mu_\c +\e f_1\mu_\c^{\frac 1 2}+\e^2 f_2 \mu_\c^{\frac 1 2} +\e^{\frac 3 2} R \mu_\c^{\frac 1 2}].\]
Therefore, subtracting the last two equations
\[\e^2 f_2 \mu_\c^{\frac 1 2} +\e^{\frac 3 2} R \mu_\c^{\frac 1 2}-\e^2\phi_\e\mu_\c^{\frac 1 2}=\mathcal{P}^w_\g[\e^2 f_2 \mu_\c^{\frac 1 2} +\e^{\frac 3 2} R \mu_\c^{\frac 1 2}-\e^2\phi_\e\mu_\c^{\frac 1 2}],\]
which implies \eqref{bcR}.

Note that, from the definition of $A_\c$, it follows that 
\be\P_\c  A_\c=0 % {\e\c\cdot\nabla f_2}
.\label{espPbarA}\ee
Moreover, it can be checked that  
\be z_{\g_-}(r)=\int_{\{v\cdot n<0\}}\dd v \, r\sqrt{\mu_\c}n\cdot v=0.\label{r0}\ee 
From the definition of $r$ it follows that 
\be |r|_{2,-}+ |r|_\infty\lesssim |\c|.\label{stimar}\ee

\medskip
\noindent {\bf Notation}. Depending on the context, we  denote $\|f\|_p=\|f\|_{L^p(\O_x^c\times \R^3_v)}$ or $\|f\|_p=\|f\|_{L^p(\O_x^c)}$ or $\|f\|_p=\|f\|_{L^p(\pt\O)}$ for $1\le p\le \infty$. $\|f\|_\nu=\|f \nu^{\frac 1 2}\|_2$. We set 
 $|f|_{p, \pm}=\big(\int_{\g_\pm }\dd \g|f(x,v)|^p \big)^{\frac 1 p}$, with 
\be\int_{\g_\pm}f\dd\g = \int_{\pt \O}dS(x)\int_{\{v\cdot n(x)\gtrless0\}}\dd v\,|v\cdot n(x)| f(x,v).\label{dgamma}\ee  
Finally, we define \be \lbr f\rbr_{\beta,\beta'}
=\e^{-1}\|\ipc f\|_\nu + \e^{-\frac 1 2}|(1-P_\g^\c)f|_{2,+}+ \|\P_\c f\|_6 +\e^{\frac 1 2}\|\P_\c f\|_3+ \e^{\frac 1 2}\|w f\|_\infty\label{lbrrbr}\ee
with the weight function $w(v)=\langle v\rangle^{{\beta'}}\exp[\beta |v|^2]$, where $\langle v \rangle=(1+|v|^2)^{\frac 1 2}$. 

\bigskip
The main result is
\begin{theorem}\label{mainth} Let $\O$ be a $C^2$   bounded open set of $\R^3$ and $\O^c=\R^3\backslash\overline{\O}$. Fix $\c\in \R^3$ such that $0<|\c|\ll1$. 
For any  $0<\e\ll1$ consider the steady boundary value problem  
\begin{eqnarray}\begin{cases}v\cdot\nabla F=\displaystyle{\frac{1}{\e}Q(F,F)},\quad\text{ in }\O^c\label{mainprob}\\\\
F (x,v)=\displaystyle{M^w\int_{\{v\cdot n>0\}}\hskip -.6cm F \/v\cdot n \/  \dd v} \,\, \text{ on } \g_-,\\\\ \displaystyle{\lim_{|x|\to \infty}} F(x,v)=\mu_\c(v).\end{cases}\end{eqnarray}
Then
\begin{itemize}
\item
the problem \eqref{mainprob} has a positive solution which can be represented as
\be F=\mu_\c+\sm_\c[\e f_1+\e^2 f_2 +\e^{\frac 3 2} R],\ee
with $f_1$ and $f_2$ given by \eqref{f1} and \eqref{f2},  $u$ solving \eqref{INS}, and  $R$ solving \eqref{eqR},  \eqref{bcR}. 
\item $R$ satisfies the bound 
\be\lbr R\rbr_{\beta,\beta'}
\lesssim |\c|,\label{136}\ee 
 for $\beta'\ge 0$ and ${0}<\beta\ll\frac 1 4$.
\item
$R$ is unique in the ball $\{ f\, :\, \lbr f\rbr_{0,\beta'}
\lesssim |\c|\}$.\end{itemize}
\end{theorem}
 
\begin{remark} Note that while the $L^2$ norm of $\ipc R$ is bounded and actually small as $\e\to 0$, $\P_\c R$ is bounded uniformly in $\e$ only in $L^6$, while the $L^3$ and $L^\infty$ bounds of $\P_\c R$ are divergent with $\e^{-\frac 1 2}$. It turns out that that the $L^p$ norm of $\P_\c R$ is bounded  for $p>2$, but the bound is not uniform in $\e$ for $2<p<6$. This is the counterpart of the slow decay of the velocity field $u$ at infinity, which is well known in Fluid Dynamics, where it is proved that the $L^2$ norm of $u$ is unbounded. We do not know if  a similar statement is true for $R$, but it is certainly true for $f_1$ which is linear in $u$ and hence for $\e^{-1}(F-\mu_\c)$.
\end{remark}
\begin{remark}We also note that combining the estimates implied by \eqref{136}, it follows that $\|R\|_6$ is bounded uniformly in $\e$. In fact, we have $\|\P_\c R\|_6\le\lbr R\rbr\lesssim |\c|$ and 
\[\|\ipc R\|_6\le \|\ipc R\|_2^{\frac 1 3}\|R\|_\infty^{\frac 2 3}\le (\e\lbr R\rbr)^{\frac 1 3} (\e^{-\frac 1 2}\lbr R\rbr)^{\frac 2 3}\le \lbr R\rbr\lesssim |\c|.\]
Since $f_1$ and $f_2$ are also bounded in $L^6$, uniformly in $\e$, we conclude that $\e^{-1}(F-\mu_\c)$ is bounded in $L^6$ uniformly in $\e$. The condition at infinity for $F$ is verified in this sense. 
\end{remark}
\begin{remark}
The uniqueness is proved in the ball $\{ f\, :\, \lbr f\rbr_{0,\beta'}
\lesssim |\c|\}$. No exponential decay in $v$ is required for uniqueness.
\end{remark}

\medskip

In Sections \ref{energiasec}--\ref{stimaPc3} we shall consider the following linear problem:
\be\begin{cases} v\cdot \nabla f +\e^{-1} L_\c f= g, \quad\quad (x,v)\in \O^c\\f=P_\g^\c f + \e^{\frac 1 2} r,\quad\quad\quad\quad\,\,\,\/\/  (x,v)\in \g_-, \\ \displaystyle{\lim_{|x|\to \infty} }f=0.\end{cases}\label{linprob0}\ee

{By \eqref{espPbarA} and 
\eqref{r0}, 
$\P_\c g=0
% {\e^{3/2}\c\cdot\nabla f_2}
$ and 
$z_{\g_-}(r)=0$ in the linearization of the problem \eqref{eqR},  \eqref{bcR}. However, to prove the positivity of the solution to \eqref{1} we are going to construct, we have to suitably modify the equation \eqref{1} and in the resulting linear problem to be studied \eqref{espPbarA} and 
\eqref{r0} is no more exact but $\P_\c g$ and 
$z_\g(r)$ is small for $\e$ small. Therefore in the next sections we shall drop the condition \eqref{espPbarA} and 
\eqref{r0}.}

We shall prove the following
\begin{theorem} \label{mainlinth}  Fixed $\c$ with $0<|\c|\ll1$, if $\e\ll 1$, the solution to the linear problem \eqref{linprob0} satisfies the inequality
\be \lbr f\rbr_{\beta,\beta'}^2\lesssim \mathscr{M}(g,r),\label{mainlinest}\ee
 where
\begin{eqnarray}&& \mathscr{M}(g,r)=\|\nu^{-\frac 1 2}\ipc g\|_2^2 +\e\|\nu^{-\frac 1 2}g\|^2_{\frac 32}+ \e^{3} \| \langle v\rangle^{-1} w g \|^2_{\infty}+\e | w r |^2_{\infty,-}+ |r|_{2,-}^2\notag\\\label{mathscrM}\\&&+\|\P_\c g\|_2^2+ \e^{-2} |\c|^{-2}\|\P_\c g\|_{\frac 6 5}^2+\e^{-1}|\c|^{-2+2\varrho}\|\P_\c g\|_{{\frac 65}^-}^2+(|\c|^{-2+2\varrho}+|\c|^{-1}\e^{-1})\|z_\g(r)\|^2_2 
,\notag\end{eqnarray}
for $\beta'\ge 0$ and $0\le \beta\ll \frac 1 4$ and $\rho>0$.\end{theorem}
\begin{remark} We note that the % {last term} in the 
second line of \eqref{mathscrM} vanishes when 
the hydrodynamic part of $g$ and 
the net mass flux of $r$ vanishes. This is the case for the  problem \eqref{eqR}, \eqref{bcR}. In the modified problem introduced for the proof of positivity it does not vanish, but its contribution  turns out to be small.
\end{remark}

Before going into a short sketch of the arguments we use, it is worth to comment the choice of the power of  $\e$ in front of $R$, $\alpha=\frac 3 2$. Clearly, to deal with the non linear term is easier when this power is large. However we are limited by the fact that $f_2$ does not satisfy the boundary conditions and a power $\alpha>2$ would require the introduction of a boundary layer correction with serious regularity issues due to the general geometry (see \cite{GW} for the analysis of such problems). On the other hand $\alpha\le \frac 3 2$ is required to avoid a divergent contribution from the boundary terms in the energy inequality. It turns out that the value   $\alpha=\frac 3 2$ is exactly what we need to bound the non linear term thanks to the uniform estimate we are able to obtain for $\e^{\frac 1 2}\|\P_\c R\|_3$.

Our analysis relies crucially on energy inequality to control entropy production. It gives important information: the microscopic part of the solution $\ipc R$ is of order $\e$ in $L^2$ and moreover $|(1-P_\g^\c) R|_{2,+} \sim \sqrt{\e}$. 

Our main technical achievement is  establishing the linear
estimate \eqref{mainlinest},    $\lbr f\rbr^2_{\beta,\beta'}\lesssim \mathscr{M}(g,r).$ The starting point is a new $L^{6}$ estimate for $\mathbf{P}_\c f$ in Section \ref{stimal6}, which extends the one in the recent paper \cite{EGKM2},
while  the $L^{\infty }$ estimate follows directly from \cite{EGKM2}. The key observation is
that the  $L^{6}$ estimate for the macroscopic part of $R$, $\P_\c R$, is valid in the unbounded exterior region, thanks to scaling invariance in the homogeneous Sobolev space $\dot{H}^{1}$. The proof, which requires a weak formulation and a careful choice of the test functions, is also based on delicate estimates of the boundary terms.

However, to deal with the nonlinear part $\Gamma_\c(R,R)$, the $L^6$ estimate is not sufficient, some control of the $L^3$ estimate is required.
Unlike in the bounded domains, the $L^{6}$ bound alone cannot imply $L^{3}$
bound, for $|x|\rightarrow \infty$. In fact, the $L^{3}$ bound requires faster
decay as $|x|\rightarrow \infty$, which is a much stronger estimate than $L^{6}$ estimate.  This gain of lower integrability near infinity can be viewed as opposite to the velocity averaging ideas which lead to higher integrability gain for bounded $|x|$.  In fact,  starting from the bound for $L^6$  norm, we need to show bounds on lower $p$'s norms. By working on the balance laws we can prove  a uniform in $\e$ bound for $\e^{\frac 1 2}\|\mathbf{P}_\c f\|_{3}$ for $|x|\gg1,$ which is sufficient to close our estimate (Sections \ref{iterazione}). 

To this purpose, inspired by Maslova, \cite{Masl}, in Section \ref{eqbilancio}, after multiplying the equation by a smooth spatial cutoff function $\zeta$ vanishing at $\pt \O$, we rewrite the macroscopic projection of the linear Boltzmann equation for $f^\z=\z f$  as a (non closed) system for $\mathbf{P}_\c f^{\zeta }$ in the whole space (see Eq.\eqref{massFs}, \eqref{momenFs}, \eqref{enerFs}) (In \cite{Masl} a similar system was introduced to solve the steady Boltzmann equation with $\varepsilon =1$, with in-flow boundary condition and asymptotic Maxwellian with prescribed mean velocity at infinity):  
\begin{eqnarray*}
\nabla _{x}\cdot b^{\zeta }+\e\c\cdot \nabla _{x}a^{\zeta } &=&s_{0},
\\
\nabla _{x}(a^{\zeta }+c^{\zeta })-\e\mathfrak{v}\Delta b^{\zeta }+\e\c\cdot \nabla
_{x}b^{\zeta } &=&\underline{s}, \\
\nabla _{x}\cdot b^{\zeta }-\e\kappa \Delta c^{\zeta }+\frac{3}{2}\e\c\cdot
\nabla _{x}c^{\zeta } &=&s_{4},
\end{eqnarray*}
where $\P_\c f^\z= [a^\z+ b^\z\cdot v_\c+ \frac 1 2 c^\z (|v_\c|^2-3)]\sm_\c$ and the sources $s_0$, $\underline{s}$, $s_4$ depend on $f$ and on $\z$.
For $|\c|\ll1$ 
we study the above system via Fourier analysis, by means of a decomposition of $\P_\c f^\z$ into  high-frequency  and  low-frequency parts. Of course, in the large $|x|$ regime the low-frequency part is the difficult one and its treatment requires a further decomposition in different contributions, the most delicate being the one for the total mass, momentum and energy fluxes at the boundary, needed in Lemma \ref{maslova}, which are obtained thanks to the condition $\c\neq 0$, an ingredient also entering crucially in the Fluid Dynamic treatment of the problem (see e.g. \cite{Gal}). We establish in Section \ref{stimaPc3} very precise $L^{p}$ estimates $p>2$ for the different parts of $\P_\c f$, because $\c\neq 0$ ensures more integrability than in the corresponding Stokes system. It is worth to stress that such arguments, however accurate they are, only produce an estimate of $\|\P_\c f\|_p\sim \e^{-1}$, which would not be good enough for our purposes, we need  at most $\|\P_\c f\|_3\sim \e^{-\frac 1 2}$ to deal with the non linearity because of the limitation explained before. It is only thanks to the essential uniform in $\e$ estimate of $\|\P_\c f\|_6\sim 1$, that, via a careful estimate of the mass momentum and energy fluxes at the boundary in Subsection  5.3 and interpolation, we can obtain a bound  $\sqrt{\varepsilon }\|\mathbf{P}_\c f\|_{3}\sim 1$, uniform in $\varepsilon$.

It is well-known that it is challenging to prove positivity for steady
Boltzmann solutions.  We succeed in this by suitably adapting and extending the
positivity-preserving scheme of Arkeryd and Nouri \cite{AN}. When dealing with the diffuse reflection boundary condition for this new scheme we encounter an extra difficulty with a new term determining a potential violation of the vanishing net mass flux condition at the boundary, that is controlled via accurate estimates in the large velocity set and the Ukai trace theorem \cite{UA1}.

Finally we prove our main theorem in Section \ref{iterazione} via iteration, based on the linear estimate \eqref{mainlinest}. A crucial information we need to close the iteration is the smallness of the velocity field when $|\c|$ is small. This estimate is proven in the Appendix \ref{appendice}. 
\bigskip
\section{Energy estimate}\label{energiasec}

We shall use in many points   the following two lemmas whose proof is standard and can be found for example in \cite{EGKM}:
\begin{lemma}
{\label{green}} Assume that $f(x,v), \ h(x,v)\in
L^p(\Omega^c\times\mathbb{R}^3)$, $p\ge2$ and $v \cdot \nabla_x f, v \cdot \nabla_x h  \in L^{\frac p{p-1}}(\Omega^c\times\mathbb{R}^3)$ and $f
\big|_{\gamma}, h\big|_{\gamma}\in L^2(\partial\Omega\times\mathbb{R}^3)$.
Then
\begin{eqnarray}
\iint_{\Omega^c\times\mathbb{R}^3}\dd x\dd v[(v \cdot \nabla_x h) f + (v \cdot \nabla_x f)h] = \int_{\gamma_+}\dd \g
f h - \int_{\gamma_-}\dd \g f h .  \label{steadyGreen}
\end{eqnarray}

\end{lemma}
\begin{lemma}

\label{trace_s}Assume $\Omega_1$ is an open bounded subset of $\mathbb{R}^{3}$
with $\partial(\Omega_1\backslash\overline{\O})$ in $C^{2}$, such that $\{x\in \O^c\, |\, d(x,\O)\le 1\}\subset \O_1$. We define 
\begin{equation}  \label{non_grazing}
\gamma_{\pm}^{\delta} : = \{ (x,v) \in \gamma_{\pm} : | n(x)\cdot v | >
\delta, \ \ \delta\leq |v| \leq \frac{1}{\delta} \}.
\end{equation}
Then 
\begin{equation*}
| f\mathbf{1}_{\gamma_{\pm}^{\delta}} |_{1} \lesssim_{\delta, \Omega_1} \| f
\|_{L^1(\O_1\backslash\O)} + \| v\cdot \nabla_{x} f \|_{L^1(\O_1\backslash\O)} .
\end{equation*}
\end{lemma}
\begin{remark}
Since, as proved in \cite{EGKM}, page 194, eq. (3.8), $|P_\g^\c f|_{2,\pm}\lesssim |P_\g^\c f\mathbf{1}_{\gamma_{\pm}^{\delta}}|_{2,\pm}$ and $\delta^{\theta/2}\lesssim\nu^{\frac 1 2}\lesssim \delta^{-\theta/2}$, from previous lemma applied to $\nu f^2$ we get
\be|
P_\g^\c 
f|_{2,\pm}\lesssim_\delta  \|f\|_{L^2(\O_1\backslash\O)}+ \|\nu^{-\frac 1 2}v\cdot \nabla f\|_{L^2(\O_1\backslash\O)}.\label{ukai}\ee
\end{remark}
  Next two lemmas are useful to bound the boundary terms in the energy inequality:
\begin{lemma}\label{23}
\be\Big|\int_{\pt\O} \dd S\int_{\{v\cdot n>0\}}\dd v \,v\cdot n |P^\c_\g f|^2-\int_{\pt\O} dS\int_{\{v\cdot n<0\}}\dd v\, |v\cdot n| |P^\c_\g f|^2\Big|\lesssim \e|\c|\int_{\g_+}|f|^2\dd \g.\label{ineq1}\ee
\be\Big|\int_{\pt\O} \dd S\int_{\{v\cdot n>0\}}\dd v\, v\cdot n P^\c_\g f (1-P^\c_\g)f\Big|\lesssim \e|\c|\int_{\g_+}|f|^2\dd \g.\label{ineq2}\ee
\end{lemma}
\begin{proof}
From the definition of $P^\c_\g $, 
\[\int_{\{v\cdot n\gtrless0\}}\dd v\,|v\cdot n||P^\c_\g f|^2=\sqrt{2\pi}|z_\g(f)|^2\int_{\{v\cdot n\gtrless0\}}\dd v\sqrt{2\pi}\mu^2\mu_\c^{-1}|v\cdot n|.\]
Since  by \eqref{mupos}
\[\mu^2\mu_\c^{-1}= \mu-\mu(\mu_\c-\mu)\mu_\c^{-1}= \mu - \mu[ \e\c\cdot v_\c\mu_\c +\e^2\phi_\e\sqrt{\mu_\c}]\mu_\c^{-1}=  \mu - \e\mu\c\cdot v_\c- \e^2\mu \phi_\e\mu_\c^{-\frac 1 2},\]
\begin{multline*}\int_{\{v\cdot n\gtrless0\}}\dd v\sqrt{2\pi}\mu^2\mu_\c^{-1}|v\cdot n|=\int_{\{v\cdot n\gtrless0\}}\dd v\sqrt{2\pi}|v\cdot n|[\mu-\e\mu\c\cdot v_\c-\e^2\mu \phi_\e\mu_\c^{-\frac 1 2}]\\=1-\e\int_{\{v\cdot n\gtrless0\}}\dd v\sqrt{2\pi}|v\cdot n|\mu\c\cdot(v-\e\c)-\e^2 \int_{\{v\cdot n\gtrless0\}}\dd v\sqrt{2\pi}\mu |v\cdot n|\phi_\e|\mu\mu_\c^{-\frac 1 2}=1+O(\e|\c|).\end{multline*}
The last  term is bounded because, by \eqref{mupos},  $|\phi_\e|\lesssim |\c|^2\mu_\c^{\frac 1 2^-}$.
Therefore
\[\int_{\{v\cdot n\gtrless0\}}\dd v|P_\g f|^2=\sqrt{2\pi}|z_\g(f)|^2(1+ O(\e|\c|).\]
Thus
\[\Big|\int_{\pt\O} \dd S\int_{\{v\cdot n>0\}}\dd v\, v\cdot n |P^\c_\g f|^2-\int_{\pt\O} \dd S\int_{\{v\cdot n<0\}}\dd v |v\cdot n| |P^\c_\g f|^2\Big|\lesssim O(\e|\c|)\int_{\pt\O} \dd S|z_\g(f)|^2 \]
and this proves   \eqref{ineq1}, because 
\be \int_{\pt\O}\dd S |z_\g(f)|^2\le |f|^2_{2,+}.\label{zgfl2}\ee
To prove   \eqref{ineq2} we note that
\[\int_{\{v\cdot n>0\}}\dd v\, v\cdot n P^\c_\g f (1-P^\c_\g)f
=\int_{\{v\cdot n>0\}}\dd v\, v\cdot n f  P^\c_\g f 
-\int_{\{v\cdot n>0\}}\dd v\, v\cdot n |P^\c_\g f|^2.
\]
\begin{multline*}\int_{\{v\cdot n>0\}}\dd v\, v\cdot n f  P^\c_\g f =\sqrt{2\pi} z_\g(f)\int_{\{v\cdot n>0\}}\dd v\, v\cdot n f \mu\mu_\c^{-\frac 1 2}\\=\sqrt{2\pi} z_\g(f)\int_{\{v\cdot n>0\}}\dd v\, v\cdot n f [\mu_\c^{\frac 1 2}+ (\mu-\mu_\c)\mu_\c^{-\frac 1 2}]\\= \sqrt{2\pi} |z_\g(f)|^2+ \sqrt{2\pi} |z_\g(f)|\int_{\{v\cdot n>0\}}\dd v\, v\cdot n f (\mu-\mu_\c)\mu_\c^{-\frac 1 2}
\end{multline*}
Using again $(\mu-\mu_\c)\mu_\c^{-\frac 1 2}= \e\c\cdot v_\c\mu_\c^{\frac 1 2} +\e^2\phi_\e$,
\begin{multline*}\int_{\{v\cdot n>0\}}\dd v\, v\cdot n f (\mu-\mu_\c)\mu_\c^{-\frac 1 2}\le \e|\c|\Big(\int_{\{v\cdot n>0\}}\dd v\, v\cdot n f^2\Big)^{\frac 1 2}\Big(\int_{\{v\cdot n>0\}}\dd v\, v\cdot n [|v_\c|^2\mu_\c+ \e{|\c|^{-2}}|\phi_\e|^2]\Big)^{\frac 1 2}\\\lesssim \e|\c|\Big(\int_{\{v\cdot n>0\}}\dd v\,  v\cdot n f^2\Big)^{\frac 1 2}\end{multline*}
Therefore
\[\Big|\int_{\pt\O} \dd S\int_{\{v\cdot n>0\}}\dd v\, v\cdot n P^\c_\g f (1-P^\c_\g)f\Big|\lesssim \e|\c|\int_{\pt\O} \dd S|z_\g(f)| \Big(\int_{\{v\cdot n>0\}}\dd v \, v\cdot n f^2\Big)^{\frac 1 2}\le \e|\c||f|_{2,+}^2,\]
and this   concludes the proof.
\end{proof}

\begin{lemma}\label{2.4} 
For any $\eta>0$,  
\be\Big|\int_{\pt\O} \dd S\int_{\{v\cdot n<0\}}\dd v  |v\cdot n|\e^{\frac 1 2} r  P^\c_\g f\Big| \lesssim \frac 1{\eta} \|z_\g(r)\|_2^2+\e\eta |f|_{2,+}^2+ \e^{\frac 3 2}|r|_{2,-}^2
.\label{ineq3}\ee
\end{lemma}
\begin{proof}
We note that
\begin{multline*}\e^{\frac 1 2}\int_{\{v\cdot n <0\}} rP_\g^\c f\dd v |v\cdot n|=\e^{\frac 1 2}\sqrt{2\pi}z_\g(f)\int_{\{v\cdot n <0\}}\dd v\, r|v\cdot n|\mu\mu_\c^{-\frac 1 2}\\= \e^{\frac 1 2}\sqrt{2\pi}z_\g(f)z_\g(r)+\e^{\frac 1 2}\sqrt{2\pi}z_\g(f)\int_{\{v\cdot n <0\}}\dd v\, r|v\cdot n|(\mu-\mu_\c)\mu_\c^{-\frac 1 2}\end{multline*}
The integral on $\pt \O$ of the first term is bounded by 
\begin{multline*}\e^{\frac 1 2}|\sqrt{2\pi}\int_{\pt\O}\dd S|z_\g(r)|z_\g(f)|\le \frac {2\pi}{4\eta}\int_{\pt\O}\dd S|z_\g(r)|^2+ \eta\e \int_{\pt\O}\dd S|z_\g(f)|^2\\\lesssim \eta^{-1}\|z_\g(r)\|_2^2+\e\eta |f|_{2,+}^2].\end{multline*} 
The second by is bounded by 
\begin{multline*}\Big|\e^{\frac 1 2}\int_{\pt\O}\dd S|z_\g(f)|\int_{\{v\cdot n <0\}}\dd v\, r|v\cdot n|(\mu-\mu_\c)\mu_\c^{-\frac 1 2}\Big|\\\lesssim \e^{\frac 3 2}|\c|\int_{\pt\O}\dd S\Big(\int_{\{v\cdot n <0\}}\dd v\, |v\cdot n||r|^2\Big)^{\frac 1 2}|z_\g(f)|\le \e^{\frac 3 2}|\c|(|r|_{2,-}^2+|f|_{2,+}^2)\end{multline*} and we obtain \eqref{ineq3}.
\end{proof}

\medskip
For fixed $\e$ the construction of the solution to the linear problem \eqref{linprob0} is standard, see e.g. \cite{Masl}. To prove Theorem \ref{mainlinth}, we begin with the energy inequality.  

\begin{proposition}\label{propener} 
For 
$|\c|$ sufficiently small the solution to \eqref{linprob0}  
 satisfies the inequality
\begin{multline} \e^{-2} \|\ipc f\|_\nu^2 +\e^{-1}|(1-P_\g^\c)f|^2_{2,+}\lesssim \|\nu^{-\frac 1 2}\ipc g\|_2^2 +|\c|^2\|\P_\c f\|_6^2+ (1+\e^{\frac 1 2})|r|_{2,-}^2\\+(\e|\c|)^{-1}\|z_\g(r)\|_2^2+ \e^{-2} |\c|^{-2}\|\P_\c g\|_{\frac 6 5}^2+\|\P_\c g\|_2^2.\label{ener99}\end{multline}
\end{proposition}
\begin{proof}
Use (\ref{steadyGreen}) with $h=f$. Then, multiplying by $\e^{-1}$  we have 
\[\frac 1 2 \e^{-1}\int_{\g_+}\dd \g f^2-\frac 1 2\e^{-1}\int_{\g_-}\dd \g f^2 +\e^{-2}\int_{\O^c\times \R^3}\dd x\dd vfL_\c f-\e^{-1}\int_{\O^c\times \R^3}\dd x\dd v fg=0.\]
We use the spectral inequality (see e.g. \cite{CIP}, Th. 7.2.5),
\[\e^{-2}\int_{\O^c\times\R^3}\dd x\dd v\, fL_\c f\gtrsim \e^{-2}\|\ipc f\|_\nu^2.\]
Moreover, using the  Holder inequality to bound $|(\P_\c f,\P_\c g)|\le \|\P_\c f\|_6\|\P_\c g\|_{\frac 6 5}$,
\begin{multline}\e^{-1}\Big|\int_{\O^c\times \R^3}\dd x\dd v fg\Big|\le \e^{-1}\|\ipc f\|_\nu\|\nu^{-\frac 1 2}\ipc g\|_2+ \e^{-1}\|\P_\c f\|_6\|\P_\c g\|_{\frac 6 5}\le\\
\eta_1\e^{-2}\|\ipc f\|_\nu^2+ \frac 1 {4\eta_1}\|\nu^{-\frac 1 2}\ipc g\|^2_2+ \eta_2\|\P_\c f\|_6^2+ \frac1 {4\eta_2 \e^2}\|\P_\c g\|^2_{\frac 6 5}   .  \end{multline}

From the boundary conditions, on $\g_-$ we have $f=P_\g^\c f + \e^{\frac 1 2} r$. Hence, using Lemma \ref{2.4},
\begin{multline}\e^{-1}\int_{\g_-}\dd \g f^2=\e^{-1}\int_{\g_-}\dd \g [P_\g^\c f+\e^{\frac 1 2} r]^2=\e^{-1} \int_{\g_-}\dd \g (|P_\g^\c f|^2+\e |r|^2+2\e^{\frac 1 2} rP_\g^\c f )\\=\e^{-1}\Big[\int_{\g_-}\dd \g |P_\g^\c f|^2 + \e|r|_{2,-}^2+\e^{\frac 3 2}|r|_{2,-}^2+\frac 1{\eta}\|z_\g(r)\|_2^2+\e\eta|f|_{2,+}^2\Big].\end{multline}
 Moreover
\[\e^{-1}\int_{\g_+}\dd \g f^2=\e^{-1}\int_{\g_+}\dd \g[(1-P_\g^\c)f]^2+ \e^{-1}\int_{\g_+}\dd \g[P_\g^\c f]^2 +2\e^{-1}\int_{\g_+}\dd \g[(1-P_\g^\c)f][P_\g^\c f].\]
The last term is bounded by (\ref{ineq2}) and the second is replaced by $ \int_{\g_-}[P_\g^\c f]^2$ by using (\ref{ineq1}). Then $(|\c|+\eta)|f|^2_{2,+}$ is split into $(|\c|+\eta)|(1-P_\g^\c)f|^2_{2,+}+(|\c|+\eta)|P_\g^\c f|^2_{2,+}$
Collecting the terms  and choosing  $\eta=|\c|$, $\eta_1$ sufficiently small and $\eta_2=|\c|^2$  we have the energy inequality
\begin{multline*}\e^{-2} \|\ipc f\|_\nu^2 +\e^{-1}|(1-P_\g^\c)f|_{2,+}\\\lesssim \|\nu^{-\frac 1 2}\ipc g\|_2^2+ \frac 1{\e^{2}|\c|^2}\|\P_\c g\|_{\frac 6 5}^2 + (1+\e^{\frac 1 2})|r|_{2,-}^2+(\e|\c|)^{-1}\|z_\g\|_2^2+ |\c||P_\g^\c f|_{2,+}^2+|\c|^2\|\P_\c f\|_6^2,\end{multline*}
where we have used
$\int_{\g_+}[(1-P_\g^\c)f]^2-|\c||(1-P_\g^\c)f|^2_{2,+}\gtrsim |(1-P_\g^\c)f|^2_{2,+}$ for $|\c|$ sufficiently small. Next we use (\ref{ukai}) to bound
\[|\c||P_\g^\c f|_{2,+}^2\le |\c|\|f\|_{L^2(\O_1\backslash\O)}^2+ |\c|\|\e^{-1}
\ipc f\|_{L^2(\O_1\backslash\O)}^2+|\c|\|\nu^{-\frac 1 2}
g\|_{L^2(\O_1\backslash\O)}^2\]
Moreover, we split $\|f\|_{L^2(\O_1\backslash\O)}^2=\|\ipc f\|_{L^2(\O_1\backslash\O)}^2+\|\P_\c f\|_{L^2(\O_1\backslash\O)}^2$ and bound
\[\|\P_\c f\|_{L^2(\O_1\backslash\O)}^2\lesssim \/_{\O_1}\|\P_\c f\|_{6}^2.\]
Finally, we bound
\[\|\nu^{-\frac 1 2}
g\|_{L^2(\O_1\backslash\O)}^2\lesssim \|{\nu}^{-\frac 1 2}
{\ipc g}\|_2^2 + \|\P_\c g\|_2^2.\]
We have so proved  \eqref{ener99}.
\end{proof}

\begin{proposition} \label{linftyest}Let $w=e^{\beta'|v|^2}\langle v\rangle^{\beta}$. 
 Then, for $0\le\beta'\ll1/4$ and $\beta\ge 0$ we have
\be\e^{\frac 1 2}\|w f\|_{L^\infty(\O^c)}\lesssim \e^{-1}\|\ipc f\|_\nu +\|\P_\c f\|_6 +  \ \e^{\frac 12} | w\, r |_{\infty}
+ \e^{\frac 32} \| \langle v\rangle^{-1} w\, g \|_{\infty} .\label{linf99}\ee
\end{proposition}
\begin{proof} As in \cite{EGKM2}, {Prop. 2.6}. \end{proof}

\section{$L^6$ estimate of $\P_\c f$}\label{stimal6}
Given $g$ and $r$, we consider the weak version of the linear problem \eqref{linprob0}:
for any  test function $\psi$,
\begin{multline}\int_{\g_+} \dd\g  f\psi-\int_{\O^c\times \R^3}\dd x\dd v\,  f v\cdot \nabla \psi +\e^{-1}\int _{\O^c\times \R^3}\dd x\dd v \, \psi L_\c f\\= \int_{\O^c\times \R^3}\dd x\dd v \, g \psi+\int_{\g_-} \dd\g   (P_\g ^\c f+\e^{\frac 1 2} r)\psi.\label{weak}\end{multline}

Remind that $\P_\c f= \sqrt{\mu_\c}[ a +b\cdot v_\c +\frac1 2({|v_\c|^2-3})]$.
To get a $L^6$ bound on $\P_\c f$ we bound separately the functions $a$, $b$ and $c$ by means of suitable choices of the test functions $\psi$. To this end we will need to solve $-\Delta \phi =h\in L^{6/5}(\Omega ^{c})$ with Dirichlet or Neumann boundary conditions. 
\begin{lemma}
For exterior domain $\Omega ^{c}$ with $C^2$ boundary $\partial \Omega $, there exists a unique solution to $-\Delta \phi =h\in
L^{6/5}(\Omega ^{c})$ \bigskip with either Dirichlet or Neumann boundary conditions such that 
\begin{equation}
\|\nabla \phi \|_{L^{2}}+\|\phi \|_{L^{6}}+\|\nabla ^{2}\phi
\|_{L^{6/5}}\leq \|h\|_{L^{6/5}}.  
\label{phih1}\end{equation}
\end{lemma}

\begin{proof}
We solve $-\Delta \phi =f\in L^{6/5}(\Omega ^{c})$ by the Lax-Milgram
theorem: define a bilinear form 
\[
\lbrack \lbrack \nabla \phi ,\nabla \psi ]]\equiv \int_{\O^c}\dd x\dd v \, \nabla \phi \cdot
\nabla \psi 
\]
with the functional $h$ defined by
\[
\langle h,\psi \rangle \equiv \int_{\O^c} \dd x\dd v \,f\psi .
\]
We choose homogeneous Sobolev space $\dot{H}^{1}(\Omega ^{c})$, with norm $\|\phi\|_{\dot{H}^{1}(\Omega ^{c})}=\|\nabla\phi \|_{L^2(\Omega ^{c})}$ for  Neumnann boundary conditions and $\dot{H}_{0}^{1}(\Omega ^{c})$ for Dirichlet boundary conditions.

We have the Sobolev embedding 
\[
\|\xi \|_{L^{6}(\Omega ^{c})}\lesssim \|\nabla \xi \|_{L^{2}(\Omega ^{c})}
\]
(see \cite{EV}, p. 263). Therefore  $\langle h,\psi
\rangle $ defines a bounded linear functional in $\dot{H}^{1}(\Omega ^{c})$ thanks to the inequality
\[
\langle h,\psi \rangle = \int_{\O^c}\dd x\dd v \, f\psi \leq \|f\|_{L^{6/5}}\|\psi
\|_{L^{6}}\leq c_{h}\|\nabla \psi \|_{L^{2}}.
\]
The existence and uniqueness as well as the first two inequalities then
follows from Lax-Milgram theorem. 
To bound $\|\nabla ^{2}\phi \|_{L^{6/5}},$ we take a smooth cutoff function $\chi$
such that 
\[
\Delta (\chi \phi )=\chi h+2\nabla \chi \cdot \nabla \phi +\Delta \chi \phi
\in L^{6/5}.
\]
If $\chi $ is zero near $\partial \Omega$, then, by the $W^{2,p}$ estimate for the whole
space, and the fact $\nabla \chi $ has compact support,  
\begin{eqnarray*}
\|\nabla ^{2}\chi \phi \|_{L^{6/5}} &\leq &\|\chi h+2\nabla \chi \cdot
\nabla \phi +\Delta \chi \phi \|_{L^{6/5}} \\
&\leq &\|h\|_{L^{6/5}}.
\end{eqnarray*}
On the other hand, if $\chi $ is zero for $|x|$ large, then by the $W^{2,p}$
estimate for mixed Dirichlet-Neumann b.c. in a fixed domain, we have 
\begin{eqnarray*}
\|\nabla ^{2}(\chi \phi )\|_{L^{6/5}} &\leq &\|\chi h+2\nabla \chi \cdot
\nabla \phi +\Delta \chi \phi \|_{L^{6/5}}+\|\chi \phi \|_{L^{6/5}} \\
&\leq &\|h\|_{L^{6/5}}.
\end{eqnarray*}
We therefore conclude \eq{phih1}. 
\end{proof}

\begin{proposition}\label{32}
If $|\c|$ is sufficiently small 
we have:
\begin{multline} \|\P_\c f\|_{6}\lesssim \e^{-1}\|\ipc f\|_\nu +\|\ipc f\|_6+
\| g{{\nu}^{-\frac 1 2}}\|_2+\e^{-\frac 1 2}|(1-P_\g^\c)f|_{2,+}\\+ \e^{\frac 1 2} |r|_\infty+ o(1)[\e^{\frac 1 2}\| f\|_{\infty}] 
.\label{P699}\end{multline}
\end{proposition}
\begin{remark}
Note that 
\begin{multline*}\|\ipc f\|_6\le \|\ipc f\|_2^{1/3}\|\ipc f\|_\infty^{2/3}=\e^{1/3}\|\e^{-1}\ipc f\|_2^{1/3}\e^{-\frac 12 \times \frac 2 3}\|\e^{\frac 1 2}\ipc f\|_\infty^{2/3}\\\lesssim
\eta \|\e^{\frac 1 2}\ipc f\|_\infty+\frac 1 {\eta} \|\e^{-1}\ipc f\|_2.\end{multline*}
Therefore by choosing $\eta$ small we obtain
\begin{multline} \|\P_\c f\|_{6}\lesssim \e^{-1}\|\ipc f\|_\nu +
\| g{{\nu}^{-\frac 1 2}}\|_2+\e^{-\frac 1 2}|(1-P_\g^\c)f|_{2,+}+ \e^{\frac 1 2} |r|_\infty+ o(1)[\e^{\frac 1 2}\| f\|_{\infty}] 
.\label{P6999}\end{multline}
\end{remark}
\begin{proof}
$\/$

\noindent\underline{Step 1}:

In order to get a bound for $c$, we choose the function $\psi_c$ in (\ref{weak}) as
\[ \psi_c=\sqrt{\mu_\c} (|v_\c|^2-\beta_c) v_\c\cdot \nabla \f_c,\]
with $\beta_c$ a suitable constant to be chosen later and $\f_c$ solution to the problem
\be -\Delta \f_c= c^5 \text{ in } \O^c,\quad \f=0 \text{ on } \pt\O.\label{deltac}\ee
Hence, by previous discussion, there is a unique $\f_c$ and
\be \|\nabla\f_c\|_{\dot H^1(\O^c)}\le \|c^5\|_{L^{\frac 6 5}(\O^c)}=\|c\|_{L^6(\O^c)}^5.\label{nablac}\ee
We start computing the term $\int_{\O^c\times \R^3} dx dv f v\cdot \nabla \psi_c$. We have:
\[\int_{\O^c\times \R^3} \dd x\dd v \, f v\cdot \nabla \psi_c=\int_{\O^c\times \R^3} \dd x\dd v \,f v_\c\cdot \nabla \psi_c+\e\int_{\O^c\times \R^3} \dd x\dd v \, f \c\cdot \nabla \psi_c.\]
By (\ref{nablac}), $f=\P_\c f+\ipc f$ and the Young inequality,
\[\Big|\e\int_{\O^c\times \R^3} \dd x\dd v \, f \c\cdot \nabla \psi_c\Big|\le \e|\c|\|c\|_6^5\|f\|_6\lesssim \e|\c|\|\P_\c f\|^6_6+\e|\c|\|\ipc f\|_6^6.\]
By using $f=\P_\c f+\ipc f$ and the expression of $\P_\c f$, we need to compute 
\be\int_{\O^c\times \R^3}\dd x\dd v \,  a \sqrt{\mu_\c} v_\c\cdot \nabla \psi_c,\label{a1}\ee 
\be
\int_{\O^c\times \R^3} \dd x\dd v \,  b\cdot v_\c\sqrt{\mu_\c} v_\c\cdot \nabla \psi_c,\label{b1}\ee
\be\int_{\O^c\times \R^3} \dd x\dd v \,  c \frac{|v_\c|^2-3}2 \sqrt{\mu_\c} v_\c\cdot \nabla \psi_c,\label {c1}\ee
\be\int_{\O^c\times \R^3} \dd x\dd v \, v\cdot \nabla \psi_c \ipc f=\int_{\O^c\times \R^3}  \dd x\dd v \,\sqrt{\mu_\c}(|v_\c|^2-\beta_c)v_\c\otimes v_\c:\nabla\otimes\nabla\f_c\ipc f.\ee
Using (\ref{nablac}), by the Young inequality, the last one is bounded by 
\[\|\ipc f\|_6 \|c\|^5_{L^6(\O^c)}\le \frac 5 6\eta \|c\|^6_{L^6(\O^c)}+\frac 1 6 \eta^{- {\frac15}} \|\ipc f\|_6^6,\]
for any $\eta>0$.

With the choice $\beta_c=5$ it results
\be \int_{\R^3} \dd v (|v_\c|^2-\beta_c) v_\c\otimes v_\c {\mu_\c=0},\label{betac}\ee 
and the term in (\ref{a1}) vanishes. The term \eqref{b1} vanishes because is odd in $v_\c$. Next we compute the term \eqref{c1}. We have
\be\int_{\R^3}\dd v\, v_\c\otimes  v_\c\frac{|v_\c|^2-3}2 (|v_\c|^2-\beta_c) \mu_\c=5\mathbf{I}.\label{betac0}\ee
Therefore
\begin{multline*}\int_{\O^c\times \R^3} \dd x\dd v \,  c \frac{|v_\c|^2-3}2 \sqrt{\mu_\c} v_\c\cdot \nabla \psi_c=\int_{\O^c} \dd x\, c \nabla\otimes \nabla \f_c :\int_{\R^3}\dd v \,  v_\c\otimes  v_\c\frac{|v_\c|^2-3}2(|v_\c|^2-\beta_c) \mu_\c=\\5\int_{\O^c}\dd x\,  c \Delta \f_c=-5\int_{\O^c}\dd x\, |c|^6=-5\|c\|_{L^6(\O^c)}^6,
\end{multline*}
because of \eqref{deltac}.
By (\ref{phih1}) and Young inequality, we have
\[\e^{-1}\Big|\int _{\O^c\times \R^3}\dd x\dd v \,  \psi_c L_\c f\Big|\le \e^{-1}\|\nabla\f_c\|_{L^2(\O^c)}\|\ipc f\|_\nu\le \frac 56 \eta \|c\|_6^6 +\frac 1 6(4\eta)^{- {\frac15}} [\e^{-1}\|\ipc f\|_\nu]^6,\]
for any $\eta>0$. 

Similarly, we get
\[\Big|\int _{\O^c\times \R^3}\dd x\dd v\,  \psi_c g\Big|\lesssim \|\nabla\f_c\|_{L^2(\O^c)}\| g{{\nu}^{-\frac 1 2}}\|_2\le \frac 56 \eta \|c\|_6^6 +\frac 1 6(4\eta)^{- {\frac15}}\| g{{\nu}^{-\frac 1 2}}\|_2^6,\]
for any $\eta>0$. 

Next we compute the boundary terms. We decompose $f$ on $\g$ as $f= P^\c_\g f +\1_{\g_+}(1-P^\c_\g)f + \1_{\g_-}\e^{\frac 1 2} r$.

First consider the term 
\[\int_\g\dd \g\, P_\g^\c f \psi_c=\int_{\pt \O} \dd S(x)\nabla \f_c\cdot \int_{\R^3} \dd v (n\cdot v)v_\c (|v_\c|^2-\beta_c)\sqrt{\mu_\c}P_\g^\c f.\]
From the expression of $P_\g^\c f$ we see that 
\[\sqrt{\mu_\c} P^\c_\g f= \sqrt{2\pi}\mu z_\g(f).\]
Therefore, we need to compute $\int_{\R^3} dv \int_{\R^3} dv (n\cdot v)v_\c (|v_\c|^2-\beta_c)\mu(v)$. We have
\[v_\c (|v-\e\c|^2-\beta_c)=v(|v|^2-\beta_c) +  {\e (-\c|v|^2-2\c\cdot v v+\beta_c \c)+ \e^2(|\c|^2v +2 \c\cdot v\c) -\e^3 |\c|^2\c}.\] 
Since the terms of order $\e$ and $\e^3$ are even in $v$, after multiplication by $v\cdot n$,
their contributions vanish (note that the integration in $v$ is on the full $\R^3$, not on $\{v\cdot n \lessgtr 0\}$. The contribution of the term of order $1$ vanishes by the choice of $\beta_c$ (\ref{betac}), so we conclude that
\[\int_\g\dd \g P_\g^\c f \psi_c=\e^2\int_{\pt \O} \dd S(x)\nabla \f_c\cdot \int_{\R^3} \dd v (n\cdot v)(|\c|^2v +2 \c\cdot v\c)\sqrt{\mu_\c}P_\g^\c f.\]
We need the Sobolev trace theorem to bound $\nabla \f_c$ on $\pt \O$.
\begin{lemma}\label{traces}
\[ \|\nabla \f_c\|_{L^{\frac 4 3}(\pt\O)}\le \|c\|_6^5.\]
\end{lemma}
\begin{proof}
If $\O$ is a $C^1$ domain in $\R^N$, we have the following trace estimate \cite{Leo}, p. 466:
\begin{equation}
\left(\int_{\pt \O } \dd S(x) |u|^{\frac{p(N-1)}{N-p}}\right)^{\frac{N-p}{p(N-1)
}}\le C(N,p)\left( {\int_{\O_1\backslash\O } \dd x |u|^{p}}+\int_{\O_1\backslash\O } dx |\nabla
u|^{p}\right)^{\frac{1}{p}}.
\end{equation}
This is a consequence of the trace theorem $W^{1,p}(\O_1\backslash \O ) \rightarrow W^{1- 
\frac{1}{p}, p}(\partial(\O_1\backslash\O) ),$ and the Sobolev embedding in $N-1$
dimensional sub-manifold $(W^{1-\frac{1}{p}, p} (\partial(\O_1\backslash\O)  ) \subset L^{ 
\frac{p (N-1)}{N-p} } (\O_1\backslash\O) $ for $\frac{N-p}{p (N-1)}=\frac{1}{p}- \frac{
1- \frac{1}{p}}{N-1}$). In particular, with $p=\frac 65$ and $N=3$ we have $
\frac{p(N-1)}{N-p}= 
\frac 4 3$. With $u=\nabla\f_{c}$, we have 
\begin{equation}  \label{4/3}
\| \nabla_{x} \f_{c} \|_{L^{\frac 4 3} (\partial\O )} \lesssim \|c\|_{L^6(\O_1\backslash\O)}^5\le\| c \|^{5}_{L^{6}
(\Omega^c)}.
\end{equation}\end{proof}
Therefore, by Holder inequality, 
\[\Big|\int_\g\dd \g P_\g^\c f \psi_c\Big|\le \e^2|\c|^2\| \nabla_{x} \f_{c} \|_{L^{\frac 4 3}(\pt\O)} \|
P_\g^\c f\|_{L^4(\g)}.\]
Since $\|
P_\g^\c f\|_{L^4(\g)} {\lesssim \e^{-{\frac 1 2}}[\e^{\frac 1 2}| f|_\infty]\le \e^{-{\frac 1 2}}[\e^{\frac 1 2}\| f\|_\infty]}$, we obtain 
\be\Big|\int_\g \dd\g P_\g^\c f \psi_c\Big|  {\lesssim} 
\e^2|\c|^2 \e^{-\frac 1 2}[\e^{\frac 1 2}\|f\|_\infty]\|c\|_6^5\lesssim \e^2\e^{-{\frac 1 2}}|\c|^2\frac 5 6 \|c\|_6^6+ \e^2|\c|^2\frac  16\e^{-\frac 1 2}[\e^{\frac 1 2}\| f\|_\infty]^6.\label{pgammapsic}\ee
Next, we need to bound   $\int_\g \1_{\g_+}(1-P_\g^\c) f \psi_c$. We have
\[\Big|\int_\g\dd\g \1_{\g_+}(1-P_\g^\c) f \psi_c\Big|\le \| \nabla_{x} \f_{c} \|_{L^{\frac 4 3}(\pt\O)} \|
\1_{\g_+}(1-P_\g^\c) f\|_{L^4(\g)}.\]
But 
\[\|\1_{\g_+}(1-P_\g^\c) f\|_{L^4(\g)}\le [\e^{-\frac 1 2}\|
\1_{\g_+}(1-P_\g^\c) f\|_{L^2(\g)}]^{\frac 1 2}[\e^{\frac 1 2}\|\1_{\g_+}(1-P_\g^\c) f\|_{L^\infty(\g)}]^{\frac 1 2}.\]
Thus, we conclude that, for any $\eta>0$ and $\eta'>0$
\be \Big|\int_\g \1_{\g_+}(1-P_\g^\c) f \psi_c {d\g}\Big|\lesssim \eta \|c\|_6^6 + \eta'[\e^{\frac 1 2}\| f\|_{\infty}]^6+C_{\eta, \eta'}[\e^{-\frac 1 2}\|\1_{\g_+}(1-P_\g^\c) f\|_{L^2(\g)}]^{6}\Big\}.\ee
In conclusion the boundary terms are bounded, for any $\eta>0$, $\eta'>0$, by
\[\Big|\int_{\gamma}\dd \g\psi_c
[P_\g^\c f +\1_{\g_+}(1-P_\g^\c)f]\Big|
\le \eta\|c\|^6_6 + \eta'[\e^{\frac 12}\|f\|_\infty]^6+C_{\eta,\eta'}[\e^{-\frac 1 2}\|(1-P_\g^\c)f)|_{2.\g_+}]^6.
\]
Finally, 
\[\Big|\int_{\g_-} \e^{\frac 12}\dd\g \, r \psi_c\Big |\le \|\nabla\f_c\|_{L^{4/3}(\pt\O)} \|\e^{\frac 1 2} r\|_{L^{4}(\pt\O)}\le \e^{\frac 1 2}\|c\|_6^5  |r|_\infty.\]
By collecting all the terms and choosing $\eta$ and $\eta'$ sufficiently small we conclude that
\begin{multline} 
\|c\|_6\lesssim 
\e^{-1}\|\ipc f\|_\nu +\|\ipc f\|_6+ {(\e|\c|)^{\frac 1 6}}\|\P_\c f\|_6+
\|{g}{{\nu}^{-\frac 1 2}}\|_{L^2(\O^c\times\R^3)}\\+\e^{-\frac 1 2}|(1-P_\g^\c)f|_{2,+}+ \e^{\frac 1 2} |r|_\infty+ o(1)[\e^{\frac 1 2}\|f\|_\infty].
\end{multline}

\noindent
\underline{Step 2}:

In order to estimate $b$ we shall use two test functions. The first is chosen as follows: for fixed $i,j$ 
\begin{equation}
\psi =\psi _{b}^{i,j}\equiv (v_{\c,i}^{2}-\beta _{b})\sqrt{\mu_\c }\partial _{j}
\f _{b}^{j},\quad i,j=1,\dots ,d,  \label{phibj}
\end{equation}
where $\beta _{b}$ is a constant to be determined, and 
\begin{equation}
-\Delta _{x}\f _{b}^{j}(x)=b^{5}_{j}(x),\ \ \ \f _{b}^{j}|_{\partial
\Omega }=0.  \label{jb}
\end{equation}
As before, there is a unique $\f_b^j$ and
\be \|\nabla\f_b^j\|_{\dot H^1(\O^c)}\le \||b_j|^5\|_{L^{\frac 6 5}(\O^c)}=\|b_j\|_{L^6(\O^c)}^5.\label{nablabj}\ee
We start computing the term $\int_{\O^c\times \R^3} dx dv f v\cdot \nabla \psi _{b}^{i,j}$. We have:
\[\int_{\O^c\times \R^3} \dd x \dd v\, f v\cdot \nabla \psi _{b}^{i,j}=\int_{\O^c\times \R^3} \dd x \dd v\, f v_\c\cdot \nabla \psi _{b}^{i,j}+\e\int_{\O^c\times \R^3} \dd x \dd v\, f \c\cdot \nabla \psi _{b}^{i,j}.\]
By (\ref{nablabj}) and the Young inequality,
\[\Big|\e\int_{\O^c\times \R^3} \dd x \dd v\, f \c\cdot \nabla \psi _{b}^{i,j}\Big|\le \e|\c|\|b_j\|_6^5\|f\|_6\le \e|\c|\|\P_\c f\|^6_6+\e|\c|\|\ipc f\|_6^6.\] 
By using $f=\P_\c f+\ipc f$ and the expression of $\P_\c f$, we need to compute 
\be\int_{\O^c\times \R^3} \dd x \dd v\, a \sqrt{\mu_\c} v_\c\cdot \nabla \psi _{b}^{i,j},\label{a1b}\ee 
\be
\int_{\O^c\times \R^3} \dd x \dd v\,  b\cdot v_\c \sqrt{\mu_\c} v_\c\cdot \nabla \psi _{b}^{i,j},\label{b1b}\ee
\be\int_{\O^c\times \R^3} \dd x \dd v\,  c \frac{|v_\c|^2-3}2 \sqrt{\mu_\c} v_\c\cdot \nabla \psi _{b}^{i,j},\label {c1b}\ee
\be\int_{\O^c\times \R^3} \dd x \dd v\, v\cdot \nabla \psi _{b}^{i,j} \ipc f=\int_{\O\times \R^3}\dd x \dd v\,  \sqrt{\mu_\c}(|v_\c|^2-\beta_b)^2v_\c\otimes v_\c:\nabla\otimes\nabla\f_{b}^j\ipc f.\ee
Using (\ref{nablabj}), the last one is bounded by 
\[\|\ipc\|_6 \|{b_j}\|^5_{L^6(\O^c)}\le \frac 5 6\eta \|{b_j}\|^6_{L^6(\O^c)}+\frac 1 6 \eta^{- {\frac15}} \|\ipc\|_6^6,\]
for any $\eta>0$.
By oddness the terms in (\ref{a1b}) and (\ref{c1b}) vanish.
We choose $\beta _{b}>0$ such that for all $i$, 
\begin{equation}
\int_{\mathbb{R}^{3}}[v_{\c,i}^{2}-\beta _{b}]\mu_\c (v) \mathrm{d} v=\frac 1{\sqrt{2\pi}}\int_{
\mathbb{R}}\dd v[v_{\c,1}^{2}-\beta _{b}]e^{-\frac{|v_{\c,1}|^2}2}\mathrm{d}
v_{1}=0,  \label{alpha}
\end{equation}
and we find $\beta_b=1$.
Note that for such choice of $\beta _{b}$ and for $i\neq k$, by an explicit
computation 
\begin{eqnarray*}
\int_{\R^3} (v_{\c,i}^{2}-\beta _{b})v_{\c,k}^{2}\mu_\c \mathrm{d} v &=&0, \\
\int_{\R^3} (v_{\c,i}^{2}-\beta _{b})v_{\c,i}^{2}\mu \mathrm{d} v &=& 2.
\end{eqnarray*}
As a consequence  

\begin{multline*}\sum_{k,\ell}\int_{\O^c\times \R^3} \dd x \dd v\,  b_k v_{\c,k} \sqrt{\mu_\c} v_{\c,\ell}(v_{\c,i}^{2}-\beta _{b})\sqrt{\mu_\c }\pt_\ell\partial _{j}
\f _{b}^{j}\\=\sum_{k,\ell}\d_{k,\ell}\d_{\ell,i}\int_{\O^c}\dd x\, b_k\pt_\ell\partial _{j}
\f _{b}^{j} = \int_{\O^c}\dd x\, b_i\pt_i\partial_j \f_{b}^{j}.\end{multline*}
We have also
\[\e^{-1}\Big|\int _{\O^c\times \R^3}\dd x \dd v\,  \psi_b^{i,j} L_\c f\Big|\le \e^{-1}\|\nabla\f_{b}^j\|_{L^2(\O^c)}\|\ipc f\|_\nu\le \frac 56 \eta \|b_j\|_6^6 +\frac 1 6\eta^{- {\frac15}} [\e^{-1}\|\ipc f\|_\nu]^6,\]
for any $\eta>0$. 

Similarly, we get
\[\Big|\int _{\O^c\times \R^3} \dd x \dd v\, \psi_b^{i,j} g\Big|\lesssim \|\nabla\f_{b_j}\|_{L^2(\O^c)}\| g{{\nu}^{-\frac 1 2}}\|_2\le \frac 56 \eta \|b_j\|_6^6 +\frac 1 6\eta^{- {\frac15}}\| g{{\nu}^{-\frac 1 2}}\|_2^6,\]
for any $\eta>0$. 

Next we compute the boundary terms. We decompose $f$ on $\g$ as $f= P^\c_\g f +\1_{\g_+}(1-P^\c_\g)f + \1_{\g_-}\e^{\frac 1 2} r$.
First consider the term 
\[\int_\g\dd\g P_\g^\c f \psi_{b}^{i,j}=\int_{\pt \O}\dd S(x)\nabla \f_{b}^j\cdot \int_{\R^3} \dd v (n\cdot v)(|v_{\c,i}|^2-\beta_b)\sqrt{\mu_\c}P_\g^\c f.\]
Since $\sqrt{\mu_\c} P^\c_\g f= \sqrt{2\pi}\mu z_\g(f)$, we need to compute $\int_{\R^3} dv \int_{\R^3} dv (n\cdot v) (v_{\c,i}^2-\beta_b)\mu(v)$. We have
\[ v_{\c,i}^2-\beta_b=v_i^2-\beta_b -2\e \c_iv_i+ \e^2\c_i^2.\] 
The terms of order $1$ and $\e^2$ vanish by oddness. Therefore
$$\int_\g P_\g^\c \dd\g f \psi_{b}^{i,j}=-\e\int_{\pt \O} \dd S(x)\nabla \f_{b}^j\cdot \int_{\R^3} \dd v (n\cdot v)2\c_i v_i\sqrt{\mu_\c}P_\g^\c f.$$
Thus,  {by using Lemma \ref{traces},}
\be\Big|\int_\g\dd\g P_\g^\c f \psi_{b}^{i,j}\Big|\le 
\e|\c| \e^{-\frac 1 2}[\e^{\frac 1 2}\|f\|_\infty]\|b_j\|_6^5\lesssim \e^{ {\frac 1 2}}|\c|\|b_j\|_6^6+ \e|\c|\e^{-\frac 1 2}[\e^{\frac 1 2}\|w f\|_\infty]^6.\label{pgammapsibj}\ee
Next, we need to bound $\int_\g \dd\g\1_{\g_+}(1-P_\g^\c) f \psi_{b}^{i,j}$. We have
$$\Big|\int_\g \dd\g\1_{\g_+}(1-P_\g^\c) f \psi_{b}^{i,j}\Big|\le \| \nabla_{x} \f_{b}^j \|_{L^{4/3}(\pt\O)} \|
\1_{\g_+}(1-P_\g^\c) f\|_{L^4(\g)}.$$
Thus, we conclude that, for any $\eta>0$ and $\eta'>0$
\be \Big|\int_\g\dd\g \1_{\g_+}(1-P_\g^\c) f \psi_{b}^{i,j}\Big|\lesssim \eta \|b_j\|_6^6 + \eta'[\e^{\frac 1 2}\| f\|_{\infty}]^6+C_{\eta, \eta'}[\e^{-\frac 1 2}\|
\1_{\g_+}(1-P_\g^\c) f\|_{L^2(\g)}]^{6}\Big\}.\ee
In conclusion, for any $\eta>0$, $\eta'>0$, by
$$\Big|\int_{\gamma}\dd\g\psi_b^{i,j}
[P_\g^\c f +\1_{\g_+}(1-P_\g^\c)f]\Big|
\lesssim \eta\|b_j\|^6_6 + \eta'[\e^{\frac 12}\|f\|_\infty]^6+C_{\eta,\eta'}[\e^{-\frac 1 2}\|(1-P_\g^\c)f)|_{2.\g_+}]^6.
$$
Finally, 
$$\Big|\int_{\g_-} \dd\g\,\e^{\frac 1 2} r \psi_{b}^{i,j}\Big |\lesssim \|\nabla\f_{b}^j\|_{L^{4/3}(\pt\O)} \|\e^{\frac 1 2} r\|_{L^{4}(\pt\O)}\le \e^{\frac 1 2}\|b\|_6^5  |r|_\infty.$$
By collecting all the terms and choosing $\eta$ and $\eta'$ sufficiently small we conclude that
\begin{multline} 
\Big|\int_{\O^c}\dd x\, b_i\pt_i\partial_j \f_{b}^{j}\Big|\lesssim 
(\e^{-1}\|\ipc f\|_\nu)^6 +\|\ipc f\|_6^6+
\| g{{\nu}^{-\frac 1 2}}\|_2^6+ \e|\c|\|\P_\c f\|^6_6\\+\eta \|b_j\|_6^6 +(\e^{-\frac 1 2}|(1-P_\g^\c)f|_{2,+})^6+ (\e^{\frac 12} |r|_\infty)^6+ o(1)[\e^{\frac 1 2}\|f\|_\infty]^6.\label{bfirst}
\end{multline}
To estimate $\partial _{j}(\partial _{j}\Delta ^{-1}b^{k-1}_{i})b_{i}$ for $
i\neq j $, we choose as test function \begin{equation}
\bar\psi_b^{i,j} =|v_\c|^{2}v_{\c,i}v_{\c,j}\sqrt{\mu_\c }\partial _{j}\f_{b}^{i}(x),\quad i\neq
j.  \label{phibij}
\end{equation}
We have:
$$\int_{\O^c\times \R^3} \dd x \dd v\, f v\cdot \nabla \bar\psi _{b}^{i,j}=\int_{\O^c\times \R^3} \dd x \dd v\,f v_\c\cdot \nabla \bar\psi _{b}^{i,j}+\e\int_{\O^c\times \R^3} \dd x \dd v\, f \c\cdot \nabla \bar\psi _{b}^{i,j}.$$
By (\ref{nablabj}) and the Young inequality,
$$\Big|\e\int_{\O^c\times \R^3} \dd x \dd v\, f \c\cdot \nabla \bar\psi _{b}^{i,j}\Big|\le \e|\c|\|b_j\|_6^5\|f\|_6\le \e|\c|\|\P_\c f\|^6_6+\e|\c|\|\ipc f\|_6^6.$$
By using $f=\P_\c f+\ipc f$ and the expression of $\P_\c f$, we need to compute 
\be\int_{\O^c\times \R^3} \dd x \dd v\, a \sqrt{\mu_\c} v_\c\cdot \nabla \bar\psi _{b}^{i,j},\label{a1b2}\ee 
\be
\int_{\O^c\times \R^3}  \dd x \dd v\, b\cdot v_\c \sqrt{\mu_\c} v_\c\cdot \nabla \bar\psi _{b}^{i,j},\label{b1b2}\ee
\be\int_{\O^c\times \R^3} \dd x \dd v\,  c \frac{|v_\c|^2-3}2 \sqrt{\mu_\c} v_\c\cdot \nabla \bar\psi _{b}^{i,j},\label {c1b2}\ee
\be\int_{\O^c\times \R^3} \dd x \dd v\, v\cdot \nabla \bar\psi _{b}^{i,j} \ipc f=\int_{\O^c\times \R^3} \dd x \dd v\, \sqrt{\mu_\c}(|v_\c|^2-\beta_c)^2v_\c\otimes v_\c:\nabla\otimes\nabla\f_{b}^j\ipc f.\ee
Using (\ref{nablabj}), the last one is bounded by 
$$\|\ipc\|_6 \|{b_j}\|^5_{L^6(\O^c)}\le \frac 5 6\eta \|{b_j}\|^6_{L^6(\O^c)}+\frac 1 6 \eta^{- {\frac 15}} \|\ipc\|_6^6,$$
for any $\eta>0$.

For $j\neq i$, the $O(\c)$ terms in (\ref{a1b2}), (\ref{b1b2}) and (\ref{c1b2}) vanish by oddness in $v_{\c,i}$. for the same reason the terms of order $1$ in  (\ref{a1b2}) and (\ref{c1b2}) vanish. The only surviving term is 
$$\sum_{k,\ell}   b_k \pt_\ell\pt_j\f_b^{i}\int_{\O^c\times \R^3}\dd x \dd v\,\mu_\c v_{\c,k}  v_{\c,\ell}v_{\c,i}  v_{\c, {j}}|v_{\c}|^2=21 \int_{\O^c}\dd x (b_j\pt_i\pt_j\f_b^i +b_i\pt_j^2\f_b^i),$$
because
$$\int_{\R^3}\dd v\mu_\c v_{\c,k}  v_{\c,\ell}v_{\c,i} v_{\c, {j}}|v_{\c}|^2=21(\d_{k,\ell}\d_{i,j}+\d_{k,i}\d_{\ell,j}+\d_{k,j}\d_{\ell,i}).$$
By taking the sum on $j$ this reduces to 
$\int_{\O^c} dx (b_i^6 + \sum_j b_j\pt_j\pt_i\Delta^{-1}b_i)$. The second term has been bounded in (\ref{bfirst}), thus, to complete the estimate of $\|b\|_6$ we just need to bound the remaining terms in the weak equation (\ref{weak}) for $\psi=\bar\psi_b^{i,j}$. 
As before, we have 
\[\e^{-1}\Big|\int _{\O^c\times \R^3} \dd x \dd v\, \bar \psi_b^{i,j} L_\c f\Big|\le \e^{-1}\|\nabla\f_{b}^j\|_{L^2(\O^c)}\|\ipc f\|_\nu\le \frac 56 \eta \|b_j\|_6^6 +\frac 1 6 {\eta^{-\frac 1 5}} [\e^{-1}\|\ipc f\|_\nu]^6,\]
and
\[\Big|\int _{\O^c\times \R^3} \dd x \dd v\, \bar\psi_b^{i,j} g\Big|\lesssim \|\nabla\f_{b_j}\|_{L^2(\O^c)}\| g{{\nu}^{-\frac 1 2}}\|_2\le \frac 56 \eta \|b_j\|_6^6 +\frac 1 6 {\eta^{-\frac 1 5}}\| g{{\nu}^{-\frac 1 2}}\|_2,\]
for any $\eta>0$. 
Finally, expanding, we have
\begin{multline*}
|v_{\c}|^{2}v_{\c,i}v_{\c,j}=|v|^2 v_iv_j+\e[|v|^2(\c_iv_j+\c_jv_i)-2\c\cdot vv_iv_i]\\+\e^2[|\c|^2v_iv_j ||v|^2\c_i\c_j-2\c\cdot v(\c_iv_j+\c_jv_i)]+\e^3[|\c|^2(\c_jv_j+\c_jv_i-2\c\cdot v\c_i\c_j] +\e^4|\c|^2\c_i\c_j.
\end{multline*}
Therefore in the contribution from $P_\g^\c f$ the term of order $0$ in $\e$ gives a vanishing contribution. Therefore, as before
\be\Big|\int_\g\dd\g P_\g^\c f \bar\psi_{b}^{i,j}\Big|\le 
\e|\c| \e^{-\frac 1 2}[\e^{\frac 1 2}\|f\|_\infty]\|b_j\|_6^5\lesssim \e^{ {\frac 1 2}}|\c|\|b_j\|_6^6+ \e|\c|\e^{-\frac 1 2}[\e^{\frac 1 2}\|w f\|_\infty]^6.\label{pgammapsibjbar}\ee
Moreover 
$$\Big|\int_\g \dd\g\1_{\g_+}(1-P_\g^\c) f \bar\psi_{b}^{i,j}\Big|\le \| \nabla_{x} \f_{b}^j \|_{L^{4/3}(\pt\O)} \|
\1_{\g_+}(1-P_\g^\c) f\|_{L^4(\g)}.$$
By collecting the previous bounds we conclude that
\begin{multline}\|b\|_6^6\lesssim (\e^{-1}\|\ipc f\|_\nu)^6 +\|\ipc f\|_6^6+
\| g{{\nu}^{-\frac 1 2}}\|_2^6+(\e^{-\frac 1 2}|(1-P_\g^\c)f|_{2,+})^6\\+\e|\c|\|\P_\c f\|_6+ (\e^{\frac 1 2} |r|_\infty)^6+ o(1)[\e^{\frac 1 2}\ f\|_\infty]^6.\end{multline}

\noindent\underline{Step 3}:

\noindent Then we bound $\|a\|_6$. The argument is similar to the one used for $c$, the only main difference being in the treatment of the boundary terms.
\begin{equation}
\psi =\psi _{a}\equiv (|v_{\c}|^{2}-\beta _{a})v_{\c}\cdot \nabla _{x}\f _{a}\sqrt{
\mu }=\sum_{i=1}^{d}(|v_\c|^{2}-\beta _{a})v_{\c,i}\partial _{i}\f_{a}\sqrt{\mu }
,  \label{phia}
\end{equation}
where 
\begin{equation}
-\Delta _{x}\f _{a}(x)=a^5,\quad\frac{\partial }{
\partial n}\f _{a}|_{\partial \Omega }=0, \label{Deltaa}
\end{equation}
whose solution satisfies
\be \|\nabla\f_a\|_{\dot H^1(\O^c)}\le \|\/\/|a|^5\|_{L^{\frac 6 5}(\O^c)}=\|a\|_{L^6(\O^c)}^5.\label{nablaa}\ee
We have
$$\int_{\O^c\times \R^3}\dd x \dd v\, f v\cdot \nabla \psi _a=\int_{\O^c\times \R^3} \dd x \dd v\,f (v-\e\c)\cdot \nabla \psi _a+\e\int_{\O^c\times \R^3} \dd x \dd v\, f \c\cdot \nabla \psi _a.$$
By (\ref{nablaa}) and the Young inequality,
$$\Big|\e\int_{\O^c\times \R^3} \dd x \dd v\, f \c\cdot \nabla \psi _a\Big|\le \e|\c|\|b_j\|_6^5\|f\|_6\le \e|\c|\|\P_\c f\|^6_6+\e|\c|\|\ipc f\|_6^6.$$

Proceeding as before, by using $f=\P_\c f+\ipc f$ and the expression of $\P_\c f$, we need to compute 
\be\int_{\O^c\times \R^3} \dd x \dd v\, a \sqrt{\mu_\c} v_{\c}\cdot \nabla \psi_a,\label{a1a}\ee 
\be
\int_{\O^c\times \R^3}  \dd x \dd v\, b\cdot v_{\c} \sqrt{\mu_\c} v_{\c}\cdot \nabla \psi_a,\label{b1a}\ee
\be\int_{\O^c\times \R^3}  \dd x \dd v\, c \frac{|v_{\c}|^2-3}2 \sqrt{\mu_\c} v_{\c}\cdot \nabla \psi_a,\label {c1a}\ee
\be\int_{\O^c\times \R^3} \dd x \dd v\, v\cdot \nabla \psi_a \ipc f=\int_{\O\times \R^3} \dd x \dd v\, \sqrt{\mu_\c}(|v_{\c}|^2-\beta_a)^2v_{\c}\otimes v_{\c}:\nabla\otimes\nabla\f_a\ipc f.\ee

Using (\ref{nablaa}), by the Young inequality, the last one is bounded by 
$$\|\ipc f\|_6 \|c\|^5_{L^6(\O^c)}\le \frac 5 6\eta \|a\|^6_{L^6(\O^c)}+\frac 1 6 \eta^{- {\frac 1 5}} \|\ipc\|_6^6,$$
for any $\eta>0$.

With the choice $\beta_a=10$ 
\be \int_{\R^3} \dd v (|v_{\c}|^2-\beta_a) (|v_{\c}|^2-3)v_{\c}\otimes v_{\c}=0,\label{betacc}\ee and the term in (\ref{c1a}) vanishes. The term of (\ref{b1a}) vanishes for the same reason. 

Now we compute the term in (\ref{a1a}):
we have 
\begin{multline*}\int_{\O^c\times \R^3} \dd x \dd v\,  a \sqrt{\mu_\c} v_{\c}\cdot \nabla \psi_a=\int_{\O^c}\dd x  a \nabla\otimes \nabla \f_a :\int_{\R^3}\dd v  v_{\c}\otimes v_{\c}v_{\c}(|v_{\c}|^2-\beta_a) \mu_\c\\=-5\int_{\O^c}\dd x  a \Delta \f_a=5\|a\|_{L^6(\O^c)}^6,
\end{multline*}
because of (\ref{Deltaa}). We have used 
\be\int_{\R^3}\dd x \dd v\,  v_{\c}\otimes  v_{\c}(|v_{\c}|^2-\beta_a) \mu_\c=-5\mathbf{I}.\label{betac01}\ee
As for the boundary term, we have 
\[\int_\g \dd\g P_\g^\c f \psi_a=\int_{\pt\O}\dd S\,  z_\g\nabla\f_a\cdot\int_{\R^3} \dd v \mu (v-\e\c)(|v-\e\c|^2-\beta_a)n\cdot v\]
But 
\begin{multline*}\int_{\R^3} \dd v \mu v_{\c}(|v_{\c}|^2-\beta_a)n\cdot v=\int_{\R^3} \dd v \mu v_{\c}(|v_{\c}|^2-\beta_a)n\cdot v_{\c}+\e\int_{\R^3} \dd v \mu v_{\c}(|v_{\c}|^2-\beta_a)n\cdot \c.\end{multline*}
The second term vanishes by oddness. The first by oddness is 
$$\int_{\R^3} \\d v \mu v_{\c,i}(|v_{\c}|^2-\beta_a) n\cdot v_{\c}=n_i \int_{\R^3} \dd v \mu |v_{\c}\cdot n|^2(|v_{\c}|^2-\beta_a)=-5 n_i.$$ 
Therefore 
$$\int_\g\dd \g P_\g^\c f \psi_a=\int_{\pt\O}dS z_\g n\cdot \nabla\f_a=0,$$
by the Neumann boundary condition on $\f_a$.
The term $\int_\g \dd \g\1_{\g_+}(1-P_\g^\c) f \psi_a$ is estimated as the similar term for $c$.
By collecting the estimate, we conclude that
\begin{multline} 
\|a\|_6\lesssim 
\e^{-1}\|\ipc f\|_\nu +\|\ipc f\|_6+
\| g{{\nu}^{-\frac 1 2}}\|_2+  {(\e|\c|)^{\frac 1 6}}\|\P_\c f\|_6+\e^{-\frac 1 2}|(1-P_\g^\c)f|_{2,+}\\+ \e^{\frac 1 2} |r|_\infty+ o(1)[\e^{\frac 1 2}\| f\|_\infty].
\end{multline}
In conclusion, for $|\c|$ small,
$$ \|\P_\c f\|_6\lesssim \e^{-1}\|\ipc f\|_\nu +\|\ipc\|_6+
\| g{{\nu}^{-\frac 1 2}}\|_2+\e^{-\frac 1 2}|(1-P_\g^\c)f|_{2,+}+ \e^{\frac 1 2} |r|_\infty+ o(1)[\e^{\frac 1 2}\| f\|_\infty] 
.$$\end{proof}
\section{Balance laws} \label{eqbilancio}
The mass, momentum and energy balance equations are obtained by projecting \eqref{linprob0} on the null space of $L_\c$. Since $\P_\c L_\c=0$, we have:
\be \P_\c(v\cdot \nabla f)=\P_\c g.\label{consR}\ee
More explicitly, we write $\P_\c g=(\ag+\bg\cdot v_\c +\frac 1 2(|v_\c|^2-3) \cg)\sm_\c$, 
and $\P_\c f= [a+b\cdot v_\c+\frac 1 2(|v_\c|^2-3)c]\sm_\c$.  We have
\begin{eqnarray} &&\nabla\cdot b+\e\c \cdot \nabla a=\ag,\label{anocut}
\\&&
\nabla P+\e \c \cdot \nabla b+\nabla\cdot \tau=\bg,
\label{Pnocut}
\\&&
\nabla \cdot b+\frac 3 2\e\c\cdot\nabla c+\nabla\cdot \mathfrak{q}=\cg,
\label{cnocut}
\end{eqnarray}
where
\be \tau=\int_{\R^3} \dd v\, v_\c\otimes v_\c\sqrt{\mu_\c}\ipc f,\label{stress}\ee
\be \mathfrak{q}=\int_{\R^3} \dd v\, \frac{|v_\c|^2-3}2v_\c\sqrt{\mu_\c}\ipc f,\label{heatflux}\ee
\be P=a+c.\ee
We have to complete  equations \eqref{anocut}, \eqref{Pnocut}, \eqref{cnocut} with boundary conditions following from \eqref{linprob0}, which are not immediately translated into conditions on $a, b,c$. 
Therefore, as in \cite{Masl}, we  introduce a smooth cutoff function 
$$\z(x)=\begin{cases} 1 \text{\hskip .3cm if }x\in \R^3\backslash\O\text{ and }d(x,\O)>1\\ 0\text{\hskip .3cm if } x\in \overline{\O}\end{cases}$$  and define $f^\z= \z f$ extended as $0$ in $\O$. If $f$ solves the problem \eqref{linprob0}, then  $f^\z$ solves the equation
 \be v\cdot \nabla f^\z +\e^{-1}L_\c f^\z= \z g+ \mathcal{C}\quad \text{ in \hskip .3cm} \R^3,\label{eqfzeta}\ee 
where \be \mathcal{C}=fv\cdot \nabla \z.\label{gzeta}\ee
By projecting the equation for $f^\z$ on the null space of $L_\c$ we obtain the balance laws
$$\P_\c(v\cdot \nabla  f^\z)=\P_\c \mathcal{C}+\z\P_\c  g,$$
More explicitly, with $\P_\c f^\z= [a^\z+b^\z\cdot v_\c+c^\z(|v_\c|^2-3)/2]\sqrt{\mu}_\c$ 
and $ P^\z=a^\z+c^\z$, we have, 
\begin{eqnarray} &&\nabla\cdot b^\z+\e\c \cdot \nabla a^\z=\z \ag+\int_{\R^3} \dd v\,\mathcal{C} ,\label{acut}
\\&&
\nabla P^\z+\e \c \cdot \nabla b^\z+\nabla\cdot \tau^\z=\z \bg+\int_{\R^3} \dd v\, \mathcal{C}\/v\sm_\c\label{Pcut}
\\&&
\nabla \cdot b^\z+\frac 3 2\e\c\cdot\nabla c^\z+\nabla\cdot q^\z=\z \cg+ \int_{\R^3}\dd v\,\frac 1 2\mathcal{C}\/  (|v|^2-3)\sm_\c ,
\label{ccut}
\end{eqnarray}
where
\be \tau^\z=\int_{\R^3} \dd v\, v_\c\otimes v_\c\sqrt{\mu_\c}\ipc f^\z,\label{stress1}\ee
\be q^\z=\int_{\R^3} \dd v\, \frac{|v_\c|^2-3}2v_\c\sqrt{\mu_\c}\ipc f^\z,\label{heatflux1}\ee
and $P^\z=a^\z+c^\z$.

It is convenient to write above equations in the Fourier space:
The Fourier transform is normalized as  
\be \hat f (k)=\mathcal{F}_x(f)(k)=\frac {1}{(2 \pi)^{\frac 32}}
\int_{\R^3} \dd x f(x) e^{i k\cdot x}.
\ee
We have  
\be ik\cdot v \hat f^\z +\e^{-1} L_\c f^\z=\widehat{\z g}+\hat{\mathcal{C}},\label{eqfzetak}\ee
By writing \be\hat f^\z= (\hat a^\z+\hat b^\z\cdot v_\c+ \frac 1 2\hat c^\z (|v_\c|^2-3))\sm_\c +\ipc \hat f^\z,\label{decompfou}\ee
the projection on $\text{Null}\/\/ L_\c$ is  
\begin{eqnarray}&& ik\cdot \hat b^\z +i\e k\cdot \c \hat a^\z =\int_{\R^3} \dd v\, \sm_\c \hat {\mathcal{C}}(k,v)+\widehat{\z \ag},\label{massk}\\
&& i k \hat P^\z+ i k\cdot \hat \tau^\z +i\e\c\cdot k\hat b^\z= \int_{\R^3} \dd v\, v_\c\sm_\c \hat {\mathcal{C}}(k,v)+\widehat{\z \bg},\label{momenk}\\
 &&i k \cdot \hat b^\z +\frac 3 2i\e k\cdot\c \hat c^\z+ i k\cdot \mathfrak{q}^\z= \int_{\R^3} \dd v\, \frac 1 2(|v_\c|^2-3)\sm \hat {\mathcal{C}}(k,v)+\widehat{\z \cg}.\label{enerk}\end{eqnarray}
Let
\be \mathscr{B}_\c=L_\c^{-1}[(v_\c\otimes v_\c-\frac 1 3|v_\c|^2\mathbf{I})\sm_\c],\quad \mathscr{A}_\c=L_\c^{-1}[\frac 1 2(v_\c(|v_\c|^2-5)\sm_\c)]\label{defab}\ee
The momentum equation \eqref{momenk}   then becomes
\be i k \hat P^\z +i\e\c\cdot k\hat b^\z+i k\cdot \int_{\R^3}\dd v\, L_\c\hat f^\z \mathscr{B}_\c= \int_{\R^3}\dd v\, v_\c\sm_\c \hat{\mathcal{C}}(k,v )+\widehat{\z \bg},\ee
and the energy equation \eqref{enerk}   becomes
\be i k\cdot  \hat b^\z+\frac 3 2i\e k\cdot\c \hat c^\z+ i k\cdot \int_{\R^3} \dd v\, L_\c\hat f^\z \mathscr{A}_\c= \int_{\R^3} \dd v\, \frac 1 2(|v_\c|^2-3)\sm_\c \hat{\mathcal{C}}(k,v)+\widehat{\z \cg}\ee
Substituting from the equation \eqref{eqfzetak} , $L_\c \hat f^\z= -\e ik\cdot v \hat f^\z +\e (\widehat{\z g}+\hat{\mathcal{C}})$,
\be i k (\hat a^\z +\hat c^\z)+i\e\c\cdot k\hat b^\z+  i k\cdot \int_{\R^3} \dd v\, [-i\e k\cdot v \hat f^\z      +\e (\widehat{\z g}+\hat{\mathcal{C}})] \mathscr{B}_\c=  {\widehat{\z \bg}+}\int_{\R^3} \dd v\, v_\c\sm_\c \hat{\mathcal{C}}(k,v)\ee
\be i k\cdot \hat b^\z+i\e \frac 3 2k\cdot\c \hat c^\z+  i k\cdot \int_{\R^3} \dd v\, [-i\e k\cdot v \hat f^\z +\e (\widehat{\z g}+\hat{\mathcal{C}}) \mathscr{A}_\c=  {\widehat{\z \cg}+}\int_{\R^3} \dd v\, \frac 1 2(|v_\c|^2-3)\sm_\c \hat{\mathcal{C}}(k,v).\ee
Using again (\ref{decompfou}),  the term $\int dv  \hat f^\z v \cdot \mathscr{B}_\c$ becomes
\begin{multline} \int_{\R^3} \dd v\,  \hat f^\zeta v\cdot \mathscr{B}_\c=\int_{\R^3} \dd v\,  v_\c \cdot (\hat a^\z+\hat b^\z\cdot v_\c+ \hat c^\z (|v_\c|^2-3)/2)\sm_\c\mathscr{B}_\c+\int_{\R^3} \dd v\,  v \cdot\mathscr{B}_\c\ipc \hat f^\z
\\+\e\int_{\R^3} \dd v\,  \c \cdot\mathscr{B}_\c (\hat a^\z+\hat b^\z\cdot v_\c+ \frac 1 2\hat c^\z (|v_\c|^2-3))\sm_\c.
\end{multline}
The second line vanishes because $\P_\c \mathscr{B}_\c=0$. From the properties of $\mathscr{B}_\c$, only the $\hat b$ term survives of the first part of first line. Since, again $\ipc(v_\c\otimes v_\c) \sm_\c= L_\c\mathscr{B}_\c$, we obtain
\be \int_{\R^3} \dd v\, \hat f^\z v \cdot\mathscr{B}_\c=\hat b^\z \int_{\R^3} \dd v\, \mathscr{B}_\c L_\c \mathscr{B}_\c+\int_{\R^3} \dd v\,  v\cdot \mathscr{B}_\c \ipc \hat f^\z.\ee
We have (see e.g. \cite{Ce})
$$-\int dvL_u\mathscr{B}_\c\otimes\mathscr{B}_\c\cdot  k \ k \cdot\hat b^\z\sm= \mathbb{K}b^\z$$
 with 
 \be\mathbb{K}=\mathfrak{v}[|k|^2 \mathbf{I}+\frac 2 3 k \otimes k],\label{defK}\ee with $\mathfrak{v}>0$ the viscosity coefficient. Note that for any  $y\in\mathbb{R}^3$ with $|y|=1$, 
 \be  y\cdot\mathbb{K}y\sim \mathfrak{v}|k|^2, \quad y\cdot\mathbb{K}^{-1}y\sim \mathfrak{v}^{-1}|k|^{-2}\label{bK} \ee We obtain:
\begin{multline} i k (\hat a^\z +\hat c^\z)+i\e\c\cdot k\hat b^\z+\e {\mathbb{K} b^\z}
+\e k\otimes k \cdot \int_{\R^3} \dd v\, v  \mathscr{B}_\c\ip \hat f^\z+\e i k \cdot\int (\widehat{\z g}+\hat{\mathcal{C}})\mathscr{B}_\c\\=  {\widehat{\z \bg}+}\int_{\R^3} \dd v\, v_\c\sm_\c \hat {\mathcal{C}} (k,v).\end{multline}
Similarly, since $\ipc [v_\c(|v_\c|^2-3)/2\sm_\c]= L_\c \mathscr{A}_\c$,
\begin{multline} i k \cdot \hat b^\z+i\e\frac 3 2 k\cdot\c \hat c^\z+\e\ \kappa|k|^2 \hat c^\z +\e k\otimes k \cdot \int_{\R^3} \dd v\, v \ipc \hat f^\z \mathscr{A}_\c+\e i k \cdot\int_{\R^3} \dd v\, (\widehat{\z g}+\hat{\mathcal{C}}) \mathscr{A}_\c\\=  {\widehat{\z \cg}+}\int_{\R^3} \dd v\, \frac 1 2(|v_\c|^2-3)\sm_\c \hat{\mathcal{C}}(k,v),\end{multline}
with $\kappa=\int_{\R^3} \dd v\, \mathscr{A}L\mathscr{A}$.
Therefore  the balance laws in the Fourier space are
\begin{eqnarray}
&&ik\cdot \hat b^\z +i\e k\cdot \c \hat a^\z=\hat s_0,\label{massFs}\\
&&i k  \hat P^\z +\e
 {\mathbb{K} b^\z}
+i\e\c\cdot k\hat b^\z=\hat{\underline{s}}\label{momenFs}\\
&& i k \cdot \hat b^\z+\e\ \kappa|k|^2+\frac 3 2ik\cdot\c ] \hat c^\z=\hat s_4\label{enerFs},\end{eqnarray}
where the transport coefficient $\kappa$
is
defined by$\kappa=\int dv \mathscr{A}_\c L^{-1}_\c\mathscr{A}_\c$ and the source terms are
\begin{eqnarray}\hat s_0&=&\int_{\R^3} \dd v\, \sm  \hat{\mathcal{C}}(k,v)+\widehat{\z \ag},\notag\\
\hat{\underline{s}}\ &=& {-}\e k \otimes k \cdot \int_{\R^3} \dd v\, v \ipc  \hat f^\z \mathscr{B}_\c-i\e k \cdot\int_{\R^3} \dd v\, (\widehat{\z g}+ \hat{\mathcal{C}}) \mathscr{B}_\c+\int_{\R^3} \dd v\, v_\c\sm_\c  \hat{\mathcal{C}}(k,v)\notag\\&+&\widehat{\z \bg},\label{sourcesfou}\\
\hat s_4&=& {-}\e k\otimes k \cdot \int_{\R^3} \dd v\, v \ipc  \hat f^\z \mathscr{A}_\c-i\e k \cdot\int_{\R^3} \dd v\, (\widehat{\z g}+ \hat{\mathcal{C}}) \mathscr{A}_\c+ \int_{\R^3} \dd v\, \frac 1 2(|v_\c|^2-3)\sm_\c  \hat{\mathcal{C}}(k,v)\notag\\&+&\widehat{\z \cg}.\notag\end{eqnarray}

To eliminate the pressure $\hat P^\z$ from (\ref{momenFs}) we apply the Leray projector $\Pi$ defined, in Fourier space, by 
$$\hat \Pi= \mathbf{I} - \frac{k\otimes k}{|k|^2}.$$
We use the short notation
$${N}_{\s,\b}(k)=\begin{cases}\e[
% {[|k|^2 b+\frac 2 3 k \,k\cdot b]]} 
 {\mathbb{K}}+\b i\c\cdot k {\mathbf{I}}]\quad \text{for } \s=\mathfrak{v} \text{ and }\b=1\\
\e[{\frac{\kappa}2} |k|^2+\frac 5 2  i\c\cdot k]\quad\text{for } \s= \frac{\kappa}2 \text{ and } \beta =\frac 5 2
.\end{cases}$$ 
Thus we get 
\be\hat\Pi\hat b^\z=N_{\mathfrak{v}, 1}^{-1}\hat\Pi\hat{\underline{s}}.\label{momenFPi}\ee
Then we multiply the momentum equation by $k$ and divide by $i|k|^2$ to obtain
$$\hat P^\z+\frac{N_{\mathfrak{v},1}}{i|k|^2}\hat b^\z\cdot k=\frac{k}{i|k|^2}\cdot \hat{\underline{s}}.$$
From the mass equation we have
\be \hat b^\z\cdot k= -i\hat s_0  {-}\e\c\cdot k \hat a^\z.\label{massa}\ee
Hence $$\hat P+\frac{{N}_{\mathfrak{v},1}}{i|k|^2}(-i\hat s_0  {-}\e\c\cdot k \hat a)=\frac{k}{i|k|^2}\cdot \hat{\underline{s}},$$
and reminding that $\hat a=\hat P-\hat c$, we have, for $|\c|$ 
sufficiently small,
\be\hat P^\z=\Big(1 {-}
\e\frac{{N}_{\mathfrak{v},1}}{i|k|^2}\c\cdot k
\Big)^{-1}\Big[\frac{{N}_{\mathfrak{v},1}}{i|k|^2}[i\hat s_0 -\e\c\cdot k \hat c^\z]+\frac{k}{i|k|^2}\cdot \hat{\underline{s}}\Big].\label{press}\ee
Subtracting the mass equation from the energy equation and using $\hat a^\z=\hat P^\z-\hat c^\z$,
the equation for $\hat c^\z$ becomes
\be({N}_{\kappa, \frac 5 2}) \hat c^\z-i\e\c\cdot k \hat P^\z=\hat s_4-\hat s_0.\label{cener}\ee
Replacing the expression of the pressure
we obtain
\be \hat c^\z=(\overline{{N}})^{-1}\Big\{\hat s_4-\hat s_0 {+}i\e\c\cdot k \Big(1 {-}
\e\frac{{N}_{\mathfrak{v},1}}{i|k|^2}\c\cdot k
\Big)^{-1}\Big[\frac{{N}_{\mathfrak{v},1}}{i|k|^2}i\hat s_0 +\frac{k}{i|k|^2}\cdot \hat{\underline{s}}\Big]\Big\}\label{energia}\ee
 with
\be\overline{{N}}=N_{\kappa, \frac 5 2}+i\e^2(\c\cdot k)^2\frac{{N}_{\mathfrak{v},1}}{i|k|^2} \Big(1 {-}
\e\frac{{N}_{\mathfrak{v},1}}{i|k|^2}\c\cdot k
\Big)^{-1}.\ee
Then $\hat a^\z=\hat P^\z-\hat c^\z$ is obtained by subtracting the expressions of $\hat P^\z$ and $\hat c^\z$ just obtained. 
Finally, using \eqref{massa} we compute $(1-\hat\Pi \hat b^\z)$.

\numberwithin{equation}{subsection}
\section{Estimate of $\|\P_\c f\|_3$} \label{stimaPc3}
\subsection{Splitting of $\P_\c f$}\label{split}   

\medskip
We define the small $k$'s cutoff as a smooth function
\be
\mathfrak{j}=\begin{cases} 1& \text{for }|k|<1\\0& \text{for } |k|>2\end{cases},\ee
and \be 
\mathfrak{j}^c=1-\mathfrak{j}.\ee

We will split the source terms $\mathbf{s}=(s_0,\underline{s},s_4)$ into five different contributions $s^{(i)}= (s_0^{(i)},\underline{s}^{(i)},s_4^{(i)})$, for $i=1,\dots, 5$:
\be \mathbf{s}=\sum_{i=1}^5 \mathbf{s}^{(i)},\ee
The source $\mathbf{s}^{(1)}$ corresponds large $k$'s:
\be\hat {\mathbf{s}}^{(1)}(k)=\mathfrak{j}^c\hat {\mathbf{s}}.\label{s1}\ee
Then we split $\hat{\mathcal{C}}(k,v)=\mathcal{F}[ f\/v\cdot  {\nabla}\z](k,v)$ as 
\be \hat{\mathcal{C}}= \hat{\mathcal{C}}_{s}+k\cdot \hat{\mathcal{C}}_{r},\ee
with
\be \hat{\mathcal{C}}_{s}(v)= \hat{\mathcal{C}}(0,v),\ee
and
\be \hat{\mathcal{C}}_{r}(k,v)=\int_0^1 \dd\l \nabla_k \hat{\mathcal{C}}(\l k,v),\label{560}\ee
so that
\be  \hat{\mathcal{C}}(k,v)- \hat{\mathcal{C}}(0,v) = \int_0^1 \dd\l\frac d{d\l}\hat{\mathcal{C}}(\l k,v)= \int_0^1 \dd\l\, k\cdot \nabla_k \hat{\mathcal{C}}(\l k,v)=k\cdot \hat{\mathcal{C}}_{r}(k,v).\ee
We set 
\begin{eqnarray}
\hat s^{(2)}_0(k)&=&\mathfrak{j}\int_{\R^3} \dd v\, \sm_\c  \hat{\mathcal{C}}_s(0,v),\notag\\
\hat{\underline{s}}^{(2)}(k)&=&\mathfrak{j}\int_{\R^3} \dd v\, v_\c\sm_\c  \hat{\mathcal{C}}_s(0,v),\\
\hat s_4^{(2)}(k)&=& \mathfrak{j}\int_{\R^3} \dd v\, \frac 1 2(|v_\c|^2-3)\sm_\c  \hat{\mathcal{C}}_s(0,v).\notag\\\notag\\
\hat s^{(3)}_0(k)&=&\mathfrak{j}\Big[\int_{\R^3} \dd v\, \sm_\c  k\cdot \hat{\mathcal{C}}_r(k,v)\Big],\notag\\
\hat{\underline{s}}^{(3)}(k)&=&\mathfrak{j}\Big[ {-}\e k \otimes k \cdot \int_{\R^3} \dd v\, v \ipc  \hat f^\z(k,v) \mathscr{B}_\c-i\e k \cdot\int_{\R^3} \dd v\, \hat{\mathcal{C}} \mathscr{B}_\c\\&&\hskip 5cm+ k\cdot\int_{\R^3} \dd v\, v_\c\sm_\c   \hat{\mathcal{C}}_r(k,v)\Big],\notag\\
\hat s_4^{(3)}(k)&=&\mathfrak{j}\Big[ {-}\e k\otimes k \cdot \int_{\R^3} \dd v\, v \ipc  \hat f^\z(k,v) \mathscr{A}_\c-i\e k\cdot\int_{\R^3} \dd v\, \hat{\mathcal{C}} \mathscr{A}_\c\notag\\&&\hskip 4cm +k\cdot\int_{\R^3} \dd v\, \frac 1 2(|v_\c|^2-3)\sm_\c  \hat{\mathcal{C}}_r(x,v)\Big].\notag
\end{eqnarray}
\begin{eqnarray}
\hat s_0^{(4)}(k)&=&0,\notag\\
\hat{\underline{s}}^{(4)}(k)&=&-\mathfrak{j}i\e k \cdot\int_{\R^3} \dd v\, \widehat{\z g}\mathscr{B}_\c,\\
\hat s_4^{(4)}(k)&=&-\mathfrak{j}\e ik\cdot \int_{\R^3} \dd v\, \widehat{\z g} \mathscr{A}_\c.\notag\\
\notag\\\hat s_0^{(5)}(k)&=&\mathfrak{j}\widehat{\z \ag},\notag\\
\hat{\underline{s}}^{(5)}(k)&=&\mathfrak{j}\widehat{\z \bg},\\
\hat s_4^{(5)}(k)&=&\mathfrak{j}\widehat{\z \cg}.
\notag\end{eqnarray}
For $i=1\dots, 5$ we denote by $a^{(i)},b^{(i)},c^{(i)}$ the solution to the system \eqref{massFs}, \eqref{momenFs}, \eqref{enerFs} with sources $\mathbf{s}^{(i)}$ and by $P^{(i)}=a^{(i)}+c^{(i)}$ the $i$-th contribution to the  pressure.

Correspondingly we have the decomposition of $\P_\c f$ into six terms:
\be \P_\c f=(1-\z)\P_\c f+ \sum_{i=1}^5 \S_if,\ee
with
\be \S_i f=\sqrt{\mu_\c}\big[a^{(i)} +v_\c\cdot b^{(i)}+ \frac 1 2 c^{(i)}(|v_\c|^2-3)\big].\ee
\subsection{Estimate of $\S_1 f$} \label{stimas1}
The components of $\S_1 \hat f$ solve the system
\begin{eqnarray}
&&ik\cdot \hat b^{(1)} +i\e k\cdot \c \hat a^{(1)}=\hat s_0^{(1)},\label{massFs1}\\
&&i k  \hat P^{(1)} +\e { {\mathbb{K} b^\z}}+i\c\cdot k]\hat b^{(1)}=\hat{\underline{s}}^{(1)}\label{momenFs1}\\
&& i k \cdot \hat b^{(1)}+\e\ [\kappa|k|^2+\frac 3 2ik\cdot\c ] \hat c^{(1)}=\hat s_4^{(1)}\label{enerFs1},\end{eqnarray}
where
\begin{eqnarray}\hat s_0^{(1)}&=&\mathfrak{j}^c\int_{\R^3} \dd v\, \sm  \hat{\mathcal{C}}(k,v)+ \mathfrak{j}^{ {c}}\widehat{\z\mathfrak{a}},\notag\\
\hat{\underline{s}}^{(1)}\ &=&\mathfrak{j}^c\Big[ {-}\e k \otimes k \cdot \int_{\R^3} \dd v\, v \cdot \mathscr{B}_\c\ipc  \hat f^\z -i\e k \cdot\int_{\R^3} \dd v\, (\widehat{\z g}+ \hat{\mathcal{C}}) \mathscr{B}_\c\notag\\&+&\int_{\R^3} \dd v\, v_\c\sm_\c  \hat{\mathcal{C}}(k,v)+\widehat{\z\mathfrak{b}}\Big],\label{sourcesfou1}\\
\hat s_4^{(1)}&=&\mathfrak{j}^c\Big[ {-}\e k\otimes k \cdot \int_{\R^3} \dd v\, v \ipc  \hat f^\z \mathscr{A}_\c-i\e k \cdot\int_{\R^3} \dd v\, (\widehat{\z g}+ \hat{\mathcal{C}}) \mathscr{A}_\c\notag\\&+&\int_{\R^3} \dd v\, \frac 1 2(|v_\c|^2-3)\sm_\c  \hat{\mathcal{C}}(k,v)+\widehat{\z\mathfrak{c}}\Big].\notag\end{eqnarray}

\begin{lemma}\label{PcRh} If $|\c|\ll 1$, and $g\in L^2$, then 
\be \|\S_1 f\|_2\lesssim \e^{-1} [ \|\P_\c f\|_6+\|\ipc f\|_2]+\| g{{\nu}^{-\frac 1 2}}\|_2.\label{PRh2}\ee
\end{lemma}
\begin{proof}
We first estimate $\hat P^{(1)}$. For this we use  the momentum balance in the form (\ref{momenk}), which for the $\mathbf{S}_1\hat R$ becomes:
\be i k \hat P^{(1)}+ i\mathfrak{j}^c k\cdot \hat \tau +i\e\c\cdot k\hat b^{(1)}=  {\mathfrak{j}^c\widehat{\z\mathfrak{b}}+ }\mathfrak{j}^c\int_{\R^3} \dd v\, v_\c\sm_\c \hat{\mathcal{C}}(k,v).\label{momenk1}\ee
We take inner product of  this equation with $\frac k{i|k|^2}$. We obtain
\be\hat P^{(1)}= \mathfrak{j}^c\Big[ {-i |k|^{-2} k\cdot\widehat{\z\mathfrak{b}}} -i |k|^{-2} k\cdot \int_{\R^3} \dd v\, v\hat{\mathcal{C}} \sm_\c -\e|k|^{-2} \c \cdot k \hat b^{ {(1)}}\cdot k {-}|k|^{-2}k\cdot \hat\tau\cdot k\Big].\label{Ph0}\ee
From the definition of $\tau$, \eqref{stress}, $\|\tau^{(1)}\|_2\le \|\ipc f\|_2$. Moreover from the definition of $\mathcal{C}$, \eqref{gzeta},
\be\|\hat{\mathcal{C}}\|_2=\|\mathcal{C}\|_2\lesssim \|\P_\c f\|_{L^2({\rm supp }\,\nabla\z)}+\|\ipc f\|_{L^2({\rm supp} \nabla\z )}\lesssim \|\P_\c f\|_6 +\|\ipc f\|_2.\ee Therefore
\be\|\hat P^{(1)}\|_2\lesssim \|\P_\c f\|_6+ \|\ipc f\|_2+ \e|\c|\|\hat b^{(1)}\|_2.\label{Ph1}\ee
To bound $\hat b^{(1)}$,  {we multiply (\ref{momenFs}) by $\e^{-1}\mathbb{K}^{1}$ and obtain
\begin{multline}\hat b^{(1)}=\mathfrak{j}^c\Big\{\e^{-1}\mathbb{K}^{-1}\Big[- {i k}\hat P^{(1)}-i\e\c\cdot k \hat b^{(1)}- \int_{\R^3} \dd v\, v_\c\sm_\c \hat{\mathcal{C}}(k,v)\\+i\e k \cdot \int_{\R^3} \dd v\, v \ipc \hat f^\z \mathscr{B}_\c+ i\e k \cdot\int_{\R^3} \dd v\,(\widehat{\z g}+\hat{\mathcal{C}}) \mathscr{B}_\c\Big]\Big\}\label{bh0}.\end{multline}
Since $|k|>1$, using \eqref{bK} and $|\c|\ll 1$, we have}
\be
\|\hat b^{(1)}\|_2\le \e^{-1}\|\hat P^{(1)}\|_2+ \e^{-1}\|\P_\c f\|_6+\|\ipc f\|_2+\| g{{\nu}^{-\frac 1 2}}\|_2,\label{bh1}\ee
having bounded $\|\int \widehat{\z g} \mathscr{B}_\c\|_2\le \|\nu^{-\frac 1 2} g\|_2\|\nu^{\frac 12} \mathscr{B}_\c\|_\infty\lesssim \|\nu^{-\frac 1 2} g\|_2$.

Using (\ref{bh1}) in (\ref{Ph1}) and $|\c|\ll 1$ we have
\be\|\hat P^{(1)}\|_2\lesssim \|\P_\c f\|_6+ \|\ipc f\|_2+\| g{{\nu}^{-\frac 1 2}}\|_2.\label{Ph2}\ee
Using (\ref{Ph2}) in (\ref{bh1}) we obtain
\be
\|\hat b^{(1)}\|_2\lesssim \e^{-1} [ \|\P_\c f\|_6+\|\ipc f\|_2]+\| g{{\nu}^{-\frac 1 2}}\|_2.
\label{bh2}\ee
To estimate $\hat c^{(1)}$ we subtract (\ref{massFs}) from (\ref{enerFs}) and replace $\hat a^{(1)}$ with $\hat P^{(1)}-\hat c^{(1)}$:
\begin{multline}-i\e\c\cdot k[\hat P^{(1)}-\hat c^{(1)}]+\int_{\R^3} \dd v\, \sm \hat{\mathcal{C}}(k,v) +\e\ [\kappa|k|^2+\frac 3 2ik\cdot\c ] \hat c^{(1)}+i\e k \cdot \int_{\R^3} \dd v\, v \ipc \hat f^\z \mathscr{A}_\c+\\\e i k \cdot\int_{\R^3} \dd v\, (\widehat{\z g}+\hat{\mathcal{C}}) \mathscr{A}_\c= \int_{\R^3} \dd v\, \frac 1 2(|v_\c|^2-3)\sm_\c \hat{\mathcal{C}}(k,v).
\end{multline}
Then we proceed as done for $\hat b^{(1)}$ and obtain:
\be
\|\hat c^{(1)}\|_2\lesssim \e^{-1} [ \|\P_\c f\|_6+\|\ipc f\|_2]+\| g{{\nu}^{-\frac 1 2}}\|_2.
\label{ch2}\ee
From the estimates of $\hat P^{(1)}$ and $\hat c^{(1)}$ we then obtain also 
\be
\|\hat a^{(1)}\|_2\lesssim \e^{-1} [ \|\P_\c f\|_6+\|\ipc f\|_2]+\| g{{\nu}^{-\frac 1 2}}\|_2.
\label{ah2}\ee
Thus, we obtain \eqref{PRh2}.\end{proof}

\bigskip
To deal with the system \eqref{massFs}, \eqref{momenFs}, \eqref{enerFs} for $|k|\le 1$ we need several estimates:
\begin{lemma}\label{multipl}  
Suppose $\c\neq 0$ and $|k|\le 1$. Let ${N}_{\s,\b}(k)=\e[\sigma |k|^2+\beta i k\cdot \c]$, for $\s>0$ and $\b>0$. There is $\varrho>0$ such that
\begin{enumerate}
\item  For $q\in[\frac 3 2,2)$
\be \|\mathfrak{j}\/{N}^{-1}_{\s,\b}\|_q\lesssim \e^{-1}|\c|^{-1+\varrho}, \label{N}\ee
and, for  $1<q< \frac 3 2$
\be \|\mathfrak{j}\/{N}^{-1}_{\s,\b}\|_q\lesssim \e^{-1}.\label{N1}\ee
\item For $q\in [3,4)$
\be \| \mathfrak{j}\/k {N}^{-1}_{\s,\b}\|_q\lesssim \e^{-1}|\c|^{-1+\varrho},\label{kN}\ee
and, for $1<q<3$
\be \| \mathfrak{j}\/k {N}^{-1}_{\s,\b}\|_q\lesssim \e^{-1},.\label{kN1}\ee
\item 
\be \| k\otimes k {N}^{-1}_{\s,\b}\|_\infty\lesssim \e^{-1}.\label{kN2}\ee
\end{enumerate}
\end{lemma}
\begin{proof} 

For $\ell\ge 0$ we compute  the norm (see \cite{Masl})
\begin{multline*} \|\e\mathfrak{j}|k|^\ell {N}^{-1}_{\s,\b}\|_q^q {\lesssim}\int_{\R^3} \dd k|k|^{q\ell}\mathfrak{j} |\sigma |k|^2+\beta i k\cdot \c|^{-q}\\\le2\pi\sigma^{-q} \int_0^2 \dd r \,r^{2+q(\ell-2)}\int_0^\pi \dd\theta \sin\theta \Big[1+r^{-2}\beta^2\sigma^{-2}|\c|^2\cos^2\theta\Big]^{-\frac q 2}\\=
\frac{2\pi\sigma^{-q}}{\beta\sigma^{-1}|\c|} \int_0^2 \dd r\,r^{3+q(\ell -2)}\int_0^{r^{-1}\beta|\c| {\sigma^{-1}}}\dd z[1+z^2]^{-\frac q 2}, \end{multline*}
with  $z= r^{-1}\beta\sigma^{-1} |\c| \cos \theta$. The integral in $dz$ is finite for $q>1$. The integral in $dr$ is finite for $3+q(\ell-2)>-1$. Hence, for $\ell<2$,  $q<\frac 4{2-\ell}$. Therefore, if  we split the integration on $r$ into $\{r\le |\c|^\d\}$ and $\{|\c|^\d<r\le  {2}\}$, with $0<\d<1$ to be chosen,
we have the bounds
\begin{multline*}\int_0^{|\c|^\d }\dd r\,r^{3+q(\ell -2)}\int_0^{r^{-1}\beta|\c| {\sigma^{-1}}}\dd z[1+z^2]^{-\frac q 2}\lesssim |\c|^{[4+q(\ell-2)]\d},\\
\int_{|\c|^\d }^2\dd r\, r^{3+q(\ell -2)}\int_0^{r^{-1}\beta|\c| {\sigma^{-1}}}\dd z[1+z^2]^{-\frac q 2}\le\int_{|\c|^\d }^2\dd r\,r^{3+q(\ell -2)}\int_0^{|\c|^{-\d+1}\beta{\sigma^{-1}}}\dd z[1+z^2]^{-\frac q 2}\lesssim |\c|^{1-\d}
.\end{multline*} 
By choosing $\d=(5+q(2-\ell))^{-1}<1$, we conclude that  
$$\|\e \mathfrak{j}\/|k|^\ell {N}^{-1}_{\s,\b}\|_q\lesssim |\c|^{-\frac{\d}q}=|\c|^{-1+\varrho}$$ because $\d<1$ and $q>1$.
Thus, for $\ell=0$ we obtain \eqref{N}, for $\ell=1$  we obtain \eqref{kN}. 

If we bound the integrand in $\dd\theta$ simply with $1$, as in the Stokes problem, we get instead
\be \|\e\mathfrak{j}\/|k|^\ell {N}^{-1}_{\s,\b}\|_q^q\le{2\pi^2} \int_0^2 \dd r\,r^{2+(\ell-2)q}. \ee
The integral in $\dd r$ is finite for $q<\frac 3{2-\ell}$. For $\ell=0$, the integral is bounded when  $q<\frac3 2$,  and hence we get \eqref{N1}; for $\ell=1$, the integral is bounded when $q<3$ and hence we obtain \eqref{kN1}. Clearly $\e|k|^2{N^{-1}_{\sigma,\beta}}\lesssim 1$ for any $k$, thus we have \eqref{kN2}.
\end{proof}
\subsection{Estimate of $\S_2 f$} \label{stimas2}

The components of $\S_2 \hat f$ solve the system
\begin{eqnarray}
&&ik\cdot \hat b^{(2)} +i\e k\cdot \c \hat a^{(2)}=\hat s_0^{(2)},\label{massFs2}\\
&&i k  \hat P^{(2)} +\e[ {\mathbb{K}}+i\c\cdot k]\hat b^{(2)}=\hat{\underline{s}}^{(2)}\label{momenFs2}\\
&& i k \cdot \hat b^{(2)}+\e\ [\kappa|k|^2+\frac 3 2ik\cdot\c ] \hat c^{(2)}=\hat s_4^{(2)}\label{enerFs2},\end{eqnarray}
where
\begin{eqnarray}\hat s_0^{(2)}&=&\mathfrak{j}\int_{\R^3} \dd v\, \sm_\c  \hat{\mathcal{C}}(0,v),\notag\\
\hat{\underline{s}}^{(2)}&=&\mathfrak{j}\int_{\R^3} \dd v\, v_\c\sm_\c  \hat{\mathcal{C}}(0,v)\Big],\label{sourcesfou2}\\
\hat s_4^{(2)}&=&\mathfrak{j}\int_{\R^3} \dd v\, \frac 1 2(|v_\c|^2-3)\sm_\c  \hat{\mathcal{C}}(0,v),\notag\end{eqnarray}
We use the  notation $\psi_0=\sm_\c$, $\psi_\alpha=\sm_\c v_{\c, \alpha}$, $\alpha=1,\dots,3$, $\psi_4= \frac 1 {\sqrt{6}}
\sm_\c(|v_\c|^2-3))$, so that 
$s^{(2)}_\a =\mathfrak{j}(\P_\c \hat {\mathcal{C}}_s, \psi_\a)_{L^2_v}$.
\begin{lemma}\label{lemma53}
\be \P_\c \hat {\mathcal{C}}_s=(2\pi)^{-\frac 3 2}\sum_{\a=0}^4Q_\a\psi_\a,\ee
with $\mathbf{Q}=(Q_0,\dots Q_4)$,  
\be Q_\a=-\int_{\pt\O}\dd S(x) \int_{\R^3} \dd v\, f v\cdot n(x) \psi_\a(v)+\int_{\O_1\backslash\O} \dd x (1-\z)\int_{\R^3} \dd v\, \psi_\a\P_\c g.
\label{defQalpha}\ee
\end{lemma}
\begin{proof}
Since $\hat{\mathcal{C}}(0,v)= (2\pi)^{-\frac 3 2}\int  \dd x f v\cdot\nabla \z $,   we have
\[\P_\c\hat{\mathcal{C}}_{s}= (2\pi)^{-\frac 3 2}\sum_{\a=0}^4 \psi_\a \int_{\O^c} \dd x\int_{\R^3} \dd v\, f\psi_\a v\cdot\nabla \z . 
\]
Since $\nabla \z=0$ outside of $\O_1=\{x\in \R^3\, |\, d(x,\O)<1\}$, 
\begin{multline*}\int_{\O^c} \dd x\int_{\R^3} \dd v\, \psi_\a (v\cdot\nabla \z) f= 
\int_{\O_1\backslash\O} \dd x\int_{\R^3} \dd v\, \psi_\a v\cdot\nabla (\z f)-\int_{\O_1\backslash\O} \dd x\int_{\R^3} \dd v\, \psi_\a\z v\cdot\nabla f =\\
 {\int_{\pt\O}\dd S(x) \int_{\R^3} \dd v\, v\cdot n(x) \z f  \psi_\a(v)+\int_{\pt \O_1}\dd S(x) \int_{\R^3} \dd v\, v\cdot N(x) \z f  \psi_\a(v)- \int_{\O_1\backslash\O} dx\z \int_{\R^3} dv \psi_\a v\cdot\nabla f}\\=- \int_{\O_1\backslash\O} \dd x\,\z \int_{\R^3} \dd v\, \psi_\a \P_\c g +\int_{\pt \O_1}\dd S(x) \int_{\R^3} \dd v\, v\cdot N(x)f  \psi_\a(v),\end{multline*}
where $N(x)$ is the exterior normal to $\pt \O_1$, because $\z=0$ on $\pt \O$ and $\z=1$ on $\pt \O_1$ and we have used by \eqref{consR}.
On the other hand, integrating (\ref{consR}) on $\O_1\backslash\O$   we get 
\be\int_{\pt \O_1}\dd S(x) \int_{\R^3} \dd v\, v\cdot N(x)f  \psi_\a(v)+ \int_{\pt\O}\dd S(x) \int_{\R^3} \dd v\, v\cdot n(x) f\psi_\a=\int_{\O_1\backslash\O}\dd x \int_{\R^3} \dd v\, \psi_\a\P_\c g  ,\label{540}\ee
and hence we obtain
\be Q_\a=-\int_{\pt\O}\dd S(x) \int_{\R^3} \dd v\, f v\cdot n(x) \psi_\a(v)+\int_{\O_1\backslash\O} \dd x (1-\z)\int_{\R^3} \dd v\, \psi_\a\P_\c g.
\label{defQalpha1}\ee
\end{proof}
\begin{lemma}[Estimate of $Q$'s]\label{estiQ} If $\|\P_\c g\|_{L^{6/5}_{loc}}=\|\P_\c g\|_{L^{6/5}(\O_1\backslash \O)}$ is bounded, then 
\be |\mathbf{Q}|\le \e \big(\e^{-1}\|\ipc f\|_2+ \|\P_\c f\|_6+\|\nu^{-\frac 1 2}
g\|_{L^2(\O_1\backslash\O)}\big)+\e^{\frac 1 2}\|z_\g(r)\|_2+ \|\P_\c g\|_{L^{6/5}_{loc}}.\label{estiQe}\ee
\end{lemma}
\begin{proof}
For any $h$ we have 
\be \int_{\R^3} \dd v\,n\cdot v \sqrt{\mu_\c}h=\int_{\{n\cdot v<0\}}\dd v\,n\cdot v \sqrt{\mu_\c}(h-P_\g^\c h).\label{fluxh}\ee
 {Indeed}  
\begin{multline}\int_{\{n\cdot v<0\}}\dd v\, n\cdot v \sqrt{\mu_\c}(h-P_\g^\c h)= \int_{\{n\cdot v<0\}}\dd v\,n\cdot v \sqrt{\mu_\c}h\\ -\int_{\{n\cdot v<0\}}\dd v\,n\cdot v \sqrt{2\pi}\mu(v)\int_{\{n\cdot v'>0\}}\dd v'\,v'\cdot n \sqrt{\mu_\c} h\\=  \int_{\{n\cdot v<0\}}\dd v\,\cdot v \sqrt{\mu_\c}h+\int_{\{n\cdot v'>0\}}\dd v'\,v'\cdot n \sqrt{\mu_\c} h= \int_{\R^3}\dd v\, n\cdot v \sqrt{\mu_\c}h,\end{multline}
because $\int_{\{n\cdot v<0\}}\dd v\, n\cdot v \sqrt{2\pi}\mu(v)=-1$. 

By \eqref{fluxh} and \eqref{linprob0}, 
\begin{multline*}\int_{\pt\O}dS(x)\int_{\R^3} dv \sm_\c f v\cdot n(x)= \int_{\pt\O}\dd S(x)\int_{\{v'\cdot n(x)<0\}} \dd v\, \sm_\c (f-P_\g^\c f) v\cdot n(x)\\=\e^{\frac 1 2}\int_{\pt\O}\dd S(x)\int_{\{v'\cdot n<0\}} \dd v\, \sm_\c  r \, v'\cdot n(x).\end{multline*}
Therefore, 
by \eqref{defQalpha},	
\be |Q_0|\le \e^{\frac 1 2}\|z_\g(r)\|_2+ \|\P_\c g\|_{L^{6/5}_{loc}}.\ee

The other components of $\mathbf{Q}$ are more involved. 
 
Let $\eta(x)= d(x,\pt\O)$ be the signed distance of $x\in \R^3$ from $\pt\O$, positive in $\O^c$, well defined at least when $|\eta(x)|<\delta$ for some sufficiently small $\d>0$. Clearly $|\nabla \eta|=1$.
We consider the family of smooth closed surfaces $\{\mathbb{S}_\xi\}_{0\le \xi<\d}$,  defined as $\mathbb{S}_\xi=\{ x\in \O^c\ |\ \eta(x)=\xi\}$. We also define, for $x\in \mathbb{S}_\xi$, $n(x)=\nabla \eta(x)$. We have  $\mathbb{S}_0=\pt \O$ and, for any $\xi>0$, the sets $\O_\xi$ whose boundaries are $\mathbb{S}_\xi$ are such that $\O_\xi\subset \O_{\xi'}$ if $\xi<\xi'$.
If we integrate the conservation law on $\O_{\xi_2}\backslash\O_{\xi_1}$, since the exterior normal to $\pt \O_{\xi_1}$, $n_1(x)=-\nabla \eta(x)$, 
setting 
\be{Q}_{\xi,\a}=-\int_{\mathbb{S}_\xi}\dd S(x) \int_{\R^3}\dd v\, f\sqrt{\mu_\c} v\cdot n(x) \psi_\a(v),\ee
by Gauss theorem and \eqref{consR} we obtain
\be|{Q}_{\xi_1,\a}-{Q}_{\xi_2,\a}|=\Big|\int _{\O_{\xi_1}\backslash \O_{ {\xi_2}}}\dd x\,\psi_\a \P_\c g\Big|\lesssim\|\P_\c g\|_{L^{6/5}(\O_{\xi_1}\backslash \O_{\xi_2})}, \quad \a=0,\dots, 4.\label{Qr1=QR2mod}\ee
In particular, with 
\[ \varpi_\a=Q_{\a}-{Q}_{0,\a}\] we have  
\[|\varpi_\a|\lesssim \|\P_\c g\|_{L^{6/5}_{loc}}\]
and hence, since $|\nabla\eta|=1$, by the coarea formula,
\[Q_\a=\varpi_\a+\d^{-1}\int_0^\d \dd\xi Q_{\xi,\a}=\varpi_\a+ \d^{-1}\int_{\O_\d\backslash\O}\dd x\int_{\R^3} \dd v\,  f\sqrt{\mu_\c} v\cdot n(x) \psi_\a(v).\]

To estimate $\underline{Q}=(Q_1,Q_2,Q_3)$, we note that from the decomposition of $f=\sm_\c(a+b\cdot v_\c+ \frac 1 2(|v_\c|^2-3) +\ipc f$ and the  definitions of $\tau$ and $P$, 
\[\underline{Q}=\underline{\varpi}+\d^{-1}\int_{\O_\d\backslash\O}\dd x [Pn +\tau\cdot n +\e \c\cdot n b].\]
To get a bound for $P$, let us denote by $\bar P$ the average of $P$ on ${\O_\d\backslash\O}$: $\bar P=\d^{-1}\int_{\O_\d\backslash\O} P dx$. 
Let $\Phi$ be a vector function such that:
\[ \nabla\cdot \Phi=P-\bar P \text{ in } \O_\d\backslash\O,\quad \Phi=0 \text{ on } \pt\O\cup\pt\O_\d.\] 
Such a vector function exists and satisfies the bound (see \cite{Lady})
$$\|\Phi\|_{H^1(\O_\d\backslash\O)}\le \|P-\bar P\|_{L^2(\O_\d\backslash\O)}.$$
Taking the inner product of the momentum  balance law \eqref{Pnocut}
\be\nabla (P-\bar P)+\e \c \cdot \nabla b+\nabla\cdot \tau=\mathfrak{b},\label{740}\ee
by $\Phi$, integrating on $\O_\d\backslash\O$ and integrating by parts, we obtain
\begin{multline}\int_{\O_\d\backslash\O}\mathfrak{b}\cdot \Phi \dd x=-\int_{\O_\d\backslash\O}\dd x [\nabla\cdot \Phi (P-\bar P)+\e\c\cdot \nabla\Phi\cdot b+ \nabla\Phi:\tau]\\+\int_{\pt \O\cup \pt\O_\d}\dd S[\Phi\cdot n (P-\bar P)+ \Phi\otimes n:\tau+\e\c (n\cdot \Phi) (b\cdot n )],\end{multline}
where $A:B=\sum_{i,j} A_{i,j}B_{i,j}$.
The boundary terms vanish because $\Phi=0$ on the boundary. 
We have 
\be \Big|\int_{\O_\d\backslash\O}\mathfrak{b}\cdot \Phi \dd x\Big|\le \|\Phi\|_6\|\|\mathfrak{b}\|_{L^{ 6 /5}_{loc}}
\le \|\mathfrak{b}\|_{L^{ 6 /5}_{loc}} \|\nabla\Phi\|_2,\ee
by using Sobolev embedding. Therefore, 
using $\nabla\cdot \Phi=P-\bar P$, we obtain 
\begin{multline*}\|P-\bar P\|_{L^2({\O_\d\backslash\O})}^2\le \|\nabla \Phi\|_{L^2({\O_\d\backslash\O})}(\|\tau\|_{L^2({\O_\d\backslash\O})}+\e |\c|\|b\|_{L^2({\O_\d\backslash\O})}+
\|\mathfrak{b}\|_{L^{6 /5}_{loc}})\\ {\lesssim}\|P-\bar P\|_{L^2({\O_\d\backslash\O})}(\|\ipc f\|_{L^2({\O_\d\backslash\O})}+\e|\c|\|\P_\c f\|_6+\|\mathfrak{b}\|_{L^{ 6 /5}_{loc}}).\end{multline*}
Hence $$\|P-\bar P\|_{L^2({\O_\d\backslash\O})}\lesssim \e(\e^{-1}\|\ipc f\|_2 +|\c|\| \P_\c f\|_6)+\|\mathfrak{b}\|_{L^{6/5}_{loc}}$$
Therefore, since $\int_{\O_\d\backslash\O}\dd x \bar P n(x)=0$, 
we obtain 
$$|\int_{\O_\d\backslash\O}\dd x Pn|\le \e(\e^{-1}\|\ipc f\|_2 +|\c|\| \P_\c f\|_6)+\|\mathfrak{b}\|_{L^{\frac 6 5}_{loc}}.$$
On the other hand, 
$$\Big|\int_{\O_\d\backslash\O}\dd x \c\cdot n b\Big|\lesssim |\c|\|\P_\c f\|_{L^6}$$ 
and
$$\Big|\int_{\O_\d\backslash\O}\dd x n\cdot \tau\Big|\lesssim \|\ipc f\|_{L^2}.$$ In conclusion
$$|\underline{Q}|\lesssim_\d \e(\e^{-1}\|\ipc f\|_2 +|\c|\| \P_\c f\|_6) +\|\mathfrak{b}\|_{L^{6 /5}_{loc}}+|\underline{\varpi}|.$$

For the estimate of $Q_4$, we use  
\begin{multline}Q_4=\varpi_4 +\d^{-1}\int_0^\d \dd\xi Q_{\xi,4}= \\\varpi_4+\d^{-1}\int_{\O_\d\backslash\O} \dd x\int_{\R^3} \dd v \sqrt{\mu_\c}\frac{|v_\c|^2-3}2 n\cdot v\{[a+ b\cdot v_\c+ c \frac{|v_\c|^2-3}2]\sqrt{\mu_\c}+\ipc f\}\\
=\varpi_4+\d^{-1}\int_{\O_\d\backslash\O} \dd x 
[b\cdot n +\frac 3 2\e \c\cdot n c +\int_{\R^3}\dd v \sqrt{\mu_\c} \frac{|v_\c|^2-3}2 n\cdot v \ipc f].\label{Q4mod}
\end{multline}
To get a bound for $\int_{\O_\d\backslash\O} \dd x b\cdot n$ we note that, from $|Q_0-Q_{\xi,0}|\le \e^{\frac 1 2}\|z_\g(r)\|_2$, integrating on $\xi$ and using again the coarea formula, by \eqref{540}
\be\Big|\d^{-1}\int_{\O_\d\backslash\O} \dd x [b\cdot n +\e a \c \cdot n]- \int_{\pt\O} \dd S  [b\cdot n + \e\c\cdot n a]\Big|\le \e^{\frac 1 2}\|z_\g(r)\|_2+ \|\P_\c g\|_{L^{6/5}_{loc}}\ee
and
\be=\Big|\int_{\pt\O} \dd S  [b\cdot n + \e\c\cdot n a]-\e \int_{\pt\O} \dd S  \c\cdot n a\Big|\le\|z_\g(r)\|_2,\ee
because $|\int_{\pt\O}\dd S b\cdot n|=|\int_{\pt\O}\int_{\R^3} \dd v\sm_\c f v\cdot n|\le \|z_\g(r)\|_2$.
Hence $\int_{\pt\O} \dd S  b\cdot n=\e \int_{\pt\O} \dd S  \c\cdot n a+O(\|z_\g(r)\|_2)$.
Now we can replace in (\ref{Q4mod}) this expression to obtain:
\begin{multline*}|Q_4|\le   \d^{-1}\Big|\int_{\O_\d\backslash\O} \dd x 
\Big[\e\c\cdot n (c-a) +\int_{\R^3}\dd v \sqrt{\mu_\c} \frac{|v_\c|^2-3}2 n\cdot v \ipc f\Big]\Big|+\e \Big|\int_{\pt\O}\dd S \c\cdot n a\Big|\\+\e^{\frac 1 2}\|z_\g(r)\|_2+ \|\P_\c g\|_{L^{6/5}_{loc}}.\end{multline*}
The first term in the first line is bounded with $\e[|\c|\|a\|_6+|\c|\|c\|_6+\e^{-1}\|\ipc f\|_2]$. The second is bounded by $\e|\c|\|\P_\c f\|_{L^2(\pt\O\times \R^3)}\lesssim\e|\c|(\|\P_\c f\|_6+\e^{-1}\|\ipc f\|_\nu+\|
\nu^{-\frac 1 2}g\|_2)$, by using the Ukai trace theorem, Lemma \ref{trace_s}. \end{proof}

\bigskip
\begin{lemma}\label{maslova} If $\c\neq 0$, then  there is $\rho>0$ such that, for any $p>2$
\begin{multline} \|\S_2 f\|_{p}\lesssim \e^{-1}|\c|^{-1+\varrho}\sum_{\a=0, \dots,4}|Q_\a|\lesssim |\c|^{-1+\varrho}(\e^{-1}\|\ipc f\|_2+\|\P_\c f\|_6+\| g{{\nu}^{-\frac 1 2}}\|_2)\\+\e^{-1}|\c|^{-1+\varrho}[\e^{\frac 1 2}\|z_\g(r)\|_2+ \|\P_\c g\|_{L^{6/5}_{loc}}].
\label{maslovae}\end{multline}
\end{lemma}
\begin{proof}\ 

\noindent \underline{Step 1}. Estimate of $\Pi b^{(2)}$:
From \eqref{momenFPi} for the system \eqref{massFs}, \eqref{momenFs}, \eqref{enerFs}, with $\mathbf{s}=\mathbf{s}^{(2)}$, we have 
\be\hat\Pi\hat b^{(2)}=\mathfrak{j}{N}_{\mathfrak{v}, 1}^{-1}\hat\Pi\underline{Q}.\label{momenFPi10}\ee
By  \eqref{N}, $\mathfrak{j}{N}_{\mathfrak{v}, 1}^{-1}$ is bounded by $\e^{-1}|\c|^{-1+\varrho}$ in $L^q(\R^3)$ for $\frac 3 2 \le q<2$, 
and hence 
\be\|\hat\Pi\hat b^{(2)}\|_q\lesssim \e^{-1}|\c|^{-1+\varrho}\underline{Q}.\ee

\noindent \underline{Step 2}. Estimate of $P^{(2)}$: by using \eqref{press} for the system \eqref{massFs}, \eqref{momenFs}, \eqref{enerFs}, 
we have
\be\hat P^{(2)}=\Big(1-
\e\frac{{N}_{\mathfrak{v},1}}{i|k|^2}\c\cdot k
\Big)^{-1}\Big[\frac{{N}_{\mathfrak{v},1}}{i|k|^2}[i\hat s_0^{(2)}+\e\c\cdot k \hat c^{(2)}]+\frac{k}{i|k|^2}\cdot \hat{\underline{s}}^{(2)}\Big].\label{press10}\ee
Since $\mathfrak{j}|k|^{-q}$ is integrable for any $q<3$, we obtain 
\be \|\hat P^{(2)}\|_q\lesssim \e^2\|\hat c^{(2)}\|_q +|\underline{Q}|\label{presstemp10}\ee

\noindent \underline{Step 3}. Estimate of $c^{(2)}$: by using \eqref{energia} for the system  
\eqref{massFs}, \eqref{momenFs}, \eqref{enerFs}, we have
\be \hat c^{(2)}=\mathfrak{j}\overline{{N}}^{-1}\Big\{Q_4-Q_0+i\e\c\cdot k \Big(1-
\e\frac{{N}_{\mathfrak{v},1}}{i|k|^2}\c\cdot k
\Big)^{-1}\frac{k}{i|k|^2}\cdot \hat{\underline{Q}}\Big\},\label{energia1}\ee
We remind that from the definition of $\overline{N}$ it follows that $|\overline{N}^{-1}|\lesssim |{N}_{\kappa, \frac 5 2}^{-1}|$. Therefore, proceeding as before, we obtain 
by \eqref{N} for $\frac 3 2\le q<2$, 
\be \|\hat c^{(2)}\|_q\lesssim  \e^{-1}|\c|^{-1+\vr}  {(|Q_4| +|Q_0|)}+\e |\c||\underline{Q}|,\label{esc1}\ee
and, in consequence,
\be \|\hat P^{(2)}\|_q\lesssim |\c|^{-1+\vr}(|Q_4|+|Q_0|)+|\underline{Q}|\label{press110}\ee

\noindent \underline{Step 4}. Estimate of $\hat a^{(2)}$:
Using $\hat a^{(2)}=\hat P^{(2)}-\hat c^{(2)}$, we have 
\be \|\hat a^{(2)}\|_q\lesssim  \e^{-1}(|\c|^{-1+\vr} {(|Q_4| +|Q_0|)}+|\underline{Q}|).\label{esa2}\ee

\noindent \underline{Step 5}. Estimate of $(1-\hat \Pi) \hat b^{(2)}$:

Since $(1-\hat \Pi) \hat b^{(2)}= k\cdot \hat b^{(2)} k|k|^{-2}$, using the mass equation, where  {$\hat s_0^{(2)}=\mathfrak{j}Q_0$,  
which implies
$k\cdot \hat b^{(2)} =-\e k\cdot \c \hat a^{(2)}+i\mathfrak{j} Q_0$}, 
we have
$$(1-\hat \Pi) \hat b^{(2)}= -\e |k|^{-2}k k\cdot \c \hat a^{(2)} {+i|k|^{-2}k\mathfrak{j} Q_0},$$
and taking the $L^q$ norm we have, using Step 4,
$$\|(1-\hat \Pi) \hat b^{(2)}\|_q\le \e \|\hat a^{(2)}\|_q {+|Q_0|}.$$ 
Then, together with Step 1 we obtain 
\be\| \hat b^{(2)}\|_q\lesssim \e^{-1} {(|Q_4| +|Q_0|)}+|\underline{Q}|.\label{esb10}\ee 
In conclusion, 
$$\|\S_2 \hat f\|_q\lesssim \e^{-1}|\c|^{-1+\vr}\underline{Q}, \quad \text{ for } \frac 3 2\le q<2.$$
 {We remind the Hausdorff-Young inequality: if $1<q\le 2$ and $\frac 1 p+\frac 1 q=1$, then
\be \| f\|_p\le \|\hat f\|_q.\label{HY}\ee}
By the Hausdorff-Young inequality 
then we have (\ref{maslovae}) with $p=\frac q{q-1}>2$.
\end{proof}
\subsection{Estimate of $\S_3 f$} \label{stimas3}The components of $\S_3 \hat R$ solve the system
\begin{eqnarray}
&&ik\cdot \hat b^{(3)} +i\e k\cdot \c \hat a^{(3)}=\hat s_0^{(3)},\label{massFs3}\\
&&i k  \hat P^{(3)} +\e[ {\mathbb{K}}+i\c\cdot k]\hat b^{(3)}=\hat{\underline{s}}^{(3)}\label{momenFs3}\\
&& i k \cdot \hat b^{(3)}+\e\ [\kappa|k|^2+\frac 3 2ik\cdot\c ] \hat c^{(3)}=\hat s_4^{(3)}\label{enerFs3},\end{eqnarray}
where
\begin{eqnarray}\hat s_0^{(3)}&=&\mathfrak{j}\Big[\int_{\R^3}\dd v\,\sm_{ {\c}}  {k\cdot} \hat{\mathcal{C}}_r(k,v)\Big],\notag\\
\hat{\underline{s}}^{(3)}(k)&=&\mathfrak{j}\Big[ {-}\e k \otimes k \cdot \int_{\R^3}\dd v\, v \ipc  \hat f^\z(k,v) \mathscr{B}_\c-i\e k \cdot\int_{\R^3}\dd v\, \hat{\mathcal{C}} \mathscr{B}_\c\notag \\&&\hskip 5cm+ k\cdot\int_{\R^3}\dd v\, v_\c\sm_\c   \hat{\mathcal{C}}_r(k,v)\Big],\label{source3}\\
\hat s_4^{(3)}(k)&=&\mathfrak{j}\Big[ {-}\e k\otimes k \cdot \int_{\R^3}\dd v\, v \ipc  \hat f^\z(k,v) \mathscr{A}_\c-i\e k\cdot\int_{\R^3}\dd v\,\hat{\mathcal{C}} \mathscr{A}_\c\notag\\&&\hskip 4cm +k\cdot\int_{\R^3}\dd v\, \frac 1 2(|v_\c|^2-3)\sm_\c  \hat{\mathcal{C}}_r(x,v)\Big].\notag\end{eqnarray}

\begin{lemma} \label{lemmagr}
\be \|\int_{\R^3}\dd v\sm_\c v_\c\mathfrak{j} {k\cdot}\hat{\mathcal{C}}_r\|_\infty\le \| \P_\c f\|_6+\|\ipc f\|_\nu\label{gzinfity}\ee
\end{lemma}
\begin{proof} 
Recall from \eqref{560}, 
\begin{eqnarray*}
\Big|\int_{\R^3}\dd v\sm_\c v_\c\mathfrak{j} {k\cdot}\hat{\mathcal{C}}_r\Big| &=&\Big|\int_{\R^3}\dd v\sm_\c v_\c\mathfrak{j}\int_{0}^{1}\dd\lambda \frac{d}{
d\lambda }\int dxe^{i\lambda k\cdot x}\mathcal{C}(x,v)\Big| \\
&=&\Big|\int_{\R^3}\dd v\sm_\c v_\c\mathfrak{j}\int_{0}^{1}\dd\lambda \{ik\cdot x\}\int \dd x\;e^{i\lambda k\cdot
x}\mathcal{C}(x,v) \Big|\\
&\leq & \int_{\R^3}\dd v\sm_\c |v_\c|^2 \int_{\O_1} \dd x |x|\int_{\R^3} \dd v\,  {(| \P_\c f(x,v)|+|\ipc f(x,v)|)}\\&\lesssim&\|f\|_6+\|\ipc f\|_\nu,
\end{eqnarray*}
because $\text{\rm supp}(\nabla \z)\subset \O_1$.
\end{proof}
\begin{lemma}\label{ar}If $|\c|\ll 1$ and $\e\ll1$, 
\be \|\S_3f\|_2\lesssim \e^{-1}\|\P_\c f\|_6+ {\e^{-1}}\|\ipc f\|_2.\label{a1e}\ee
\end{lemma}
\begin{proof}
\noindent \underline{Step 1}. Estimate of $\Pi b^{(3)}$:

From \eqref{momenFPi} for the system \eqref{massFs3}, \eqref{momenFs3}, \eqref{enerFs3}, with $\mathbf{s}^{(3)}$ given by \eqref{source3}, we have 
\be\hat\Pi\hat b^{(3)}={N}_{\mathfrak{v}, 1}^{-1}\hat\Pi\hat{\underline{s}}^{(3)},\label{momenFPi1}\ee
where
$$\hat{\underline{s}}^{(3)}=\mathfrak{j}\Big[\e k \otimes k \cdot \int_{\R^3}\dd v\, v \ipc  \hat f^\z \mathscr{B}_\c-i\e k\cdot \int_{\R^3}\dd v\, \hat{\mathcal{C}}\mathscr{B}_\c+k\cdot\int_{\R^3}\dd v\, v_\c\sm_\c  \hat{\mathcal{C}}_r(k,v)\Big],$$
By  \eqref{kN1}, $\mathfrak{j}k{N}_{\mathfrak{v}, 1}^{-1}$ is bounded by $\e^{-1}$ in $L^2(\R^3)$, 
and hence
$$\Big\|\mathfrak{j}{N}^{-1}_{\mathfrak{v}, 1} \hat\Pi k {\cdot}\int_{\R^3}\dd v\, v_\c\sm_\c \hat{\mathcal{C}}_r(k,v)\Big\|_2\lesssim \e^{-1}\Big \|\int_{\R^3}\dd v\, v_\c\sm_\c  {k\cdot}\hat{\mathcal{C}}_r(k,v)\Big \|_\infty\le \e^{-1}\|\P_\c f\|_6+ {\e^{-1}}\|\ipc f\|_\nu,$$
$$\Big\|\mathfrak{j}\e{N}^{-1}_{\mathfrak{v},1} \hat\Pi k\int_{\R^3}\dd v\, \hat{\mathcal{C}}(k,v)\mathscr{B}_\c\Big\|_2\lesssim  \Big\|\int_{\R^3}\dd v\, \hat{\mathcal{C}}(k,v)\mathscr{B}_\c\Big\|_\infty\lesssim \|\P_\c f\|_6+\|\ipc f\|_\nu,$$
by using (\ref{gzinfity}).

On the other hand, by Lemma \ref{multipl}, $ k\otimes k {N}^{-1}_{\s,\b}\in L^\infty$, so 
$$\Big\|i\e \mathfrak{j}{N}^{-1}_{\mathfrak{v}, 1}\hat\Pi k\otimes k \cdot \int_{\R^3}\dd v\, v \ipc \hat f^\z \mathscr{B}_\c\Big\|_2\lesssim  \|\ipc f\|_\nu.$$
therefore we have 
$$\|\hat\Pi \hat b^{(3)}\|_2\lesssim \e^{-1}\|\P_\c f\|_6 + {\e^{-1}}\|\ipc f\|_2.$$

\noindent \underline{Step 2}. Estimate of $P^{(3)}$: by using \eqref{press} for the system \eqref{massFs3}, \eqref{momenFs3}, \eqref{enerFs3}, we have
\be\hat P^{(3)}=\Big(1 {-}
\e\frac{{N}_{\mathfrak{v},1}}{i|k|^2}\c\cdot k
\Big)^{-1}\Big[\frac{{N}_{\mathfrak{v},1}}{i|k|^2}[i\hat s_0^{(3)} +\e\c\cdot k \hat c]+\frac{k}{i|k|^2}\cdot \hat{\underline{s}}^{(3)}\Big].\label{press1}\ee
Taking the $L^2$ norm, for $\e\ll1$ we get
\be \|\hat P^{(3)}\|_2\le \e^2\|\hat c^{(3)}\|_2 +\|\P_\c f\|_6+\|\ipc f\|_\nu\label{presstemp1}\ee
\noindent \underline{Step 3}. Estimate of $c^{(3)}$: by using \eqref{energia} for the system  
\eqref{massFs}, \eqref{momenFs}, \eqref{enerFs}, we have
\be \hat c^{(3)}=(\overline{{N}})^{-1}\Big\{s_4^{(3)}-s_0^{(3)}+i\e\c\cdot k \Big(1 {-}
\e\frac{{N}_{\mathfrak{v},1}}{i|k|^2}\c\cdot k
\Big)^{-1}\Big[\frac{{N}_{\mathfrak{v},1}}{i|k|^2}i\hat s_0^{(3)} +\frac{k}{i|k|^2}\cdot \hat{\underline{s}}^{(3)}\Big]\Big\},\label{energia12}\ee
with  
\begin{multline*}\hat s_4^{(3)}(k)=\mathfrak{j}\Big[\e k\otimes k \cdot \int_{\R^3}\dd v\, v \ipc  \hat f^\z(k,v) \mathscr{A}_\c-i\e k\cdot\int_{\R^3}\dd v\, \hat{\mathcal{C}} \mathscr{A}_\c\\+\int_{\R^3}\dd v\, \frac 1 2(|v_\c|^2-3)\sm_\c  k\cdot\hat{\mathcal{C}}_r(k,v)\Big].\end{multline*}
We remind that from the definition of $\overline{N}$, it follows that $|\overline{N}^{-1}|\lesssim |N_{\kappa, \frac 5 2}^{-1}|$. Therefore, proceeding as before, we obtain 

\be \|\hat c^{(3)}\|_2\lesssim  \e^{-1} \|\P_\c f\|_6 +\|\ipc f\|_2.\label{esc12}\ee
and, in consequence,
\be \|\hat P^{(3)}\|_2\le \|\P_\c f\|_6+\|\ipc f\|_2\label{press11}\ee

\noindent \underline{Step 4}. Estimate of $\hat a^{(3)}$:
Using $\hat a^{(3)}=\hat P^{(3)}-\hat c^{(3)}$, we have 
\be \|\hat a^{(3)}\|_2\lesssim  \e^{-1} \|\P_\c f\|_6 +\|\ipc f\|_2.\label{esa22}\ee

\noindent \underline{Step 5}. Estimate of $(1-\hat \Pi) \hat b^{(3)}$:

Since $(1-\hat \Pi) \hat b^{(3)}= k\cdot \hat b^{(3)} k|k|^{-2}$, using the mass equation which implies
$k\cdot \hat b^{(3)} =-\e k\cdot \c \hat a^{(3)} {+i\hat s_0^{(3)}}$, 
we have
$$(1-\hat \Pi) \hat b^{(3)}= -\e |k|^{-2}k k\cdot \c \hat a^{(3)} {+i|k|^{-2}k\hat s_0^{(3)}}-i |k|^{-2}k\cdot\Big[\int_{\R^3}\dd v\, v_\c \sm \hat{\mathcal{C}}_r(k,v)+\e\int_{\R^3}\dd v\, \hat{\mathcal{C}}(k,v)\mathscr{B}_\c\Big],$$
and taking the $L^2$ norm we have, using Step 4,
$$\|(1-\hat \Pi) \hat b^{(3)}\|_2\le \e \|\hat a^{(3)}\|_2 +\|\P_\c f\|_6\lesssim \|\P_\c f\|_6+\|\ipc f\|_\nu
.$$ Then, together with Step 1 we obtain 
\be\| \hat b^{(3)}\|_2\lesssim \e^{-1}\|\P_\c f\|_6+\|\ipc f\|_\nu.\label{esb1}\ee \end{proof} 
\subsection{Estimate of $\S_4 f$}\label{stimas4}
The components of $\S_4 \hat f$ solve the system
\begin{eqnarray}
&&ik\cdot \hat b^{(4)} +i\e k\cdot \c \hat a^{(4)}=\hat s_0^{(4)},\label{massFs4}\\
&&i k  \hat P^{(4)} +\e[ {\mathbb{K}}+i\c\cdot k]\hat b^{(4)}=\hat{\underline{s}}^{(4)}\label{momenFs4}\\
&& i k \cdot \hat b^{(4)}+\e\ [\kappa|k|^2+\frac 3 2ik\cdot\c ] \hat c^{(4)}=\hat s_4^{(4)}\label{enerFs4},\end{eqnarray}
where
\begin{eqnarray}\hat s_0^{(4)}(k)&=&0,\\
\hat{\underline{s}}^{(4)}(k)&=&-\mathfrak{j}i\e k \cdot\int_{\R^3}\dd v\, \widehat{\z g}\mathscr{B}_\c,\\
\hat s_4^{(4)}(k)&=&-\mathfrak{j}\e ik\cdot \int_{\R^3}\dd v\, \widehat{\z g} \mathscr{A}_\c.\end{eqnarray}
\begin{lemma}\label{ar2}Let $p\ge 2$ and assume $g\in L^{\frac {3p}{3+p}}$.
Then
\be\|\S_4 f\|_p\lesssim \|\nu^{-\frac 1 2}g\|_{\frac  {3p}{3+p}},\label{ar213}\ee
\end{lemma}
\begin{proof} 

We proceed as in the proof of Lemma \ref{ar}: 

\noindent 
\underline{Step 1}. Estimate of $\Pi b^{(4)}$:\newline
From \eqref{momenFPi} for the system \eqref{massFs}, \eqref{momenFs}, \eqref{enerFs}, with $\mathbf{s}=\mathbf{s}^{(4)}$, 
\be\hat\Pi\hat b^{(4)}={N}_{\mathfrak{v}, 1}^{-1}\hat\Pi\hat{\underline{s}}^{(4)}.
\label{momenFPi4}\ee
Since for the multipliers $kN_{b,1}^{-1}\hat{
\Pi}k$ direct computations yields 
\[
\partial _{k}^{l}\{\e kN_{b,1}^{-1}\hat{\Pi}k\}\lesssim _{l}|k|^{-l},
\]
with constants independent of $\e$, 
by Mihlin-Hormander's \cite{Mih,Hor}
multiplier theorem, we deduce
\be\|\nabla \Pi b^{(4)}\|_{\frac  {3p}{3+p}}= \Big\|\nabla \mathcal{F}^{-1}[{N}_{\mathfrak{v}, 1}^{-1}\hat\Pi(\e\mathfrak{j} ik\cdot \int_{\R^3}\dd v\,\widehat{\zeta g}\mathscr{B}_\c)]\Big\|_{\frac  {3p}{3+p}}
\lesssim \|\nu^{-\frac 1 2}g\|_{\frac  {3p}{3+p}},
\ee 
by the Sobolev estimate  
\[\|\Pi b^{(4)}\|_p\lesssim\|\nabla \Pi b^{(4)}\|_{\frac  {3p}{3+p}}\lesssim\|\nu^{-\frac 1 2}g\|_{\frac  {3p}{3+p}}.\]
\noindent \underline{Step 2}. Estimate of $P^{(4)}$:  
by using \eqref{press} for the system \eqref{massFs}, \eqref{momenFs}, \eqref{enerFs}, we have
\be\hat P^{(4)}=\Big(1+
\e\frac{{N}_{\mathfrak{v},1}}{i|k|^2}\c\cdot k
\Big)^{-1}\Big[\frac{{N}_{\mathfrak{v},1}}{i|k|^2}[\e\c\cdot k \hat c^{(4)}]+\frac{k}{i|k|^2}\cdot \hat{\underline{s}}^{(5)}\Big],\label{press21}\ee
from which we get 
\be\|P^{(4)}\|_p\lesssim \e^2 \|c^{(4)}\|_p+ \e\|\nu^{-\frac 1 2}g\|_{\frac  {3p}{3+p}}.\label{presstemp2}\ee
\noindent \underline{Step 3}. Estimate of $c^{(4)}$: 
by using \eqref{energia} for the system  
\eqref{massFs}, \eqref{momenFs}, \eqref{enerFs}, we have
\be \hat c^{(4)}=(\overline{{N}})^{-1}\Big\{s_4^{(4)}+i\e\c\cdot k \Big(1+
\e\frac{{N}_{\mathfrak{v},1}}{i|k|^2}\c\cdot k
\Big)^{-1}\frac{k}{i|k|^2}\cdot \hat{\underline{s}}^{(4)}\Big\},\label{energia2}\ee
with 
$$s_4^{(4)}= \e i k \cdot\int_{\R^3}\dd v\, \widehat{\z g} \mathscr{A}_\c. $$
This implies
\be\|c^{(4)}\|_p\le \|\nu^{-\frac 1 2}g\|_{\frac  {3p}{3+p}}.\label{ener2}\ee
In consequence 
\be \|\hat P^{(4)}\|_p\le \e\|\nu^{-\frac 1 2}g\|_{\frac  {3p}{3+p}}.\label{press22}\ee
\noindent \underline{Step 4}. Estimate of $\hat a^{(4)}$: 
 Using $\hat a^{(4)}=\hat P^{(4)}-\hat c^{(4)}$, we have 
\be \|a^{(4)}\|_p\lesssim \|\nu^{-\frac 1 2}g\|_{\frac  {3p}{3+p}} .\label{esa212}\ee

\noindent \underline{Step 5}. Estimate of $(1-\hat \Pi) \hat b^{(4)}$:
Since $(1-\hat \Pi) \hat b^{(4)}= k\cdot b^{(4)}_l k|k|^{-2}$, using the equation for the mass we have $(1-\hat \Pi) \hat b^{(4)}= -\e |k|^{-2}k k\cdot \c \hat a^{(4)}$, and hence,  by Step 4
\be\|(1-\Pi)b^{(4)}\|_p\le \e \|a^{(4)}\|_p\lesssim \e\|\nu^{-\frac 1 2}g\|_{\frac  {3p}{3+p}}.\ee
Then, together with Step 1 we obtain 
\be\| b^{(4)}\|_p\lesssim \|\nu^{-\frac 1 2}g\|_{\frac  {3p}{3+p}}.
\label{esb11}\ee
\end{proof}
\subsection{Estimate of $\mathbf{S}_5 f$}\label{stimas5}
 The components of $\S_5 \hat f$ solve the system
\begin{eqnarray}
&&ik\cdot \hat b^{(5)} +i\e k\cdot \c \hat a^{(5)}=\hat s_0^{(5)},\label{massFs6}\\
&&i k  \hat P^{(5)} +\e[ {\mathbb{K}}+i\c\cdot k]\hat b^{(5)}=\hat{\underline{s}}^{(5)}\label{momenFs6}\\
&& i k \cdot \hat b^{(5)}+\e\ [\kappa|k|^2+\frac 3 2ik\cdot\c ] \hat c^{(5)}=\hat s_4^{(5)}\label{enerFs46},\end{eqnarray}
where
\begin{eqnarray}\hat s_0^{(5)}(k)&=&\mathfrak{j}\int_{\R^3}\dd v\,\sm_\c \widehat{\z\P_\c g},\notag\\
\hat{\underline{s}}^{(5)}(k)&=&\mathfrak{j}\int_{\R^3}\dd v\,\sm_\c v \widehat{\z\P_\c g},\\
\hat s_4^{(5)}(k)&=&\mathfrak{j} \int_{\R^3}\dd v\,\sm_\c \frac 1 2(|v|^2-3) \widehat{\z\P_\c g}.\notag\end{eqnarray}
\begin{lemma}\label{S6}Let $p>2$. 
Suppose that ${\z\P_\c g}\in L^q$, $1<q<\frac{2p}{p+2}$.
Then there is $\rho>0$ such that
\be \|\mathbf{S}_{5} f\|_p\lesssim\e^{-1}|\c|^{-1+\varrho} \|{\z\P_\c g}\|_q
.\label{s5}\ee
\end{lemma}
\begin{proof} 
By Hausdorff-Young inequality  {\eqref{HY}},
\be \|\mathbf{S}_5f\|_p\lesssim \|\widehat{\mathbf{S}_5f}\|_{\frac p{p-1}}.\ee
We have
\be \hat\Pi \hat b^{(5)}= N_{\mathfrak{v},1}^{-1} \hat\Pi \mathfrak{j}\int_{\R^3}\dd v\,\sm_\c v \widehat{\z\P_\c g}.\ee
Therefore,  if $1<r$ and $r\frac p{p-1}<2$, so that we can use with \eqref{N}, then, with $\frac 1 q+\frac 1{q'}=1$,
we have
\be \|\hat\Pi \hat b^{(5)}\|_{\frac p{p-1}}\le \|N_{\mathfrak{v},1}^{-1}\|_{r\frac p{p-1}}\|\widehat{\z\P_\c g}\|_{\frac p{p-1}\frac r{r-1}}\lesssim \e^{-1}|\c|^{-1+\varrho}\|\widehat{\z\P_\c g}\|_{q'}\lesssim \e^{-1}|\c|^{-1+\varrho} \|{\z\P_\c g}\|_{q},\ee
where $q'=\frac p{p-1}\frac r{r-1}$, having used \eqref{N} in the second inequality  and again the Hausdorf-Young inequality in the last step. Since $\frac {p-1}{r}>\frac p 2$, then $q=\frac{pr}{p+r-1}=\frac{p}{\frac{p-1}{r}+1}<\frac{2p}{p+2}$. Since $r>1$, then $q>1$.
\end{proof}
\subsection{Proof of Theorem \ref{mainlinth}}\label{provalineare}
\begin{proposition}\label{L^3}
If $\c\neq 0$ and $\e\ll1$, 
then there is $\rho>0$ such that,
\begin{multline} \e^{\frac 1 2}\|\P_\c f\|_3\lesssim \|\P_\c f\|_6+ |\c|^{-1+\varrho}\e^{-1}\|\ipc f\|_2+o(1)\|
\nu^{-\frac 1 2}g\|_{L^2(\O_1\backslash\O)}
+\e^{\frac 1 2}\|\nu^{-\frac 1 2}g\|_{\frac 32}\\+|\c|^{-1+\varrho}\big[\e^{-\frac 1 2}\|z_\g(r)\|_2+\e^{-1}\|\P_\c g\|_{L^{{\frac 6 5}^-}}\big].\label{L_3}\end{multline}
\end{proposition}
\begin{proof}
To get the $L^3$ bound of $\P_\c f$ we proceed as follows: we look at the problem in $\R^3$ by passing to the cut-offed problem. Thus we obtain $\P_\c f=(1-\z) \P_\c f+\sum_{i=1}^5 \S_i f$. Since $1-\z(x)=0$ if $x\notin\O_1=\{ x\,|\, d(x,\O)\le 1\}$, $\|(1-\z) \P_\c f\|_3\lesssim\|\P_\c f\|_6$. For the other terms we use the previous lemmas.

The bounds in previous subsettions are too singular in $\e$ for our purposes. Therefore, we take advantage of the uniform-in-$\e$ estimate of $\|\z \P_\c f\|_6$ to improve the estimate of $\|\z \P_\c f\|_3$ by means of interpolation between the $L^6$ norm and some lower norm. Since
\be \|\sum_{i=1}^5\S_i f\|_6\le \|\P_\c f\|_6,\ee 
we have 
\be \|\S_1f+\S_3 f\|_6\le \|\P_\c f\|_6+ \|\S_2f\|_6+\|\S_4 f\|_6+\|\S_5f\|_6.\ee
Therefore, using \eqref{maslovae}, \eqref{ar213} and \eqref{s5} with $p=6$, 
we obtain:
\begin{eqnarray}  \|\S_1f+\S_3 f\|_6&\lesssim& \|\P_\c f\|_6+ |\c|^{-1+\varrho}(\e^{-1}\|\ipc f\|_2+\|\P_\c f\|_6+\| g{{\nu}^{-\frac 1 2}}\|_2)\notag\\ &+&\e^{-1}|\c|^{-1+\varrho}[\e^{\frac 1 2}\|z_\g(r)\|_2+ \|\P_\c g\|_{L^{6/5}_{loc}}]+\|\nu^{-\frac 1 2}g\|_2+\e^{-1}|\c|^{-1+\varrho} \|{\z\P_\c g}\|_\frac65\notag\\&\lesssim&
\|\P_\c f\|_6+|\c|^{-1+\varrho}\e^{-1}\|\ipc f\|_2+|\c|^{-1+\varrho}\| g{{\nu}^{-\frac 1 2}}\|_2\\&+& |\c|^{-1+\varrho}[\e^{-\frac 1 2}\|z_\g(r)\|_2+\e^{-1} \|\P_\c g\|_{L^{6/5}_{loc}}].\notag\end{eqnarray}
Note that only the last line is singular in $\e$, but we will apply the inequality in a situation where $z_\g(r)$ and $\P_\c g$ are small in $\e$.

For  $\S_1 f+\S _3 f$, by \eqref{PRh2} and \eqref{a1e} (by interpolation ($\|f\|_r\le \|f\|_p^\theta\|f\|_q^{1-\theta}$ with $r^{-1}= \theta p^{-1} +(1-\theta)q^{-1}$)) we obtain, with $r=3$, $p=2$, $q=6$ and $\theta=\frac 1 2$, 
\begin{multline} \e^{\frac 1 2}\|\S_1 f+\S_3 f\|_3\le (\e\|\S_1 f+\S_3f\|_2)^{\frac 1 2}\|\S_1 f+\S_3 f\|_6^{\frac 1 2}\lesssim  [ \|\P_\c f\|_6^{\frac 1 2}+\|\ipc f\|_2^{\frac 1 2}+\e^{\frac 1 2}\|\nu^{-\frac 1 2}g\|_2^{\frac 1 2}]\times\\
\Big[\|\P_\c f\|_6+|\c|^{-1+\varrho}\e^{-1}\|\ipc f\|_2+|\c|^{-1+\varrho}\| g{{\nu}^{-\frac 1 2}}\|_2+ |\c|^{-1+\varrho}[\e^{-\frac 1 2}\|z_\g(r)\|_2+ \e^{-1}\|\P_\c g\|_{L^{6/5}_{loc}}]\Big]^{\frac 1 2}
\\\lesssim \|\P_\c f\|_6 +|\c|^{-1+\varrho}\e^{-1}\|\ipc f\|_2 + (\e+|\c|^{-1+\varrho}) \|\nu^{-\frac 1 2}g\|_2+|\c|^{-1+\varrho}[\e^{-\frac 1 2}\|z_\g(r)\|_2+ \e^{-1}\|\P_\c g\|_{L^{6/5}_{loc}}].\end{multline}
As for $\S_4 f$, we have from Lemma \ref{ar2} with $p=3$, 
\be\e^{\frac 1 2}\|\S_4 f\|_3\lesssim \e^{\frac 1 2}\| \nu^{-\frac 1 2}g\|_{\frac 32}
.\ee
For $\mathbf{S}_5 f$ we use Lemma \ref{S6} with $p=3$ and hence $1<q<\frac 6 5$.

By \eqref{maslovae}, by interpolation we obtain, 
\begin{multline}\e^{\frac 1 2}\|\S_2 f\|_3\lesssim \e^{\frac 1 2}\Big[|\c|^{-1+\varrho}(\e^{-1}\|\ipc f\|_2+\|\P_\c f\|_6+\|\nu^{-\frac 1 2}g\|_2)\\+\e^{-1}|\c|^{-1+\varrho}[\e^{\frac 1 2}\|z_\g(r)\|_2+ \|\P_\c g\|_{L^{6/5}_{loc}}]\Big]^\theta \|\mathbf{S}_2 f\|_6^{(1-\theta)},\end{multline}
with $\theta$ such that $\frac 1 3= \theta p^{-1} +\frac 1 6(1-\theta)$, and hence $\theta=\frac 1 2^+$  when $p=2^+$. Therefore, by Young inequality, 
\begin{multline}\e^{\frac 1 2}\|\S_2 f\|_3\lesssim\e^{\frac 1 2}|\c|^{-1+\vr}(\e^{-1}\|\ipc f\|_2+\|\P_\c f\|_6+\|\nu^{-\frac 1 2}g\|_2)\\+|\c|^{-1+\rho}[\|z_\g(r)\|_2+ \e^{-\frac 1 2}\|\P_\c g\|_{L^{6/5}_{loc}}].
\end{multline}
Combining the estimates we obtain \eqref{L_3}.
\end{proof}

Now we have all the information needed to prove Theorem \ref{mainlinth}.
\begin{proof}[Proof of Theorem \ref{mainlinth}] To bound the first two terms of $\lbr f\rbr_{\beta,\beta'}$, we use Proposition \ref{propener}. Then we use Proposition \ref{linftyest} in \ref{P6999}:
\begin{multline} (1-o(1))\|\P_\c f\|_{6}\lesssim (1+o(1)+|\c|^{-2+2\rho})\big[\e^{-1}\|\ipc f\|_\nu +\e^{-\frac 1 2}|(1-P_\g^\c)f|_{2,+}\big]+
\|\nu^{-\frac 1 2}g\|_{2}+ \e^{\frac 1 2} |r|_\infty\\+ o(1)[\e^{\frac 12} | w r |_{\infty}
+ \e^{\frac 32} \| \langle v\rangle^{-1} w g \|_{\infty}] 
.\label{P69999}\end{multline}
Using this in \eqref{ener99}, if $|\c|$ is so small that $|\c|(1+o(1)+|\c|^{-1+\rho})(1-o(1))^{-1}<\frac 1 2$, we obtain 
\begin{eqnarray}&& \e^{-2} \|\ipc f\|_\nu^2 +\e^{-1}|(1-P_\g^\c)f|^2_{2,+}\lesssim \|\nu^{-\frac 1 2}\ipc g\|_2^2 +|\c|[ \|\nu^{-\frac 1 2}g\|_{2}+ \e^{\frac 1 2} |r|_\infty\label{ener9999}\\&&+ o(1)\Big [\e^{\frac 12} | w r |_{\infty}
+ \e^{\frac 32} \| \langle v\rangle^{-1} w g \|_{\infty}] \Big]+ |r|_{2,-}^2+(\e|\c|)^{-1}\|z_\g(r)\|_2^2+ \e^{-2} |\c|^{-2}\|\P_\c g\|_{\frac 6 5}^2+\|\P_\c g\|_2^2.\notag\end{eqnarray}
Using this in \eqref{P6999} we obtain a similar bound for $\|\P_\c f\|_{6}$:
\begin{multline} \|\P_\c f\|_{6}\lesssim \|\nu^{-\frac 1 2}\ipc g\|_2^2 +|\c|[ \|\nu^{-\frac 1 2}g\|_{2}+ \e^{\frac 1 2} |r|_\infty+ o(1)\Big [\e^{\frac 12} | w r |_{\infty}
+ \e^{\frac 32} \| \langle v\rangle^{-1} w g \|_{\infty}] \Big]\\+ |r|_{2,-}^2+(\e|\c|)^{-1}\|z_\g(r)\|_2^2+ \e^{-2} |\c|^{-2}\|\P_\c g\|_{\frac 6 5}^2+\|\P_\c g\|_2^2.
\label{P699991}\end{multline}
Using \eqref{ener9999} and \eqref{P699991} in \eqref{linf99} we get a similar bound for $\e^{\frac 1 2}\|w f\|_\infty$. Finally, using \eqref{ener9999} and \eqref{P699991} in \eqref{L_3} we obtain the bound on $\e^{\frac 1 2}\|\P_\c f\|_3$. Rearranging the terms we obtain \eqref{mainlinest}.
\end{proof}
\bigskip
\section{Construction of the positive solution to the non linear problem}\label{iterazione}
\subsection{Positivity scheme}\label{positvity}
In order to construct a non negative solution to the problem \eqref{1} we use a  modification of the argument introduced in \cite{AN}.

We define $F^+= \max\{F,0\}$ and $F^-=\max\{-F,0\}$, so that $F=F^+-F^-$.
Consider the system
\begin{eqnarray} &&v\cdot \nabla F= \e^{-1}[Q(F^+, F^+) - 2Q(\mu_\c, F^-)] \text{ in }\O^c,\label{mainpose} \\&& F\Big|_{-}= \mathcal{P}^w_\g (F^+) \text{ on } \pt\O, \quad \lim_{|x|\to {\infty}} F=\mu_\c.\label{mainposb}\end{eqnarray} 
\begin{proposition}\label{arkeryd} 
Let $F\in L^\infty$ solve problem \eqref{mainpose}, \eqref{mainposb}. Then $F^-=0$ and $F^+$ solves the Boltzmann equation. \end{proposition}
\begin{remark}
Since $F^-=0$, $F$ is non negative.\end{remark}
\begin{proof}
In fact, the equation for $F^-$ is
\[-v\cdot \nabla F^-= \e^{-1}\1_{F^-\neq 0}[Q^+(F^+, F^+)- Q(\mu_\c, F^-)-Q(F^-,\mu_\c)], \quad F^-\Big|_-=0.\]
because  $F^-\neq 0$ implies $F^+=0$, and hence the term $\1_{F^-\neq 0} Q^-(F^+,F^+)=\1_{F^-\neq 0}F^+\nu(F^+)=0$. Moreover, since $F>0$ on $\g_-$, if follows that $F^-=0$ on $\g_-$. Since $F\to \mu_\c>0$ as $|x|\to \infty$, then $F^-\to 0$ as $|x|\to \infty$.

By multiplying this equation by $-\mu_\c^{-1} F^-$ and integrating, we obtain:
\begin{multline*}\int_{\O^c\times \R^3} \dd x \dd v \mu_\c^{-1}v\cdot \nabla \frac {(F^-)^2} 2= -\e^{-1}\int_{\O^c\times \R^3}\dd x \dd v \1_{F^-\neq 0} Q^+(F^+,F^+) F^-\mu_\c^{-1}\\+\e^{-1}\int_{\O^c\times \R^3}\dd x \dd v \1_{F^-\neq 0}\mu_\c^{-1} F^- [Q(\mu_\c, F^-)+Q(F^-,\mu_\c)]\end{multline*}
By the spectral inequality,
\begin{multline*}-\int_{\O^c\times \R^3}\dd x \dd v \1_{F^-\neq 0}\mu_\c^{-1} F^- [Q(\mu_\c, F^-)+Q(F^-,\mu_\c)]\\=-\int_{\O^c\times \R^3}\dd x \dd v \mu_\c^{-1} F^- [Q(\mu_\c, F^-)+Q(F^-,\mu_\c)]\gtrsim \|\ipc (\mu_\c^{-\frac 1 2}F^-)\|_2^2.\end{multline*}
Therefore by also integrating by parts the l.h.s., we obtain
\begin{multline*}\frac 1 2\int_{\pt\O\times \R^3} \dd S(x) \int_{\R^3} \dd v \mu_\c^{-1}v\cdot n (F^-)^2 + \e^{-1}\|\ipc(\mu_\c^{-\frac 1 2}F^-)\|_2^2\\\lesssim -\e^{-1}\int_{\O^c\times \R^3}\dd x \dd v\1_{F^-\neq 0} Q^+(F^+,F^+) F^-\mu_\c^{-1}\le 0\end{multline*}
This implies that $F^-=0$ on $\g^+$, thus $F^-=0$ on $\g$. Moreover $\ipc(\mu_\c^{-\frac 1 2}F^-)=0$ and hence $Q(\mu_\c, F^-)+Q(F^-,\mu_\c)=0$. Thus
$$-v\cdot\nabla F^-=  \e^{-1}\1_{F^-\neq 0}Q^+(F^+, F^+)\ge 0.$$
Therefore $F^-$ satisfies
$$ v\cdot \nabla F^-\le 0 \quad \text{ in }\O^c, \quad F^-=0 \quad \text{ on } \g.$$
This implies that  $F^-\le 0$, but $F^-\ge 0$ by definition and hence $F^-=0$ identically.
Then, $F=F^+$ and \eqref{mainpose} coincides with the Boltzmann equation \eqref{1} and \eqref{mainposb} is the usual diffuse reflection boundary condition \eqref{bc0}.\end{proof}

\bigskip
Therefore, to construct a positive solution to \eqref{1} we need to construct a solution to  \eqref{mainpose}, \eqref{mainposb}. We need some notation:

Let $\chi=\mathbf{1}_{|v|<\e^{-m}}$, $\bar\chi=\mathbf{1}_{|v|\ge\e^{-m}}=1-\chi$ where   $m>0$ is such such that 
\be\mu_{\c}+\e\sqrt{\mu}_\c \chi(f_1+\e f_2)>0.\label{muchi}\ee
 Such an $m$ certainly exist because, by definition $f_1$ and  by \cite{Bob} $f_2$, are bounded by $\sm_\c P_s$, for some $s>1$, where $P_s$ is a polynomial of degree $s$ in $v$.

Since, for $\beta>0$,  $\exp[-\e^{-\beta}]\lesssim \e^\ell$ for any $\ell>0$,  in the rest of this section we shall use the short notation
\be \e^\infty=\exp[-\e^{-\beta}], \quad \text{ for some } \beta>0.\ee
Remind that 
\be0\le M_{1,\e(\c+ u), 1}=\mu_\c+\e \sqrt{\mu_\c} f_1 + \e^2\sqrt{\mu_\c}\phi_\e.\label{mupos1}\ee
We denote 
$$\mathscr{Q}=f_1+\e (\chi f_2+\bar \chi \phi_\e).$$ 
By \eqref{muchi}, if $\chi=1$, then $\mu_\c+\e \sqrt{\mu} \mathscr{Q}\ge 0$, and the same is true  {if} $\bar\chi=1$  by \eqref{mupos1}. Therefore
$$\mu_\c+\e \sqrt{\mu} \mathscr{Q}\ge 0.$$
We decompose
\be F=\mu_\c+ \e\sqrt{\mu}_\c \mathscr{Q} +\e^{\frac3 2} R\sm_\c.\label{decompR}\ee
Then we define
\be\bar R=\begin{cases} R\quad \text{ if } \mu_\c+\e \sqrt{\mu}\mathscr{Q} +\e^{\frac3 2}\sqrt{\mu}_\c R\ge 0\\
-\e^{-\frac 3 2}(\mu_\c+\e \mathscr{Q}\sqrt{\mu}_\c )\quad \text{ if }\mu_\c+\e \mathscr{Q}\sqrt{\mu}_\c +\e^{\frac3 2}\sqrt{\mu}_\c R< 0\end{cases},\label{defbarR}\ee
and \be\tilde R=\bar R-R.\label{deftildeR}\ee
It follows that
\be F^+= \mu_\c+ \e\sqrt{\mu_\c} \mathscr{Q} +\e^{\frac3 2} \bar R\sm_\c.\label{decompbarR}\ee
\be F^-= \e^{\frac3 2} \tilde R\sm_\c.\label{decomptildeR}\ee
 {Indeed}, if $F(x,v)>0$, then 
\[F^+=F=\mu_\c+ \e\sqrt{\mu}_\c \mathscr{Q} +\e^{\frac3 2} R\sm_\c=\mu_\c+ \e\sqrt{\mu} \mathscr{Q} +\e^{\frac 3 2} \bar R\sm_\c.\]
Moreover, if $F(x,v)\le 0$, then $\e^{\frac3 2} R\le -(\mu_\c+ \e \mathscr{Q}\sqrt{\mu_\c})$ and hence
$$0=\mu_\c+ \e \mathscr{Q}\sqrt{\mu_\c}+\e^{\frac3 2}\bar R\sm_\c=F^+,$$
and
$$F^-=F^+-F= \mu_\c+ \e\sqrt{\mu} \mathscr{Q} +\e^{\frac 3 2} \bar R\sm_\c -(\mu_\c+ \e\sqrt{\mu} \mathscr{Q} +\e^{\frac3 2} R\sm_\c)= 
\e^{\frac 3 2}(\bar R-R)\sm_\c=\e^{\frac3 2}\tilde R\sm_\c.
$$
\begin{lemma} We have the following inequalities:
\be |\bar R|\le |R|,\label{barR<R}\ee
\be |\tilde R|\le \mathbf{1}_{\{\mu_\c+ \e\sqrt{\mu}_\c \mathscr{Q} +\e^{\frac 3 2} \bar R<0\}}2|R|,\label{tildeR<R}\ee
\be |\bar R_1-\bar R_2|\le |R_1-R_2|,\label{lipbarR}\ee
\be |\tilde  R_1-\tilde  R_2|\le 2|R_1-R_2|(\mathbf{1}_{\{\mu_\c+ \e\sqrt{\mu}_\c \mathscr{Q} +\e^{\frac 3 2}  R_1\sm_\c<0\}}+\mathbf{1}_{\{\mu_\c+ \e\sqrt{\mu}_\c \mathscr{Q} +\e^{\frac 3 2}  R_2\sm_\c<0\}}).\label{liptildeR}\ee 
\end{lemma}
\begin{proof}
Indeed,  $\mu_\c+\e \mathscr{Q}\sqrt{\mu}_\c +\e^{\frac 3 2}\sqrt{\mu}_\c R< 0$ implies $R<0$ and hence
$\e^{\frac 3 2}|\bar R|=- \e^{\frac 3 2}\bar R =\sm_\c^{-1}(\mu_\c+\e\sqrt{\mu_\c}\mathscr{Q})<-\e^{\frac 3 2}R=\e^{\frac 3 2}|R|$, which proves \eqref{barR<R}. Moreover, $|\tilde R|\le (|R|+|\bar R|)\mathbf{1}_{\{\mu_\c+ \e\sqrt{\mu} (\mathscr{Q} +\e^{\frac 1 2} \bar R)<0\}}\le 2|R|\mathbf{1}_{\{\mu_\c+ \e\sqrt{\mu} (\mathscr{Q}+\e^{\frac 1 2} \bar R)<0\}}$ which proves \eqref{tildeR<R}.
Furthermore given $F_1$ and $F_2$, we have
$$|F_1^+-F_2^+|\le |F_1-F_2|.$$
$$|F_1^--F_2^-|\le 2|F_1-F_2|.$$
In fact, fixed $(x,v)$, without loss of generality suppose $F_1(x,v)\ge F_2(x,v)$. If $F_2(x,v)>0$ there nothing to show. Thus assume $F_2(x,v)\le 0$. 
If $F_1(x,v)\le 0$ then  $F^+_1(x,v)=F^+_2(x,v)=0$ and the inequality is obviously verified.
Therefore we only need to consider the case $F_1(x,v)> 0$ and $F_2\le0$. We have
$$|F_1^+(x,v)-F_2^+(x,v)|=F_1(x,v)\le F_1(x,v)-F_2(x,v)=|F_1(x,v)-F_2(x,v)|.$$
Moreover, since $F^-= F^+-F$, $|F_1^--F_2^-|\le |F_1^+-F_2^+|+|F_1-F_2|\le 2 |F_1-F_2|$.
Therefore, with $R_i$ defined by \eqref{decompR} and $\bar R_i$ by \eqref{decompbarR},
it follows that
\begin{multline*}\e^{\frac 3 2}|\bar R_1-\bar R_2|= \sm_\c^{-1}|(F_1^+-\mu_\c -\e\sqrt{\mu_\c}\mathscr{Q}) -(F_2^+-\mu_\c -\e\sqrt{\mu_\c}\mathscr{Q})|= \sm_\c^{-1}|F_1^+-F_2^+|\le \sm_\c^{-1}|F_1-F_2|\\=\sm_\c^{-1}|(F_1-\mu_\c -\e\sqrt{\mu_\c}\mathscr{Q}) -(F_2-\mu_\c -\e\sqrt{\mu_\c}\mathscr{Q})|=\e^{\frac 3 2}|R_1-R_2|.\end{multline*}
Hence \eqref{lipbarR} is proved.
Furthermore
\be |\tilde  R_1-\tilde  R_2|\le 2|R_1-R_2|(\mathbf{1}_{\{\mu_\c+ \e\sqrt{\mu_\c}( \mathscr{Q} +\e^{\frac 1 2}  R_1)<0\}}+\mathbf{1}_{\{\mu_\c+ \e\sqrt{\mu_\c}( \mathscr{Q} +\e^{\frac 1 2}  R_2)<0\}}).\label{liptildeR1}\ee 
In fact, since $F_1^--F_2^-$ vanishes outside of the set $$\{\mu_\c+ \e\sqrt{\mu_\c}( \mathscr{Q} +\e^{\frac 1 2}  R_1)<0\}\cup\{\mu_\c+ \e\sqrt{\mu_\c} (\mathscr{Q} +\e^{\frac 1 2}  R_2)<0\}$$
and $\e^{-\frac3 2}|F_1-F_2|=|R_1-R_2|$, we have
\begin{multline*}|\tilde  R_1-\tilde  R_2|=\e^{-\frac 3 2}\sm_\c^{-1}|(F^-_1-\mu_\c -\e\sqrt{\mu_\c}\mathscr{Q})-(F^-_2-(\mu_\c -\e\sqrt{\mu_\c}\mathscr{Q})|=\e^{-\frac 3 2}\sm_\c^{-1}|F^-_1-F^-_2|\le\\
 2\e^{-\frac 3 2}\sm_\c^{-1}|F_1-F_2|(\mathbf{1}_{\{\mu_\c+ \e\sqrt{\mu_\c} (\mathscr{Q} +\e^{\frac 1 2}  R_1)<0\}}+\mathbf{1}_{\{\mu_\c+ \e\sqrt{\mu_\c}( \mathscr{Q} +\e^{\frac 1 2}  R_2)<0\}})=\\
 2|R_1-R_2|(\mathbf{1}_{\{\mu_\c+ \e\sqrt{\mu_\c}( \mathscr{Q} +\e^{\frac 3 2}  R_1)<0\}}+\mathbf{1}_{\{\mu_\c+ \e\sqrt{\mu_\c} (\mathscr{Q} +\e^{\frac 3 2}  R_2)<0\}}).\end{multline*}
 \end{proof}
\medskip
As for the boundary conditions, we have
$$\mu_\c +\e f_1\mu_\c^{\frac 1 2}+\e^2 (\chi f_2+\bar\chi \phi_\e) \mu_\c^{\frac 1 2} +\e^{\frac 3 2} R \mu_\c^{\frac 1 2}= \mathcal{P}^w_\g[\mu_\c +\e\chi f_1\mu_\c^{\frac 1 2}+\e^2 (\chi f_2+\bar\chi \phi_\e) \mu_\c^{\frac 1 2} +\e^{\frac 3 2} \bar R \mu_\c^{\frac 1 2}].$$
Therefore, subtracting this equations from \eqref{mu=Pmu},
$$\e^2 \chi (f_2-\phi_\e) \mu_\c^{\frac 1 2} +\e^{\frac 3 2} R \mu_\c^{\frac 1 2}=\mathcal{P}^w_\g[\e^2 \chi  (f_2-\phi_\e) \mu_\c^{\frac 1 2} +\e^{\frac 3 2} \bar R \mu_\c^{\frac 1 2}].$$
Hence
\be R=P^\c_\g R +\e^{\frac 1 2} \bar r+P^\c_\g \tilde R,\label{bcRmod1}\ee
with
$$\bar r= {P}^\c_\g[ \chi (f_2- \phi_\e)]- \chi (f_2- \phi_\e).$$  
We have 
\be\Big\|\int_{\R^3} dv [\mu_\c+ \sqrt{\mu_\c}\e \mathscr{Q}] v\cdot n\Big\|_\infty=\e^\infty
\quad \text{ on }\pt\O.\label{fluxmod2}\ee 
In fact 
\be\int_{\R^3} \dd v (\mu_\c+(\e f_1+\e^2\phi_\e)\sqrt{\mu_\c}) v\cdot n=\int_{\R^3} \dd v M_{1,\e(\c+u),1} v\cdot n =\e n\cdot(\c+u)=0,\ee on $\pt\O$ because $u= {-}\c$ on $\pt \O$,  {see \eqref{INS}}. We have also $\int_{\R^3} dv (\mu_\c+\e \sm_\c f_1)v\cdot n=0$ on $\pt\O$ and hence $\int_{\R^3} dv \sm_\c  \phi_\e v\cdot n=0$ on $\pt\O$. Therefore, by \eqref{estzeta}, since $u|_{\pt \O}=-\c$,

\be\Big|\int_{\R^3}\dd v\, n\cdot v \sqrt{\mu_\c}\chi\phi_\e\Big|=\Big|-\int_{\R^3}\dd v\,n\cdot v \sqrt{\mu_\c}\bar\chi\phi_\e\Big|\le e^{-\e^{-m}}|\c|^2\lesssim\e^\infty \quad \text{on } \pt\O.\label{619}\ee
Since $\P_\c f_2=0$, in the same way we obtain 
\be\Big|\int_{\R^3}\dd v n\cdot v \sqrt{\mu_\c}\chi f_2\Big|=\Big|-\int_{\R^3}\dd v\, n\cdot v \sqrt{\mu_\c}\bar\chi f_2\Big|\le e^{-\e^{-m}}|\c|^2\lesssim\e^\infty \quad \text{on } \pt\O,\label{620}\ee
because $|f_2|\le\sm_\c P_\ell|(|\nabla u|+|u|^2)$ and $\nabla u$ is bounded in $L^p$ for any $p>\frac 4 3$ and \eqref{fluxmod2} follows.

The boundary conditions for $F$ imply 
\[\int_{\R^3}F\dd v\,v\cdot n =-\int_{\{v\cdot n>0\}} \dd v\, F^-n\cdot v,\]
on $\pt\O$. 
Therefore we have 
\[\int_{\R^3}\dd v\, \sm_\c Rv\cdot n=-\int_{v\cdot n>0} \dd v\, \sm_\c\tilde R n\cdot v+O(\e^\infty)\]
\begin{lemma}
\be\Big\|\int_{v\cdot n<0}\dd v\, \bar r\sqrt{\mu_\c}n\cdot v\Big\|_\infty\lesssim\e^\infty.\label{r0mod1}\ee 
\end{lemma}
\begin{proof}
We have $\bar r=P_\g^\c (\chi f_2-\bar\chi\phi_\e)-(\chi f_2-\bar\chi\phi_\e)$ 
and, using \eqref{fluxh}, $\int_{v\cdot n<0}\dd v\, \bar r\sqrt{\mu_\c}n\cdot v=\int_{\R^3}\dd v\, \sqrt{\mu_\c}(\chi f_2-\chi\phi_\e)n\cdot v$. \eqref{619} and \eqref{620} imply \eqref{r0mod1}. 
\end{proof}

\medskip
We rewrite the problem \eqref{mainpose}, \eqref{mainposb} using the decompositions \eqref{decompR} and \eqref{decompbarR}. Reminding the definitions of $f_1$, $f_2$ and the incompressible Navier-Stokes equations, we are reduced  to construct the solution to the problem:
\begin{eqnarray}&& v\cdot \nabla R +\e^{-1} L_\c R= L_\c^{(1)} \bar R +\e^{\frac 1 2} \Gamma_\c(\bar R,\bar R) +\e^{\frac 1 2} \bar A_\c,\label{Rmod}\\
&&R\Big|_{\gamma_-}= P_\gamma^\c R +\e^{\frac 1 2} r,\label{bmod}\\
&&\text{where}\notag\\
&&L_\c^{(1)} \bar R = 2\tilde\Gamma_\c(\mathscr{Q} , \bar R),\label{modL1}\\
&&\P_\c\bar A_\c= \P_\c[\bar \chi v\cdot\nabla(\bar\chi(\phi_\e-f_2))  {-\e\c\cdot\nabla f_2}]\label{pabar}\\&&\ipc\bar A_\c=\ipc(v\cdot\nabla (\chi f_2+\bar \chi \phi_\e))\notag\\&&-\tilde\Gamma_\c (2f_1+\e(\chi f_2+\bar\chi\phi_\e),\chi f_2+\bar\chi\phi_\e)+\e^{-1}L_\c[ \bar \chi(\phi_\e-f_2)],\label{1-pabar}\\
&&r=\bar r -\e^{-\frac 1 2} P_\g^\c \tilde R.\label{rmod} \end{eqnarray}
In fact, reminding
\eqref{f2} and \eqref{Pvnablaf1},
we have 
\[L_\c (\chi f_2) -\G_\c(f_1,f_1)+ \ipc(v\cdot \nabla f_1)=-L_\c (\bar \chi f_2),\]
 and 
\[\P_\c(v\cdot \nabla \chi f_2)= -\P_\c(v\cdot \nabla \bar \chi f_2),\]
so that 
\[v\cdot \nabla(\chi f_2) =\ipc(v\cdot \nabla(\chi f_2))-\P_\c(v\cdot \nabla(\bar\chi f_2)).\]

Therefore,
\begin{eqnarray*}
\bar A_\c:&=&\e^{-(\frac 12 +\frac 3 2)}\Big\{\e\{\ipc(v\cdot\nabla f_1)+L_\c(\chi f_2)-\G_\c(f_1,f_1)\}+\e L_\c(\bar\chi \phi_\e)\\&+&\e^2v\cdot \nabla(\chi f_2+\bar\chi \phi_\e)-\e^2\tilde\G_\c((\chi f_2+\bar\chi \phi_\e),  2f_1+\e(\chi f_2+\bar\chi \phi_\e))
\Big\} {-\e\c\cdot\nabla f_2}\\&=&-\e^{-1}L_\c(\bar\chi f_2)+\e^{-1} L_\c(\bar  \chi \phi_\e)+
\ipc(v\cdot \nabla  (\chi f_2))\\ &-&\P_\c(v\cdot \nabla(\bar\chi f_2))+v\cdot \nabla (\bar\chi \phi_\e) -\tilde\G_\c((\chi f_2+\bar\chi \phi_\e),  2f_1+\e(\chi f_2+\bar\chi \phi_\e)) {-\e\c\cdot\nabla f_2}\\&=&
\e^{-1}L_\c[\bar\chi(\phi_\e-f_2)]+ \P_\c[v\cdot\nabla(\bar\chi(\phi_\e-f_2))]+\ipc(v\cdot\nabla(\chi f_2+\chi_\c\phi_\e))\\&-&\tilde\G_\c((\chi f_2+\bar\chi \phi_\e),  2f_1+\e(\chi f_2+\bar\chi \phi_\e)) {-\e\c\cdot\nabla f_2}.
\end{eqnarray*}
\begin{proposition}\label{721}
Let $X\in L^p(\O^c\times \R^3)$  {and $\bar X$ and $\tilde X$ be defined as $\bar X$ and $\tilde R$, as in \eqref{defbarR} and \eqref{deftildeR}. Assume }$p>1$, $|\c|\ll1$, Then:
\begin{enumerate}
\item Let $w(v)$ be such that $w^{-1}\lesssim \sm_\c^{\beta}\langle v\rangle^{-\beta'}$ for some $0<\beta\ll1$ and $\beta'>0$. If $\e^{\frac 1 2}\|w X\|_\infty$ is bounded, then 
\be|P_\g^\c\tilde X|_{2,-}\lesssim [\e(|\c|+\e^{\frac 1 2}\|w X\|_\infty)]^{1+\beta}|X|_{2,+}
.\label{estR-barR}\ee
\be|P_\g^\c\tilde X|_{\infty,-}\lesssim [\e(|\c|+\e^{\frac 1 2}\|w X\|_\infty)]^{1+\beta}|X|_{\infty,+}
.\label{estR-barRinfty}\ee
\item Given $X_1$ and $X_2$ such that $\e^{\frac 1 2}\|w X_i\|_\infty$ are bounded, then,
\be |P_\g^\c[\tilde X_1-\tilde X_2]|_{2,-}\lesssim  [\e(|\c|+\max_{i=1,2}(\e^{\frac 1 2}\|wX_i\|_\infty))]^{1+\beta}|X_1-X_2|_{2,+}
.\label{estlipR-barR}\ee
\be |P_\g^\c[\tilde X_1-\tilde X_2]|_{\infty,-}\lesssim   [\e(|\c|+\max_{i=1,2}(\e^{\frac 1 2}\|wX_i\|_\infty))]^{1+\beta}|X_1-X_2|_{\infty}
.\label{estlipR-barRinfty}\ee
\item $\Gamma^\pm_\c(\bar X, \bar X)\le \Gamma_\c^\pm( |X|,  |X|) |$ and $|\Gamma_\c^\pm(\bar X_1, \bar X_1)-\tilde\Gamma^\pm_\c(\bar X_2, \bar X_2)|\lesssim \tilde\Gamma^\pm_\c(| X_1|+| X_2|, |X_1-X_2|)$.
\end{enumerate}
\end{proposition}
\begin{proof}

To prove \eqref{estR-barR}, note that
\[\mathbf{1}_{\{\mu_\c +\e\sqrt{\mu_\c}\{\mathscr{Q}+\e^{\frac 1 2} X\}<0 \}}\le\mathbf{1}_{\{\sqrt{\mu_\c}< \e(|\mathscr{Q}|+\e^{\frac 1 2}\| w X\|_\infty) w^{-1} \}}\]
Therefore, by \eqref{tildeR<R}, since $w^{-1}\lesssim \sm_\c^{\beta}\langle v\rangle^{-\beta'}$ for some $0<\beta\ll1$ and $\beta'>0$, and $|\mathscr{Q}|\lesssim |\c|\sm_\c \langle v\rangle^\ell$ for some $\ell>0$, we have 
\begin{multline}|P^\c_\gamma\tilde X|\le2\frac{\mu}{\sqrt{\mu_\c(v)}}\int_{v'\cdot n >0}\dd v'\sqrt{\mu_\c(v')}|v'\cdot n| \mathbf{1}_{\{\sqrt{\mu_\c(v')}< \e(|\mathscr{Q}|+\e^{\frac 1 2}\| w X\|_\infty w^{-1})\}}|X|dv'\\ \le 2\|X|v\cdot n|^{\frac 1 2}\|_{L^2_v}\frac{\mu}{\sqrt{\mu_\c(v)}}\Big(\int_{\R^3}\dd v' w^{-2}(v')
 \mathbf{1}_{\{\sqrt{\mu_\c(v')}< \e(|\mathscr{Q}|+\e^{\frac 1 2}\| w X\|_\infty w^{-1})\}}\mu_\c(v')|v'\cdot n|\Big)^{\frac 1 2}\\\lesssim [\e(|\c|+\e^{\frac 1 2}\| w X\|_\infty)]^{1+\beta}\|X|v\cdot n|^{\frac 1 2}\|_{L^2_v}\frac{\mu}{\sqrt{\mu_\c(v)}},\label{netfluxmod}
\end{multline}  
because $\int_{\R^3}dv' \langle v\rangle^{-2\beta'}(v')|v'\cdot n|\lesssim 1$ by choosing $\beta'>2$.
Therefore \be\int_{\gamma_-}\dd v|P_\gamma \tilde X|^2\lesssim [\e(|\c|+\e^{\frac 1 2}\| w X\|_\infty)]^{2+2\beta}|X|^2_{2,+},\ee
so \eqref{estR-barR} is proven.

We also have 
\begin{multline}|P^\c_\gamma\tilde X| \le \|X\|_{\infty}\frac{\mu}{\sqrt{\mu_\c(v)}}\Big(\int_{\R^3}\dd v' w^{-2}(v')
 \mathbf{1}_{\{\sqrt{\mu_\c(v')}< \e(|\mathscr{Q}|+\a w^{-1})\}}\mu_\c(v')|v'\cdot n|^2\Big)^{\frac 1 2}\\\lesssim  [\e(|\c|+\e^{\frac 1 2}\| w X\|_\infty)]^{1+\beta}\|X\|_{\infty}\frac{\mu}{\sqrt{\mu_\c(v)}},\label{netfluxmod-}
\end{multline}  
from which \eqref{estR-barRinfty} follows.

To prove (2), we observe that, if $\|w X_i\|_\infty\le \a$, by \eqref{liptildeR}         
\be|P^\c_\gamma(\tilde X_1-\tilde X_2)|\le 4\a\frac{\mu}{\sqrt{\mu_\c(v)}}\int_{v'\cdot n >0}\dd v'\sqrt{\mu_\c(v')}v'\cdot n | X_1-X_2| \mathbf{1}_{\sqrt{\mu_\c(v')}< 4\e\a w^{-1}\}}dv'. \ee The rest of the proof is as before. 

Statement (3) follows immediately from \eqref{barR<R} and \eqref{lipbarR}.
\end{proof}
\begin{proposition}\label{722}
Let $u$ be the solution to the incompressible Navier-Stokes equations. Then, if $\e\ll 1$,
\begin{itemize}\item
for any $p>1$
\be\|\P_\c \bar A_\c\|_p\lesssim \e^\infty% {\e|\c|\|D^2 u\|_p}
;\label{Pbarac}\ee
\item for any $p>\frac 4 3$
\be\|\ipc \bar A_\c\|_p\lesssim |\c|;\label{ipcbara}\ee
\item  
\begin{eqnarray}&& \|\nu^{-\frac 1 2}L_\c^{(1)} \bar X\|_p\lesssim |\c|\|X\|_{\frac {3p}{3-p}}\quad \text{ for } p<3,\label{L1baraR}\\
&&\|\nu^{-\frac 1 2}L_\c^{(1)} \bar X\|_p\lesssim |\c|\|wX\|_{\infty}\quad \text{ for } p\ge 3.\label{L1baraR1}\end{eqnarray}
\end{itemize}
\end{proposition}
\begin{proof}
First note that,
by \eqref{pabar}, since $f_1=\sm_\c v_\c\cdot u$ and 
$f_2= \sum_{i,j=1}^3\mathscr{B}_{i,j}\pt_i u_j+L_\c^{-1}\Gamma_\c(f_1,f_1)$, 
we obtain
\begin{multline}\P_\c \bar A_\c= \P_\c 
\Big\{ \bar\chi v_\c\cdot\nabla \phi_\e+\bar \chi\sum_{j_1, j_2,j_3=1}^3 u_{j_2}\pt_{j_1} u_{j_3}v_{j_1} L_\c^{-1}\G_\c(v_{\c, j_2}\sm_\c, v_{\c,j_3} \sm_\c)\\+\bar \chi
\Big[\sm_\c \sum_{j_1, j_2,j_3=1}^3v_{j_1}\mathscr{B}_{j_2.j_3} \pt_{j_1}\pt_{j_2}u_{j_3}
\Big]% {+\e\c\cdot\nabla f_2} 
\Big\}.\label{pbara1}\end{multline}

We remind that from \cite{Gal}, Th. X.6.4, we know that, if $\c\neq 0$, then  $u\in L^p$ for any $p>2$, $D u\in L^p$ for any $p>4/3$ and $D^2 u\in L^p$ for any $p>1$. Therefore, for any $p\ge 1$, $\|uDu\|_p\lesssim 1$. Moreover, for any $\b>0$,
\be\bar\chi \mu_\c^\b\le \exp[-\frac {\b} 2\e^{-2m}]\lesssim\e^\infty,\label{barchimu}\ee
and we obtain that the second term is less than $\e^{\infty}$ in $L^p$-norm, for any $p\ge 1$. From the definition of $\phi_\e$ we have $D\phi_\e\sim \mu^{\frac 1 2}_\c |u||Du|$ and hence also the first term is less than $\e^{\infty}$ in $L^p$-norm, for any $p\ge 1$. 
Finally, since $\|D^2 u\|_p\lesssim 1$ for any $p>1$, the third term is less than $\e^{\infty}$ in $L^p$-norm, for any $p> 1$,
so the first  {item of Proposition \ref{722}} is proved. 

To prove the second item we first observe that, for any $p>1$, $\|\Gamma_\c(\chi f_2+\bar\chi \phi_\e,f_j)\|_p\lesssim 1$. This follows as the estimate of $\|\Gamma_\c(f_1,f_1)\|_p$. Next we need to take care of the term $\e^{-1}L\bar\chi \ipc (v\cdot\nabla f_1)$ entering in $f_2$. Since this is proportional to $D u$ this is bounded in $L^p$ for $p>\frac 4 3$. The diverging factor $\e^{-1}$ is dealt with using \eqref{barchimu}.

To prove third item we remind that $\|u\|_3\lesssim |\c|$ for $|\c|\ll1$ (proof in Appendix) and hence also $\|f_1\|_3\lesssim|\c| $. We use the definition of $L^{(1)}_\c$, the inequalities \eqref{barR<R} and for any $p\ge 1$ and $q^{-1}+q'\/^{-1}=1$,
$\|\nu^{-\frac 1 2}\Gamma_\c(f,g)\|_p\le \|\nu^{\frac 1 2}f\|_{pq}\|g\|_{pq'}$
with $q$ such that $pq=3$ and hence $pq'=\frac{3p}{3-p}$ to conclude.
\end{proof}

\subsection{Iteration}\label{iterazionesub}
The construction of   the solution is obtained as follows: we define 
the sequence $\{R\}_{\ell=0}^\infty$ as: $R_0=0$; $R_{\ell+1}$ is the solution to the linear problem
\be v\cdot \nabla R_{\ell+1} +\e^{-1} L_\c R_{\ell+1}= L_\c^{(1)}\bar R_\ell +\e^{\frac 1 2} \Gamma_\c(\bar R_\ell,\bar R_\ell) +\e^{\frac 1 2} \bar A_\c,\label{Rmodlin}\ee
with boundary conditions
\be R_{\ell+1}=P_\g^\c R_{\ell+1}+\e^{\frac 1 2}  r_{\ell},\label{bcrn}\ee
where 
\be r_{\ell}=\bar r -\e^{-\frac 1 2}P_\g^\c \tilde R_{\ell}.\label{rn}\ee
By denoting $\bar g= L_\c^{(1)} \bar R_\ell +\e^{\frac 1 2} \Gamma_\c(\bar R_\ell,\bar R_\ell) +\e^{\frac 1 2} \bar A_\c$, we are reduced to the linear problem studied in the previous sections.

Remind the definition \eqref{lbrrbr} of $\lbr\,\cdot\,\rbr_{\beta,\beta'}$. 
Since in the rest of this section $\beta$ and $\beta'$ are fixed, we drop the indices.
Let $\mathscr{X}$ be the Banach space of the functions $X(x,v)$ such that $\lbr X\rbr$ is finite.
\begin{theorem} There are $\t<1$ and $c_0\ll1$ such that, 
if $\e\ll 1$ and  $|\c|\le c_0\t$, and
\be \sup_{0\le j\le \ell}\lbr R_j\rbr \le \t,\ee
then 
\be \lbr R_{\ell+1}\rbr<\t. \label{thesi}\ee
Moreover, there is $\l<1$ such that
\be \lbr R_{\ell+1}- R_{\ell}\rbr\le \l \lbr R_{\ell}- R_{\ell-1}\rbr
.\label{thesi2}\ee
Therefore $R_\ell$ converges $\lbr\, \cdot\,\rbr$-strongly to $R\in \mathscr{X}$ which solves \eqref{Rmod}, \eqref{bmod}.
\end{theorem}
\begin{proof}
  By Theorem  {\ref{mainlinth}}, we need to show that, when $g=\e^{\frac 1 2}\G_\c(\bar R_\ell,\bar R_\ell) +L_\c^{(1)}\bar R_\ell +\e^{\frac 1 2}\bar A_\c$ and $r= \bar r + \e^{-\frac 1 2}P_\g^\c \tilde R_\ell$, if  $\e\ll 1$, $|\c|\ll 1$, then $\mathscr{M}(g, r)<\t$.

We need to bound all the term in the right hand side of \eqref{mathscrM}. 
To estimate the norms of $\Gamma_\c(f,h)$ we state the following 
\begin{proposition}
We have the following estimates: let $X\in \mathscr{X}$. Then 
\begin{eqnarray}&& \e^{\frac 1 2} \|\nu^{-\frac 1 2}\Gamma_\c(\bar X,\bar X)\|_2\lesssim 
\lbr X\rbr^2,\label{2XX}\\
&&
\e^{\frac 1 2} \|\nu^{-\frac 1 2} {w}\Gamma_\c(\bar X,\bar X)\|_\infty\lesssim \e^{-\frac 1 2}\lbr X\rbr^2,\label{gammainftyinfty}\\
&&
\e^{\frac 1 2}\| \nu^{-\frac 1 2}\G_\c(\bar X,\bar X)\|_{\frac32}\le
\e^{-\frac 1 2}
\lbr X\rbr^2.\label{3P3+P}\end{eqnarray}
\end{proposition}
\begin{proof}
We make use of the following inequality (see \cite{EGKM2}):
\be \|\nu^{-\frac 1 2}\Gamma^\pm(f,h)\|_{\frac {qp}{q+p}}\lesssim \|\nu^{\frac 1 2}f\|_q\|h\|_p,\label{3-p}\ee
In particular, for $q=3$, $p=3$ we get
\be\|\nu^{-\frac 1 2}\G^\pm_\c(f,h)\|_{\frac32}\le \|\nu^{\frac 1 2}f\|_3\|h\|_3,\label{3-2}\ee
and for $q=3$, $p=6$,
\be\|\nu^{-\frac 1 2}\G^\pm_\c(f,h)\|_2\le \|\nu^{\frac 1 2}f\|_3\|h\|_6.\label{3-6}\ee
We will also use 
\be \|\nu^{-\frac 1 2}\Gamma^\pm(f,h)\|_{2}\lesssim \|f\|_\nu\|h\|_\infty,\label{2-infty}\ee
and 
\be\|\langle v\rangle^{-1}\G^\pm_\c(f,h)\|_\infty\le \|f\|_\infty\|h\|_\infty.\label{infty-infty}\ee
By \eqref{barR<R},
$$|\G_\c^\pm(\bar X, \bar X)|\le \G_\c^\pm(|\bar X|, |\bar X|)\le \G_\c^\pm(|X|, |X|).$$
We split  $|X|\le|\ipc X| +|\P_\c X|$. We have
$$\G^\pm_\c( |X|,|X|)\le \G^\pm_\c(|\ipc X|,|\ipc X|)+\G^\pm_\c(|\P_\c X|,|\P_\c X|)+2\tilde\G^\pm_\c(|\ipc X|,|\P_\c X|),$$
where $\tilde \G^\pm_\c(f,g)= \frac 1 2[\G^\pm_\c(f,g)+\G^\pm_\c(g,f)]$.

Using \eqref{3-6} we get
\be\e^{\frac 1 2}\|\nu^{-\frac 1 2}\G^\pm_\c(|\P_\c X|,|\P_\c X|)\|_2\lesssim (\e^{\frac 1 2}\|\P_\c X\|)_3\|\P_\c X\|_6\le \lbr X\rbr^2.\label{2PP}\ee
Using \eqref{2-infty} we get
\be\e^{\frac 1 2}\|\nu^{-\frac 1 2}\G^\pm_\c(|\ipc X|,|\ipc X|)\|_2\le\e (\e^{\frac 1 2}\|\ipc X\|)_\infty(\e^{-1}\|\ipc X\|_\nu)\le \e\lbr X\rbr^2.\label{2IPIP}\ee
Similarly, 
\be\e^{\frac 1 2}\|\nu^{-\frac 1 2} \tilde\G^\pm_\c(|\ipc X|,|\P_\c R_n|)\|_2\le \e (\e^{\frac 1 2}\|\P_\c X\|)_\infty(\e^{-1}\|\ipc X\|_\nu)\le \e\lbr X\rbr^2.\label{2PIP}\ee
Therefore \eqref{2XX} follows.
Moreover, by \eqref{infty-infty}, \eqref{gammainftyinfty} follows.
\newline
Since
$$ \e^{\frac 1 2}\|\Gamma^\pm(|\ipc X|, |\ipc X|) \|_{\frac32}\le \e^{\frac 1 2}\|\ipc X\|_3^2,$$
and, by interpolation, $\|\ipc X\|_3\lesssim \|\ipc X\|_\nu^{\frac12} \|\ipc X\|_6^{\frac12}$, then
\[ \e^{\frac 1 2}\|\Gamma^\pm(|\ipc X|, |\ipc X|) \|_{\frac 32}\le \e^{\frac 3 2}(\e^{-1}\|\ipc X\|_\nu\|)\|\ipc X\|_6.\]
Since 
\be\|\ipc X\|_6\lesssim \e^{\frac 1 3}\|\e^{-1}\ipc X\|_\nu^{\frac 1 3} \e^{-\frac 1 3}\|\e^{\frac 1 2}\ipc X\|_\infty^{\frac 2 3}\le \lbr X\rbr,\label{IPC6}\ee
we have 
\be  \e^{\frac 1 2}\|\Gamma^\pm(|\ipc X|, |\ipc X|) \|_{\frac32}\le \e^{\frac 3 2}\lbr X\rbr^2.\label{gammaIPIP65}\ee
Moreover, 
\begin{multline} \e^{\frac 1 2}\|\tilde\Gamma^\pm (|\ipc X|, |\P_\c X|) \|_{\frac32}\le 
(\e^{\frac 1 2}\|\P_\c X\|_3) \|\ipc X\|_3\lesssim  \\(\e^{\frac 1 2}\|\P_\c X\|_3)\|\ipc X\|_\nu^{\frac12} \|\ipc X\|_6^{\frac12}
\le \e^{\frac12}\lbr X\rbr^2.\label{gammaPIP65}\end{multline}
Finally 
\begin{multline} \e^{\frac 1 2}\|\Gamma^\pm (|\P_\c X|, |\P_\c X|) \|_{\frac32}\le 
(\e^{\frac 1 2}\|\P_\c X\|_3) \|\P_\c X\|_3\le  \e^{-\frac 1 2}(\e^{\frac 1 2}\|\P_\c X\|_3)^2
\le \e^{-\frac12}\lbr X\rbr^2.\label{gammaPP32}\end{multline}
and \eqref{3P3+P} follows.
\end{proof}
Now we are ready to bound the several terms entering in $\mathscr{M}$.
\begin{proposition} If $|\c|\ll1$ and $\e\ll 1$ then, with $$\Xi_\ell= \sup_{0\le j\le \ell}\lbr R_j\rbr,$$ we have 
\be\mathscr{M}\big(\e^{\frac 1 2} \G_\c(\bar R_\ell, \bar R_\ell) +L_\c^{(1)}\bar R_\ell +\e^{\frac 1 2} \bar A_\c,r_\ell\big) \lesssim \Xi_\ell^4+|\c|^2 \Xi_\ell^2+\e|\c|^2+\e^{\infty}.\label{P10003}\ee\end{proposition}
\begin{proof}
With $g=\e^{\frac 1 2}\G_\c(\bar R_\ell,\bar R_\ell)+ L_\c^{(1)} \bar R_\ell +\e^{\frac 1 2} \bar A_\c$, we have
\begin{multline} \|\nu^{-\frac 1 2} {\ipc} g\|_2^2\le \e \| \nu^{-\frac 1 2} \G_\c(\bar R_\ell,\bar R_\ell)\|_2^2+ \| \nu^{-\frac 1 2} L_\c^{(1)}\bar R_\ell\|_2^2 +\e \| \nu^{-\frac 1 2}  {\ipc}\bar A_\c\|_2^2\\\lesssim \lbr R_\ell\rbr^4+|\c|^2 \lbr R_\ell\rbr^2 +\e |\c|^2,\end{multline}
by using \eqref{2XX}, \eqref{L1baraR}, \eqref{Pbarac} and \eqref{ipcbara}.

The next term in \eqref{mathscrM} is
\begin{multline} \e\|\nu^{-\frac 1 2}(\e^{\frac 1 2}\G_\c(\bar R_\ell,\bar R_\ell)+ L_\c^{(1)} \bar R_\ell +\e^{\frac 1 2} \bar A_\c)\|_{\frac 3 2}^2\le \e\|\nu^{-\frac 1 2}\e^{\frac 1 2}\G_\c(\bar R_\ell,\bar R_\ell)\|_{\frac 3 2}^2+\e\|\nu^{-\frac 1 2} L_\c^{(1)} \bar R_\ell\|_{\frac 3 2}^2 \\+\e\|\nu^{-\frac 1 2} \e^{\frac 1 2}\bar A_\c\|_{\frac 3 2}^2\lesssim \lbr R_\ell\rbr^4+ \e|\c|\lbr R_\ell\rbr^2+ \e^4 |\c|^2,
\end{multline}
by using \eqref{3P3+P}, \eqref{L1baraR}, \eqref{Pbarac} and \eqref{ipcbara}.

Then we have 
\begin{multline} \e^3\Big\|\langle v\rangle^{-1} {w}\big[\e^{\frac 1 2}\G_\c(\bar R_\ell,\bar R_\ell)+ L_\c^{(1)} \bar R_\ell +\e^{\frac 1 2} \bar A_\c\big]\Big\|_{\infty}^2\le \\\e^3\Big\|\langle v\rangle^{-1} {w}\e^{\frac 1 2}\G_\c(\bar R_\ell,\bar R_\ell)\Big\|_{\infty}^2+\e^3\Big\|\langle v\rangle^{-1} {w} L_\c^{(1)} \bar R_\ell\Big\|_{\infty}^2 +\e^3\Big\|\e^{\frac 1 2} \langle v\rangle^{-1} {w}\bar A_\c\big]\Big\|_{\infty}^2\\\lesssim \e^2\lbr R_\ell\rbr^4+ \e^2|\c|\lbr R_\ell\rbr^2+ \e^2 |\c|^2,
\end{multline}
by using \eqref{gammainftyinfty}, \eqref{L1baraR1}, \eqref{Pbarac} and \eqref{ipcbara}.

Since $\P_\c g= {\e^{1/2}}\P_\c \bar A_\c$, the term $\|\P_\c g\|_2^2+ \e^{-2} |\c|^2\|\P_\c g\|_{\frac 6 5}^2$ in \eqref{mathscrM} is bounded by $\e^{\infty}$
%$ {\e|\c|}$ 
using \eqref{Pbarac}. Next we bound 
\be\|z_\g(r_\ell)\|_2^2\lesssim \|z_\g(\bar r)\|_2^2+\e^{-1}\|z_\g(P_\g^\c \tilde R_\ell)\|_2^2.\ee
The first term is bounded by $\e^\infty$ using \eqref{r0mod1}. Moreover, by \eqref{estR-barR}, 
\begin{multline*} (\e^{\frac 1 2-2\sigma}|\c|^{-2+2\varrho}+|\c|^{-1}\e^{-1}) \e^{-1}\|z_\g(P_\g^\c \tilde R_\ell)\|_2^2\\
\le \e^{-2}|\c|^{-1}[\e(|\c|+\e^{\frac 1 2}\|w R_\ell\|_\infty)]^{2(1+\beta)}|R_\ell|_{2,+}^2\\
\le \e^{2\beta}|\c|^{-1}(|\c|+\e^{\frac 1 2}\|w R_\ell\|_\infty)^{2(1+\beta)} |R_\ell|_{2,+}^2\end{multline*}
To bound $|R_\ell|_{2,+}$ we use Lemma \ref{trace_s} and \eqref{Rmodlin} with $\ell$ replaced by $\ell-1$ to obtain
\begin{multline} |R_\ell|^2_{2,+}\lesssim\|\P_\c R_\ell\|^2_6+(\e^{-1}\|\ipc R_\ell\|_              \nu)^2+\e\| \nu^{-\frac 1 2}\G_\c(\bar R_{\ell-1}, \bar R_{\ell-1})\|_2 +\|\nu^{-\frac 1 2}L_\c^{(1)}\bar R_{\ell-1}\|_2^2 \\+\e \|\nu^{-\frac 1 2}\bar A_\c\|_2^2\lesssim\lbr R_\ell\rbr^2+ \Xi_\ell^4 + |\c|^2\Xi_\ell^2+\e |\c|^2,\end{multline}
by using \eqref{2XX}, \eqref{L1baraR}, \eqref{Pbarac} and \eqref{ipcbara}.
Hence, since $\Xi<\t<1$, for $\e\ll1$ we have  
\begin{multline} (\e^{-2\sigma}|\c|^{-2+2\varrho}+|\c|^{-1}\e^{-1}) \e^{-1}\|r_\ell)\|_2^2\\\lesssim \e^\infty+\e^{2\beta}|\c|^{-1}(|\c|+\lbr R_\ell\rbr)^{2(1+\beta)}\{\lbr R_\ell\rbr^2+ \Xi_\ell^4 + |\c|^2\Xi_\ell^2+|\c|^2\}\\\lesssim \Xi_\ell^4 +|\c|^2\Xi_\ell^2+\e|\c|^2+\e^\infty,\end{multline}
\newline
The terms $|r|_{2,-}^2$is treated in a similar way. As for $\e|w r|_\infty$ we proceed as before using \eqref{estR-barRinfty}, \eqref{gammainftyinfty}, \eqref{L1baraR1} and \eqref{Pbarac} and \eqref{ipcbara}.

Collecting the estimates, since $\e<1$, we conclude that 
\be \mathscr{M}\big(\e^{\frac 1 2} \G_\c(\bar R_\ell, \bar R_\ell) +L_\c^{(1)}\bar R_\ell +\e^{\frac 1 2} \bar A_\c,r_\ell\big) \lesssim \Xi_\ell^4+ |\c|^2 \Xi_\ell^2 +\e|\c|^2 +\e^\infty\ee
\end{proof}
Since $\Xi_\ell\le \t$, from \eqref{mainlinest} we obtain
\be \lbr R_{\ell+1}\rbr^2\le \t^2[\t^2+c_0^2\t^2+c_0^2+ \e^\infty\t^{-2}]< \t^2,\ee
provided that $$\t^2+c_0^2\t^2+c_0^2+ \e^\infty\t^{-2}<1.$$ This is verified if $\t\ll1$, $\e\ll1$, $c_0\ll 1$.

The same arguments prove \eqref{thesi2}, by using \eqref{lipbarR} and \eqref{liptildeR}.
The sequence $\{R_\ell\}$ thus converges strongly to $R$ such that $\lbr R\rbr<\t$. It is standard to check that $R$ solves \eqref{Rmod}. Since convergence in $\lbr \,\cdot\,\rbr$ implies pointwise convergence, by \eqref{liptildeR} it follows that $R$ satisfies \eqref{bmod}.

\end{proof}\bigskip
Therefore $F=\mu_\c+\e\mathscr{Q}+\e^{\frac 3 2} R$ solves the problem \eqref{mainpose}, \eqref{mainposb} and hence it is positive by construction. Moreover, it is in $L^\infty$, even if not uniformly bounded in $\e$. We can use Proposition \ref{arkeryd} to conclude that it is also solution to the original problem \eqref{1}, with boundary condition \eqref{bc0} and condition at infinity \eqref{Finfty}. The same estimates also prove uniqueness in the larger space because we can drop the assumption $\beta>0$ which was used before only to deal with terms appearing in the modified problem  \eqref{mainpose}, \eqref{mainposb}.

\appendix
\section{Bounds on the velocity field}\label{appendice}
\numberwithin{equation}{section}
\begin{proposition}
If $|\c|$ is sufficiently small, then the solution to the problem
\begin{eqnarray}
&&U\cdot \nabla U+\nabla p =\Delta U, \quad \nabla\cdot U=0, \quad \text{in } \O^c \\
&&\lim_{|x|\rightarrow \infty }U =\c, \quad U\Big|_{\partial \Omega } =0,
\end{eqnarray}
is such that
\be \|U-\c\|_p\lesssim |\c|, \quad \text{ for any } p\ge 3.\label{A3}\ee
\end{proposition}

\begin{proof}
We first construct $w(x)$ such that $\nabla \cdot w(x)=0,
$ $\lim_{|x|\rightarrow \infty }w(x)=\c$,  and $w(x)|_{\partial \Omega }=0$, with $|w(x)-\c|=0$ for $x$ sufficiently large. In fact (see \cite{Lady}) we can choose
\[ w=\c- \mathop{\rm curl}[\chi({d(x,\pt\O)})(\c_2x_3,\c_3x_1,\c_1x_2)],\]
where $\chi(z)$ is smooth with $\chi(z)=1$ for $x<\frac 1 2$ and $\chi(z)=0$ for $z\ge 1$.
By construction $\nabla\cdot w=0$. Moreover, we have
\begin{multline*}\mathop{\rm curl}[\chi({d(x,\pt\O)})(\c_2x_3,\c_3x_1,\c_1x_2)]= \\
\chi' ({d(x,\pt\O)})\nabla_x d(x,\pt\O) (\c_2x_3,\c_3x_1,\c_1x_2)+
\chi({d(x,\pt\O)})\c$$
=\begin{cases}\c \quad x\in\pt\O,\\
 0 \quad d(x,\pt\O)>1.\end{cases}\end{multline*}
Clearly $w-\c$ is compactly supported and $\|w\|_{W^{s,p}}\lesssim |\c|$ for any $p\ge 1$ and any $s\ge 0$. We then seek for $U=w+v$, with $v$  such that 
\begin{eqnarray}
&&w\cdot \nabla v-\Delta v+\nabla p =-(w+v)\cdot
\nabla w+\Delta w-v\cdot \nabla v \label{A40}\\
&&\lim_{|x|\rightarrow \infty }v =0 \quad
v\big|_{\partial \Omega} =0.\label{A41}
\end{eqnarray}
We construct the approximating sequence solving 
\begin{eqnarray}
&&w\cdot \nabla v^\ell-\Delta v^\ell+\nabla p^\ell =-(w+v^\ell)\cdot
\nabla w+\Delta w-v^{\ell-1}\cdot \nabla v^\ell \label{A4}\\
&&\lim_{|x|\rightarrow \infty }v^\ell =0 \quad
v^\ell\big|_{\partial \Omega =0} =0,
\end{eqnarray}
for $\ell\ge 1$ and $v^0=0$.

\noindent\underline{Step 1}. By energy estimate and weak solution theory, we can show there is a solution $v$ to \eqref{A40}, \eqref{A41}, unique for $|\c|\ll1$, which is the weak limit of $\{v^\ell\}$ and for any $\ell$ 
\[
\|\nabla v^\ell\|_{L^{2}}+\|v^\ell\|_{L^{6}}\lesssim|\c| .
\]

\noindent\underline{Step 2}. We now show that  $v\in L^{3}$  and $\|v\|_{L^3}\lesssim |\c|$. Using $\pt_j\c=0$, $\nabla\cdot v^\ell=0$ and $\nabla\cdot w=0$, we write the $i$-th component of \eqref{A4} as
\be
\sum_{j=1}^3[\c_j \pt_j v_i^\ell-\pt_j^2 v_i^\ell]+\pt_i p^\ell =
-\sum_{j=1}^3\pt_j\big[
-\pt_j w_i  + (w_j-\c_j)v_i^\ell +v_j^\ell(w_i-\c_i) +w_jw_i+ v_j^{\ell-1} v_i^\ell\big]\\
.
\ee
In Fourier space, we have
(using the Leray Projector $\Pi$,  and $k\cdot \hat v(k)=0$): 
\[\widehat{\pt_m v_i^{\ell}}=\sum_{j=1}^3\frac{k_m k_j}{|k|^{2}+i\c\cdot k}\hat \Pi\mathcal{F}\big\{
-\pt_j w_i + (w_j-\c_j)v_i^\ell +v_j^\ell(w_i-\c_i) +w_jw_i+ v_j^{\ell-1} v_i^\ell \big\}.
\]
We have 
\[
\Big|\pt_k^l\frac{k_m k_j}{|k|^{2}+i\c\cdot k}\hat\Pi\Big|\leq \frac{1}{|k|^l},
\]
independent of $\c$. Hence we can use  the Mihlin-Hormander theorem. Therefore, by Sobolev embedding in 3D ($
W^{1,\frac 3 2}\subset L^{3}$) and the compact support of $w-\c$, we obtain  
\begin{eqnarray*}
\|v^\ell\|_{L^{3}} &\leq &\sup_{m}\|\pt_m v^\ell\|_{\frac 3 2}\\&\le&\big\|-\pt_j w_i  + (w_j-\c_j)v^\ell_i +v_j^\ell(w_i-\c_i) +w_jw_i+ v_j^{\ell-1} v_i^\ell\big\|_{L^{\frac 3 2}} \\
&\lesssim &|\c|(1+\|v^\ell\|_6)+\||v^{\ell-1}|\/|v^\ell|\|_{L^{\frac 3 2}}\lesssim|\c|+\|v^\ell\|_{L^3 }\|v^{\ell-1}\|_{L^3 }.
\end{eqnarray*}
Therefore, if we assume the recurrence hypothesis $\displaystyle{\sup_{0\le m\le \ell-1}\|v^{m}\|_{L^3}\le C |\c|}$,  by choosing  $|\c|\ll 1$ we obtain 
\[
\|v^\ell\|_{L^{3}}\le C |\c|,
\]
 and the limit satisfies $\|v\|_{L^3}\le C |\c|$.

\noindent \underline{Step 3}. By differentiating the equation, from the energy inequality for the derivative we obtain $\| D v\|_6\lesssim |\c|$ and hence $\|v\|_\infty\lesssim |\c|$. By interpolation we conclude \eqref{A3}.
\end{proof}
{\bf Acknowledgments} 
Yan Guo's research is supported in part by NSF grant 1611695,  Chinese NSF grant 10828103, as well as a Simon Fellowship. R. Marra's research is partially supported by MIUR-Prin.
We are very grateful for extensive and constructive comments from the referees, which help us to improve the presentation of the paper. We thank Junhwa Jung for pointing out some omissions.


\begin{thebibliography}{99}
\bibitem{AN} Arkeryd L.; Nouri  A.: \textit{On a Taylor-Couette Type Bifurcation for the Stationary Nonlinear Boltzmann Equation} Jour.  Stat. Phys., {\bf124}, pp. 401--443 (2006)
\bibitem{Bab} Babenko, K.I.: \textit{On Stationary Solutions of the Problem of Flow Past a Body of a Viscous Incompressible Fluid}, Mat. Sb., {\bf 91} (133), 3-27 (1973); English Transl.:
Math. SSSR Sbornik, {\bf 20} 1973, 1-25 (1973)
\bibitem{BGL89} C. Bardos; F. Golse; D. Levermore: \textit{Sur les limites
asymptotiques 41Acad. Sci. Paris S\'er. I Math.}, \textbf{309}, 727--732
(1989)
\bibitem{BU} C. Bardos; S. Ukai: \textit{The classical incompressible
Navier-Stokes limit of the Boltzmann equation,} Math. Mod. Meth. Appl. S.,
vol.1, p 235 (1991)
\bibitem{Bob}Bobylev A. V.; Mossberg E.: \textit{On Some Properties of Linear and Linearized Boltzmann Collision Oerators for Hard Spheres} Kinetic and Related Models, {\bf 1}, pp. 521-555 (2008)
\bibitem{Ce} C. Cercignani, Mathematical methods in Kinetic Theory, Plenum Press, New York 1969
\bibitem{CIP} C. Cercignani; R. Illner; M. Pulvirenti: \textit{The
mathematical theory of dilute gases,} Springer-Verlag (1994)
\bibitem{DEL} A. De Masi; R. Esposito; J. L. Lebowitz: \textit{
Incompressible Navier-Stokes and Euler Limits of the Boltzmann Equation}, {
Comm. Pure and Appl. Math.}, \textbf{42}, 1189--1214 (1989)
\bibitem{EGKM} R. Esposito; Y. Guo; C. Kim; R. Marra: \textit{Non-Isothermal
Boundary in the Boltzmann Theory and Fourier Law}, Comm. Math. Phys. \textbf{
323}, 177--239 (2013) 
\bibitem{EGKM2} R. Esposito; Y. Guo; C. Kim; R. Marra: \textit{Stationary solutions to the Boltzmann equation in the Hydrodynamic limit}, Annals of PDE, {\bf 4} pp. 1--119 (2018), DOI 10.1007/s40818-017-0037-5; ArXiv, 1502.05324v3 (2016)
\bibitem{EV} L. C. Evans: Partial Differential Equations, Graduate Studies
in Mathematics, vol. 19 American Mathematical Society (1998)
\bibitem{Fin} Finn, R.: \textit{Estimates at Infinity for Stationary Solutions of the Navier-Stokes Equations}, Bull. Math. Soc. Sci. Math. Phys. R. P. Roumaine, {\bf 3} (51) 387-418 (1959)
\bibitem{Gal} Galdi G. P.: \textit{An Introduction to the Mathematical Theory of the Navier-Stokes Equations} Springer (2011)
\bibitem{GSR} F. Golse; L. Saint-Raymond: \textit{The Navier-Stokes limit of
the Boltzmann equation for bounded collision kernels}, Invent. Math. \textbf{
155}, 81--161 (2004) 
\bibitem{Guo06} Y. Guo: \textit{Boltzmann diffusive limit beyond the
Navier-Stokes approximation.} {Comm. Pure and Appl. Math.} \textbf{59},
626--687 (2006)
\bibitem{GJ} Y. Guo; J. Jang: \textit{Global Hilbert expansion for the
Vlasov-Poisson-Boltzmann system.} Comm. Math. Phys. \textbf{299}, 469--501
(2010)
\bibitem{GJJ} Y. Guo; J. Jang; N. Jiang: \textit{Acoustic limit for the
Boltzmann equation in optimal scaling.} Comm. Pure Appl. Math. \textbf{63},
337--361 (2010)
\bibitem{Hor} H\"ormander L.: \textit{Estimates for translation invariant operators in $L^p$ spaces}. Acta Math. {\bf 104}, 93--140 (1960).
\bibitem{Lady} O. Ladyzhenskaya: Mathematical Theory of Viscous Incompressible Flow, 2th. ed., Gordon and Breach (1969) 
\bibitem{Lam} Lamb, H.: Hydrodynamics (Sixth ed.).  Dover Publications. (1945)
\bibitem{Leo} Leoni G.: A First Course in Sobolev Spaces, AMS Grad. Stud. in Math. (2009)
\bibitem{Ler} Leray, J.: \textit{\'Etude de Diverses \'Equations Integrales non Lin\'eaires et de Quelques Probl\`emes que Pose l' Hydrodynamique}, J. Math. Pures Appl., {\bf 12}, 1-82 (1933)
\bibitem{LM} Lions, P.-L. ; Masmoudi, N. : \textit{From the Boltzmann equations
to the equations of incompressible fluid mechanics. I, II.} Arch. Ration.
Mech. Anal. \textbf{158}, no. 3, 173--193, 195--211 (2001)
\bibitem{Masl} N. Maslova: \textrm{Nonlinear Evolution Equations}, World Scientific (1993)
\bibitem{MS}Masmoudi,  N. ; Saint-Raymond, L. : \textit{From the Boltzmann
equation to the Stokes-Fourier system in a bounded domain.} Comm. Pure Appl.
Math. \textbf{56}, no. 9, 1263--1293 (2003)

\bibitem{Mih}  Mihlin S. G.: \textit{On the multipliers of Fourier integrals} (Russian). Dokl. Akad. Nauk SSSR {\bf 12}, 143-155 (1957).
\bibitem{OS} Oseen, C. W.: \"Uber die Stokes'sche formel, und \"uber eine verwandte Aufgabe in der Hydrodynamik, Arkiv f\"or matematik, astronomi och fysik, vi (29) (1910)
\bibitem{SR} Saint-Raymond, L.: Hydrodynamic Limits of the Boltzmann Equation, Springer Lecture Notes in Mathematics, (2009)
\bibitem{Th} Thomson, W.: On ship waves, Institution of Mechanical Engineers, Proceedings, {\bf 38} : 409--434 (1887) 
\bibitem{UA1} Ukai  S.; Asano  K.: \textit{Steady solutions of the Boltzmann equation for a gas flow past an obstacle. I. Existence}. Arch. Rational Mech. Anal. {\bf 84}, 249--291 (1983).
\bibitem{UA2}Ukai  S.; Asano  K.:\textit{ Steady solutions of the Boltzmann equation for a gas flow past an obstacle. II. Stability}. Publ. Res. Inst. Math. Sci. {\bf 22}, 1035--1062  (1986)
\bibitem{UYZ}Ukai S.; Yang  T.; Zhao  H.: \textit{Stationary solutions to the exterior problems for the Boltzmann equation. I. Existence}. Discrete Contin. Dyn. Syst. {\bf 23}, 495--520  (2009)
\bibitem{UYZ1}Ukai S.; Yang  T.; Zhao  H.: \textit{Exterior Problem for the Boltzmann equationwith Temperature Difference}. Commun. in Pure and Appl. Analysis, {\bf 8}, pp 473-491 (2009)

\bibitem{GW} L. Wu; Y. Guo: \textit{Geometric Correction for Diffusive
Expansion of Steady Neutron Transport Equation}, arXiv:1404.2583, 
in Comm. Math. Phys.  {\bf 336} 1473--1533 (2015)
\end{thebibliography}
\end{document}